\documentclass[a4paper,UKenglish,cleveref, autoref, thm-restate]{lipics-v2021}
\pdfoutput=1 %uncomment to ensure pdflatex processing (mandatatory e.g. to submit to arXiv)
\hideLIPIcs  %uncomment to remove references to LIPIcs series (logo, DOI, ...), e.g. when preparing a pre-final version to be uploaded to arXiv or another public repository

\title{Tight Bounds for Feedback Vertex Set Parameterized by Clique-width} %TODO Please add

\author{Narek Bojikian}{Humboldt-Universität zu Berlin, Germany}{bojikian@hu-berlin.de}{https://orcid.org/0000-0003-1072-4873}{}

\author{Stefan Kratsch}{Humboldt-Universität zu Berlin, Germany}{stefan.kratsch@hu-berlin.de}{https://orcid.org/0000-0002-0193-7239}{}

\authorrunning{N.\ Bojikian, and S.\ Kratsch}

\Copyright{Narek Bojikian and Stefan Kratsch} %TODO mandatory, please use full first names. LIPIcs license is "CC-BY";  http://creativecommons.org/licenses/by/3.0/

\ccsdesc[500]{Theory of computation~Parameterized complexity and exact algorithms}

\keywords{Feedback Vertex Set, Treewidth, Clique-width} %TODO mandatory; please add comma-separated list of keywords

\category{} %optional, e.g. invited paper

\relatedversion{} %optional, e.g. full version hosted on arXiv, HAL, or other respository/website
\nolinenumbers %uncomment to disable line numbering

\EventEditors{John Q. Open and Joan R. Access}
\EventNoEds{2}
\EventLongTitle{42nd Conference on Very Important Topics (CVIT 2016)}
\EventShortTitle{CVIT 2016}
\EventAcronym{CVIT}
\EventYear{2016}
\EventDate{December 24--27, 2016}
\EventLocation{Little Whinging, United Kingdom}
\EventLogo{}
\SeriesVolume{42}
\ArticleNo{23}
\usepackage{macros}
\usepackage{drawings}

\allowdisplaybreaks

\begin{document}

\maketitle

\begin{abstract}
    We introduce a new notion of acyclicity representation in labeled graphs, and present three applications thereof. Our main result is an algorithm that, given a graph $G$ and a $k$-clique expression of $G$, in time $\Oh(6^kn^c)$ counts modulo $2$ the number of feedback vertex sets of $G$ of each size.
    We achieve this through an involved subroutine for merging partial solutions at union nodes in the expression.
    In the usual way this results in a one-sided error Monte-Carlo algorithm for solving the decision problem in the same time. We complement these by a matching lower bound under the Strong Exponential-Time Hypothesis (SETH). This closes an open question that appeared multiple times in the literature [ESA 23, ICALP 24, IPEC 25].
    
    We also present an algorithm that, given a graph $G$ and a tree decomposition of width $k$ of $G$, in time $\Oh(3^kn^c)$ counts modulo $2$ the number of feedback vertex sets of $G$ of each size. This matches the known SETH-tight bound for the decision version, which was obtained using the celebrated cut-and-count technique [FOCS 11, TALG 22].
    Unlike other applications of cut-and-count, which use the isolation lemma to reduce a decision problem to counting solutions modulo $2$, this bound was obtained via counting other objects, leaving the complexity of counting solutions modulo $2$ open.

    Finally, we present a one-sided error Monte-Carlo algorithm that, given a graph $G$ and a $k$-clique expression of $G$, in time $\Oh(18^kn^c)$ decides the existence of a connected feedback vertex set of size $\target$ in $G$. We provide a matching lower bound under SETH.
\end{abstract}

\newpage
\setcounter{page}{1}
\section{Introduction}

Given a graph $G$ and an integer $k$, the \Fvsp problem asks whether there is a set of at most $k$ vertices that intersects all cycles in $G$. This well-known NP-hard problem has been studied intensely in a variety of algorithmic paradigms. In particular, it is an important benchmark problem in the study of optimal (dynamic programming) algorithms on graph decompositions. Continuing this line of work, we study the parameterized complexity of \Fvsp, and its connected variant, where we require the solution set to induce a connected subgraph, relative to the treewidth and clique-width of the graph.

Parameterized complexity extends classical complexity by allowing complexity bounds to depend not only on the input size but also on one or more parameters, which may quantify various input properties. In particular, one aims to find so-called \emph{fixed-parameter tractable} algorithms whose running time is polynomial in the input size but (usually) exponential in the parameter value, i.e., time $f(k)n^c$ for size $n$ and parameter $k$. In recent years, there has been much success in obtaining fine-grained bounds for the dependency $f(k)$ on the parameter that is possible for such algorithms, which includes (almost) matching lower bounds conditioned on the (Strong) Exponential-Time Hypothesis (ETH resp.\ SETH).\footnote{ETH posits that there is $c>0$ such that no algorithm solves \textsc{$3$-SAT} in time $\Oh(2^{cn})$ where $n$ is the number of variables. SETH posits that for each $c<1$ there is a clause size $q$ such that no algorithm solves \textsc{$q$-SAT} in time $\Oh(2^{cn})$.}

For graph problems, the most studied parameter is the \emph{treewidth} of the input graph: Intuitively, it is equal to the smallest value $k$ such that the graph can be completely decomposed along non-crossing separators of size at most $k$. A similar parameter, aimed at dense graphs, is \emph{clique-width}: Roughly, it is the smallest number $k$ of labels, such that the graph can be recursively constructed using creation of a single labeled vertex, relabeling, adding edges between label classes, and disjoint union.

For many fundamental problems, the complexity relative to treewidth is quite well understood.
This was initiated by Lokshtanov et al.~\cite{DBLP:journals/talg/LokshtanovMS18}, who showed for many well-studied problems, such as \ISp, \DSp, and \MCp parameterized by treewidth that known dynamic programming (DP) algorithms with time $\Oh(\alpha^kn^c)$ are optimal under SETH, for problem-specific $\alpha$. For many other problems, however, including connectivity problems like \Cvcp and \Fvsp, it remained open whether the known algorithms with time $\Oh(\alpha^{k\log k}n^c)$ were optimal (by showing matching lower bounds), or whether they could be significantly improved upon. The crux was that the natural approaches for DP would have to maintain connectivity patterns of partial solutions, which leads to a $k^k$ factor in number of solutions and (hence) time.
This was resolved by Cygan et al.~\cite{DBLP:journals/talg/CyganNPPRW22} by establishing $\Oh(\alpha^kn^c)$ time algorithms and matching lower bounds modulo SETH for most of these problems. Their \cnc technique has two key ingredients: (1) The isolation lemma probabilistically reduces to the special case where we have a unique weighted solution of minimum weight (if one exists), so that we can count (weighted) solutions modulo $2$ to decide existence. (2) The main insight is that to count modulo $2$ the connected solutions, one can instead count modulo $2$ the number of ways that relaxed (possibly disconnected) solutions can be cut, because all disconnected solutions will cancel out modulo $2$. Notably, for \Fvsp this idea does not seem to apply directly, so they instead count the number of solutions for each specified number of vertices, edges, and connected components. This suffices to decide existence of a solution but does not include an algorithm for counting solutions modulo $2$.

After these two pioneering works, there was much activity for establishing SETH-tight bounds for a variety of problems under different parameters such as cutwidth~\cite{DBLP:conf/stacs/BojikianCHK23,DBLP:conf/stacs/GroenlandMNS22,DBLP:journals/tcs/JansenN19,DBLP:journals/jgaa/GeffenJKM20}, pathwidth~\cite{DBLP:journals/corr/bojikianfghs25,DBLP:conf/soda/CurticapeanLN18,DBLP:journals/jacm/CyganKN18}, and modular treewidth~\cite{DBLP:conf/wg/HegerfeldK23}. The first to study connectivity problems parameterized by clique-width were Bergougnoux and Kant\'e~\cite{DBLP:journals/tcs/BergougnouxK19}. They provided deterministic single-exponential time algorithms for many connectivity problems including the \Fvsp problem. Since they use the rank-based approach of Bodlaender et al.~\cite{DBLP:journals/iandc/BodlaenderCKN15}, they do not obtain SETH-tight bounds. Instead, there is an additional factor of $\omega$, the matrix multiplication exponent, in the exponent of the running times, which leaves a gap even if $\omega=2$. The first SETH-tight algorithms for connectivity problems parameterized by clique-width were introduced by Hegerfeld and Kratsch~\cite{DBLP:conf/wg/HegerfeldK23}, who provided tight bounds for both the \Cvcp and the \Cdsp problems using the \cnc technique. They left open to get such tight bounds also for \Stp, \Coctp, and \Fvsp, but pointed out obstructions with doing so via \cnc: For the first two, there remained a gap between a \cnc-based algorithm and the seemingly best possible lower bound. For the latter, already Bergougnoux and Kanté had pointed out that the approach of counting solutions separately for each number of vertices, edges, and connected components would incur a factor of $n^k$ when working with clique-width. Two recent works of Bojikian and Kratsch~\cite{DBLP:conf/icalp/BojikianK24, DBLP:journals/corr/bojikiank24} closed the gap for \Stp and \Cdsp by finding faster algorithms based on a new approach via connectivity patterns called ``isolating a representative'' and avoiding \cnc. Only \Fvsp parameterized by clique-width remained open.

\subparagraph{Our contribution.} 
We present a dynamic programming algorithm that counts (modulo $2$) the number of feedback vertex sets of a fixed size $\target$ in a graph in time $\Oh(6^{k}n^c)$, given that this graph is provided together with a $k$-clique expression. 
We achieve this running time by introducing a new representation of acyclicity in graphs, which enables us to bypass counting edges induced by a partial solution. We prove the following theorem:

\begin{theorem}\label{theo:count-cw}
    There exists an algorithm that given a graph $G$ together with a $k$-clique expression of $G$ for some integer value $k$ and a positive integer $\target$, counts (modulo $2$) the number of feedback vertex sets of size $\target$ in $G$ in time $\Oh(6^kn^c)$.
\end{theorem}

Our algorithm actually counts weighted feedback vertex sets, for any polynomially bounded vertex-weight function. This allows us to use the isolation lemma to reduce the decision version of the problem to the counting modulo $2$ version with high probability. As a result we get the following:

\begin{theorem}\label{theo:ub}
There exists a one-sided error Monte-Carlo algorithm that given a graph $G$ together with a $k$-clique expression of $G$ for some value $k\in\mathbb{N}$, and a positive integer $\target\in\mathbb{N}$, decides whether $G$ contains a feedback vertex set of size $\target$ in time $\Oh(6^kn^c)$. Errors are limited to false negatives, and the correct answer is output with probability at least $1/2$.
\end{theorem}

We also provide a lower bound that excludes algorithms with running time $\Oh\big((6-\varepsilon)^{k}n^c\big)$ for \Fvsp assuming SETH. As is the case for similar lower bounds for other problems parameterized by clique-width and treewidth (with its linear variant pathwidth), our lower bound holds even when parameterized by the linear clique-width~\cite{DBLP:journals/tcs/AdlerK15,DBLP:journals/tcs/GurskiW05,DBLP:journals/dam/HeggernesMP12}; where each union operation requires that one of its operands is a single vertex.

\begin{theorem}\label{theo:lower-bound}
    Assuming SETH, the \Fvsp problem cannot be solved in time $\Oh\big((6-\varepsilon)^kn^c\big)$ for any $\varepsilon > 0$, even when the input graph $G$ is provided with a linear $k$-expression of $G$.
\end{theorem}

We show that our representation technique also yields an independent treewidth-based algorithm: by restricting the indices of the dynamic programming tables to the states realized by a tree decomposition, one obtains an algorithm that runs in time $\Oh(3^{\tw}n^c)$ on a graph $G$ of treewidth $\tw$, and counts (modulo $2$) the number of feedback vertex sets of size $\target$ in $G$. Again this is SETH-tight, since an algorithm with running time $\Oh((3-\varepsilon)^{\tw}n^c)$ would imply---using the isolation lemma---an algorithm for the decision version with the same running time, which contradicts SETH~\cite{DBLP:journals/talg/CyganNPPRW22}.

\begin{theorem}\label{theo:tw-count}
    There exists an algorithm that given a graph $G$ together with a tree decomposition of $G$ of width $\tw$ and a positive integer $\target$, counts (modulo $2$) the number of feedback vertex sets of size $\target$ in $G$ in time $\Oh(3^{\tw}n^c)$.
\end{theorem}

We note that, to the best of our knowledge, these are the first SETH-tight algorithms that correctly count (modulo $2$) the number of feedback vertex sets when parameterized by a structural parameter, as other known methods either count different objects (cut and count~\cite{DBLP:journals/talg/CyganNPPRW22}), or are based on the ``Gaussian elimination'' technique or on the ``squared determinant'' technique, and hence, are not tight under SETH~\cite{DBLP:journals/iandc/BodlaenderCKN15,DBLP:conf/stacs/BergougnouxKN23}.

As an additional application of our representation technique, we also show that acyclicity representation can be combined with the ``isolating a representative'' technique of Bojikian and Kratsch~\cite{DBLP:conf/icalp/BojikianK24} resulting in an algorithm for the \Cfvsp problem with running time $\Oh(18^{\cw}n^c)$ proving the following theorem:

\begin{theorem}\label{theo:cfvs-ub}
There exists a one-sided error Monte-Carlo algorithm that given a graph $G$ together with a $k$-clique expression of $G$ for some value $k\in\mathbb{N}$, and a positive integer $\target\in\mathbb{N}$, decides whether $G$ contains a connected feedback vertex set of size $\target$ in time $\Oh(18^kn^c)$. Errors are limited to false negatives, and the correct answer is output with probability at least $1/2$.
\end{theorem}

We also prove the tightness of our running time, by ruling out algorithms with running time $\Oh((18-\varepsilon)^{\lcw}n^c)$ under SETH, for all $\varepsilon > 0$, where $\lcw$ denotes the linear clique-width of the graph. 
To the best of our knowledge, this is the largest known single-exponential lower bound for a natural (non-generalized) problem parameterized by a structural parameter. This poses a challenge in the construction of the lower bound.

\begin{theorem}\label{theo:cfvs-lb}
    Assuming SETH, the \Cfvsp problem cannot be solved in time $\Oh\big((18-\varepsilon)^kn^c\big)$ for any $\varepsilon > 0$, even when the input graph $G$ is provided with a linear $k$-expression of $G$.
\end{theorem}

In doing so, we close the gap for two additional connectivity problems relative to clique-width. This leaves only one problem unresolved among those whose tight complexity relative to treewidth was determined in the original work of Cygan et al.~\cite{DBLP:journals/talg/CyganNPPRW22}, which introduced the \cnc technique, namely the \textsc{Exact $k$-Leaf Spanning Tree} problem. However, this problem is at least $W[1]$-hard when parameterized by clique-width, since it generalizes the $\textsc{Hamiltonian Path}$ problem by setting $k=2$~\cite{DBLP:journals/siamcomp/FominGLS10}.

\subparagraph{Organization.} 
In \cref{sec:techrev} we provide an overview of our techniques.
In \cref{sec:preliminaries} we provide preliminaries and some notation. We present the main algorithm of this work in \cref{sec:ub} proving \cref{theo:count-cw} and \cref{theo:ub}. We present the lower bound in \cref{sec:lb}. We prove \cref{theo:tw-count} in \cref{sec:tw}. Finally, we prove the \cref{theo:cfvs-ub} and \cref{theo:cfvs-lb} in \cref{sec:cfvs} and conclude with final remarks in \cref{sec:conclusion}.

\section{Technique overview.}\label{sec:techrev}

\subsection{Upper bounds}

As is typical with structural parameters, we assume that the input graph $G$ is provided together with a $k$-clique expression $\mu$ of $G$. The recursive definition of clique-expressions (see \cref{sec:preliminaries}) induces a syntax tree $\syntaxtree$ of $\mu$ in a natural way, where each subtree $\syntaxtree_x$ rooted at a node $x$ of $\syntaxtree$ corresponds to a subexpression $\mu_x$ of $\mu$.
Since each subexpression $\mu_x$ defines a subgraph $G_x$ of $G$,
we build our algorithms as bottom-up dynamic programming schemes over $\syntaxtree$, where for each node $x$, we keep tables $T_x$ that count (modulo $2$) some notion of partial solutions in the graph $G_x$.

Similar to other algorithms for \Fvsp~\cite{DBLP:journals/tcs/BergougnouxK19,DBLP:journals/iandc/BodlaenderCKN15, DBLP:journals/talg/CyganNPPRW22}, our algorithms count rather induced forests in the graph, which are in one-to-one correspondence with feedback vertex sets. We first show that this can be done by representing labeled graphs as multisets of vectors, that count the number of times a label appears in each connected component of the induced forest. We call these multisets \emph{patterns}. 

However, the total number of different patterns is super-exponential in $n$, and hence, is too large to be used in an efficient dynamic programming algorithm. Therefore, we aim to reduce the families of patterns into a simpler representation thereof, introducing the notion of acyclicity representation. First, we add an isolated vertex $v_0$ of a new label $0$ to the graph. This vertex will allow for a cleaner ``reduced'' representation of patterns. While the vertex $v_0$ is isolated in the original graph, we will use some ``rerouting technique'' in the representation, creating graphs where $v_0$ is connected to other vertices. We call labeled graphs with this added vertex $v_0$ \emph{extended labeled graphs}. Since the vertex $v_0$ is the only vertex having label $0$, we assume that each pattern contains a unique vector $z$ with $z_0 = 1$ and $x_0 = 0$ for all other vectors $x$ in the pattern. we call $z$ the zero vector of the pattern.

Essentially, we prove first that, in order to represent acyclicity in a labeled graph, it suffices to upper bound the multiplicity of each label by $2$, both in each connected component and globally, eventually by removing isolated vertices, whose labels appear more than twice in the graph. This results in multisets over vectors, where each index is upper bounded by $2$, and where each vector has multiplicity at most two in the pattern.
See \cref{fig:technical-pats} for an example of an extended labeled forest and its corresponding pattern.

\begin{figure}[ht]
    \centering%
    \begin{subfigure}[b]{.3\textwidth}
        \centering
        \includegraphics[width=.7\textwidth]{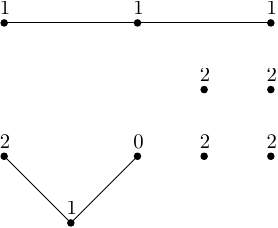}
        \caption{A $2$-labeled forest $F$}
    \end{subfigure}
    \hspace{1cm}%
    \begin{subfigure}[b]{.35\textwidth}
        \centering
        $\left\langle
        \begin{pmatrix}
            0 \\
            2 \\
            0
        \end{pmatrix},
        \begin{pmatrix}
            1 \\
            1 \\
            1
        \end{pmatrix},
        \begin{pmatrix}
            0 \\
            0 \\
            1
        \end{pmatrix},
        \begin{pmatrix}
            0 \\
            0 \\
            1
        \end{pmatrix}
        \right\rangle
        $
        \vspace{1em}
        \caption{The pattern $p:=\pat_F$}
    \end{subfigure}

    \caption{A $k$-labeled forest $F$ ($k=2$) and its corresponding pattern. Note that we upper bound the multiplicity of each label in each connected component by $2$, and the multiplicity of each connected component by $2$ as well. Removing additional unit-vectors happens in a later stage.}
    \label{fig:technical-pats}
\end{figure}

In the second step, we present a compact family of patterns $\CSP$ of size $6^k$. We show that one can replace any family of partial solutions in $G_x$, for a node $x$ of $\syntaxtree$, with a subset of $\CSP$, such that for any extension of a partial solution to a solution in the whole graph, this extension is compatible with an odd number of partial solutions from the original family, if and only if it is compatible with an odd number of patterns in the latter family.

More formally, the family $\CSP$ consists of all patterns $p$ where the only non-unit vector of $p$ is the zero vector. The size of the family follows then, since for each label $i$, we can assign one of the six states: $\stnone$ if it does not appear in the pattern, $\stdisc$ if it appears only once in a unit vector, $\stddisc$ if it appears twice in unit vectors, $\stconn$ if it appears only once in the zero vector, $\stplus$ if it appears once in a unit vector and once in the zero vector, and $\stdconn$ if it appears twice in the zero vector. This builds a bijection between the patterns in $\CSP$ and the set of mappings that assign to each label in $[k]$ one of these six states.

In order to represent each family of patterns with a subset of $\CSP$, we introduce a count-preserving replacement rule: First, we show for a pattern $p$ and three vectors $x,y,z$, such that $x$ and $y+z$ are both in $p$, and for $p_0 = p\setminus \{x,y+z\}$, that the three patterns $p_0\cup\langle x+z, y\rangle$,$p_0\cup\langle x+y, z\rangle$,$p_0\cup\langle x+z+y\rangle$ together represent $p$.
Hence, our reduction rule replaces a pattern $p$ with these three patterns, by choosing $x$ to be the zero vector of $p$, and $y+z$ to be any other non-unit vector different from $x$ in $p$. This results in three patterns, each having a smaller total sum of non-unit vectors different from the zero vector than the same sum in $p$. Therefore, by applying this rule a finite number of times we get a family of patterns $R$ that represents $p$, such that for each pattern $r \in R$, each vector different from the zero vector of $r$ is a unit vector. This implies that $R \subseteq \CSP$. See \cref{fig:reduce} for an example of this reduction rule.

\begin{figure}[ht]
    \centering
    \[
    \left\langle
    \begin{pmatrix}
        1 \\
        1 \\
        0
    \end{pmatrix},
    \begin{pmatrix}
        0 \\
        2 \\
        1
    \end{pmatrix}
    \right\rangle
    \xrightarrow
        [
    y = \begin{pmatrix}
        0 \\
        2 \\
        0 
    \end{pmatrix}\ 
    z = \begin{pmatrix}
        0 \\
        0 \\
        1
    \end{pmatrix}
    ]
    {
    x = \begin{pmatrix}
        1 \\
        1 \\
        0
    \end{pmatrix}}
    \left\langle
        \begin{pmatrix}
            1 \\
            1\\
            1
        \end{pmatrix},
        \begin{pmatrix}
            0 \\
            2 \\
            0
        \end{pmatrix}
    \right\rangle,
    \left\langle
        \begin{pmatrix}
            1 \\
            2 \\
            0
        \end{pmatrix},
        \begin{pmatrix}
            0 \\
            0 \\
            1
        \end{pmatrix}
        \right\rangle,
        \left\langle
        \begin{pmatrix}
            1 \\
            2 \\
            1
        \end{pmatrix}
        \right\rangle
        \]
    \caption{The three patterns resulting from applying the reduction rule to the pattern $p$. We choose $x$ to be the zero vector as specified by the reduction rule. Note that the total sum of non-unit vectors different from the zero vector decreases from $3$ in $p$ to $2$, $0$ and $0$ in the three resulting patterns, respectively. Note also that we keep the upper bound $2$ on the sums in each index.}
    \label{fig:reduce}
\end{figure}

This count-preserving replacement argument allows us to restrict the dynamic programming tables to the family $\CSP$ only, preserving the parity of the number of extensions of partial solutions in $G_x$ to solutions in the whole graph $G$.
However, a bottleneck of this approach is to process union nodes in the clique expression efficiently, combining all pairs of partial solutions in the unified graphs. While this can be trivially done in time $\Oh(36^kn^c)$, we show that this can be done in (optimal) time $\Oh(6^kn^c)$ by an involved convolution operation.

In general, fast convolution techniques have proven essential to process a join node in a tree decomposition or a union node in a clique expression efficiently. Among others, covering product, fast subset convolution, fast lattice convolution over power lattices, and multidimensional fast fourier transformation have been used to process such nodes more efficiently. Van Rooij~\cite{DBLP:conf/birthday/Rooij20} combined some of these techniques to develop a novel convolution for the $(\sigma,\rho)$-\DSp problems, introducing the ``Count and Filter'' technique, inspired by the infamous fast subset convolution of Björklund et al.~\cite{DBLP:conf/stoc/BjorklundHKK07}. The author states that combinations of these convolution techniques cover a vast spectrum of convolution requirements. However, essential for their proof was the fact, that their case only spanned a convolution of depth two, where the filter was only applied in the later phase. Hence, by reverting the second step, one gets a closed form transformation (Zeta) that can be directly inverted. In this paper, however, we encounter a more sophisticated $3$-level convolution ($4$-level for the connected variant) with two different filters. As one can see in the proof, reverting the last step does not yield a closed form transformation. Therefore, a more involved reversing process is needed. We show, somewhat surprisingly, that retracting the convolution steps after applying the product still yields the right answer. We believe that this convolution technique can be generalized to an arbitrary depth, and hence, can be of independent interest.

We note that our convolution is tight, and it improves on a trivial convolution with time $\Oh(36^{\cw}n^c)$, and, to the best of our knowledge, on the best known algorithm for general convolutions~\cite{DBLP:journals/corr/BrandCLP25} with running time $\Oh(6^{(2\omega/3) \cw}n^c)$, where $\omega$ is the matrix multiplication exponent, even for $\omega = 2$.

Back to our algorithm, this results in an algorithm that correctly counts (modulo $2$) the number of weighted feedback vertex sets of size $\target$ in time $\Oh(6^{\cw}n^c)$, when the weights are bounded polynomially in $n$. Using the isolation lemma~\cite{DBLP:journals/combinatorica/MulmuleyVV87}, we show that by choosing the weights large enough, independently and uniformly at random, this results in an algorithm that solves the decision version with high probability.

\subparagraph{Connected feedback vertex set.}

In order to solve the \Cfvsp problem parameterized by clique-width, we need to preserve connectivity of the solution set in addition to the acyclicity of the remaining graph. We achieve this by combining our acyclicity representation technique with a connectivity representation technique, that was introduced recently by Bojikian and Kratsch~\cite{DBLP:conf/icalp/BojikianK24}, called ``isolating a representative''. Surprisingly, this results in a tight algorithm with running time $\Oh(18^{\cw}n^c)$. In this technique, each partial solution $X$ is first represented by its \emph{connectivity pattern}, defined as a set of subsets of labels, where we add a set for each connected component $C$ of $G_x[X]$, that contains the labels of all vertices of $C$. Compared to acyclicity representation, where we needed to distinguish whether a label appears once or at least twice in a connected component, for connectivity, it suffices to track existence, and hence, simple sets suffice.

However, a count-preserving replacement argument, similar to the one above, as the authors suggest, would only reduce the size of the family to $4^k$ (where each label has a state $\stnone$, $\stdisc$, $\stconn$, or $\stplus$). Instead, they first introduce a notion of ``existential representation'', where the representing family does not preserve the parity of the number of compatible partial solutions, but only preserves the existence of a compatible partial solution. In order to achieve this, they define local operations, called \emph{actions} over the nodes of $\syntaxtree$, that decide for a given label (with state $\stconn$), whether it will be used in the future to connect different components of a partial solution (and hence, it can be added as a singleton safely, turning its state to $\stplus$), or will not be used anymore (and hence, it can be safely removed from all sets, turning its state to $\stnone$). The resulting family can then be reduced to size $3^k$ by a count-preserving reduction argument similar to the one above.

A problem however arises, since the first step is only an existential representation, that might not preserve the parity of the number of compatible partial solutions, while the second step only preserves the parity, by possibly creating new compatible partial solutions. They solve this problem by a double usage of the isolation lemma, where in addition to isolating the solution itself, they also isolate an existential representative thereof---and hence, the name ``isolating a representative''. They achieve this by assigning different weights to the actions taken to create each representative of a partial solution. They call the combinations of actions taken to create a single representative an \emph{action sequence}.

Therefore, our partial solutions are pairs of a vertex set $X$ and an action sequence $\pi$, where we index the dynamic programming tables by the acyclicity-pattern of $G\setminus X$ and by the connectivity-pattern of the representative created by $\pi$ from the partial solution $X$. This results in $6\cdot 3 = 18$ states per label.
Again, we need to efficiently process union nodes in the clique expression. We build upon the convolution technique introduced above, extending it to a $4$-level convolution, where we start by processing the connectivity part of the pattern. We show that this can be done in time $\Oh(18^{k}n^c)$, resulting in the claimed running time.

\subparagraph{Treewidth algorithm.}

Finally, in order to count the number of feedback vertex sets of a specific size when parameterized by treewidth, we make use of the same acyclicity representation technique. Essentially, in a tree decomposition, each vertex appears in a bag at most once, compared to labels in a labeled graph, while the size of the bags is bounded by the parameter. We start by defining a mapping $\phi$ that assigning to each vertex in the graph a label $i\in [k+1]$ such that the restriction of $\phi$ to each bag is injective. Using these labels, we can restrict the family of all patterns to patterns where each label appears at most once. This turns our definition of patterns, as multisets of vectors into sets over subsets of $[k]$. Essentially, it reduces the six states defined above into the ``single'' states $\stnone$, $\stdisc$ and $\stconn$ resulting in the family of \emph{nice treewidth patterns} $\CTP$ of size $3^{k+1}$. We use this family to index the dynamic programming tables.

Similar to the algorithm when parameterized by clique-width, join nodes form a bottleneck in this dynamic programming routine, where we need to efficiently combine all pairs of partial solutions in the two children bags. We show that this can be done in time $\Oh(3^{k}n^c)$ by introducing an ordering over these states, and showing that the required convolution corresponds to the join product of the $k$th power of the corresponding lattice. We refer to \cref{sec:preliminaries} for a brief introduction of these terms. The running time then follows from a result by Hegerfeld and Kratsch~\cite{DBLP:conf/esa/HegerfeldK23} that shows how to compute such convolutions efficiently.

\subsection{Lower bounds}

\subparagraph{Feedback Vertex Set.}

In our lower bounds, we follow the general framework of SETH based reductions for structural parameters debuted by Lokshtanov et al.~\cite{DBLP:journals/talg/LokshtanovMS18}. However, instead of reducing from the $d$-SAT problem, we reduce from the $q$-CSP-$B$ problem defined as follows: Given is a set of $n$ variables for some integer $n$, and a set of $m$ constraints, where each constraint is defined over $q$ variables, and specifies which assignments of values from the set $[B]$ to these variables are allowed. The goal is to decide whether there exists an assignment of values from $[B]$ to all variables, such that all constraints are satisfied.

We base our lower bound on a result by Lampis~\cite{DBLP:journals/siamdm/Lampis20}, that states, informally, that for each fixed value of $B$, and each $\delta > 0$, there exists a value of $q$, such that the $q$-CSP-$B$ problem cannot be solved in time $\Oh((B-\delta)^n)$, unless SETH fails.
Hence, for our lower bound for \Fvsp, we fix $B= 6$, and provide a reduction for each value of $q$, from the $q$-CSP-$6$ problem to \Fvsp. 

Intuitively, we translate each variable into a gadget of constant size, called \emph{path gadget}. We make use of the six states $\stnone$, $\stdisc$, $\stconn$, $\stddisc$, $\stplus$, and $\stdconn$ defined above and translate the intersection of a solution with this gadget into one of these states. We also fix a bijective mapping from these states to the different assignments of value from $[6]$ to the corresponding variable. We also define a so called \emph{constraint gadget} corresponding to each constraint $C$. We define the adjacencies between a constraint gadget and the path gadgets in such a way, that ensures that the states defined in the path gadgets by a solution corresponds to an assignment that satisfies the constraint $C$. 

In order to find an assignment that satisfies all constraints, we will add $m$ copies of a path gadget for each variable $v$, each corresponding to a different constraint, where consecutive copies are connected by bilciques of size $2$, spanning two ``exit vertices'' of each gadget, and two ``entry vertices'' of the following gadget, resulting in a \emph{path sequence} corresponding to each variable $v$.
This results in a grid structure, where each row is a sequence of path gadgets corresponding to a variable $v$, and each column corresponds to a constraint $C$. We attach a corresponding constraint gadget to each column.
  
The symmetric ``biclique'' cuts between consecutive path gadgets ensure that we can build all these sequences by ``spending'' a single label (unit of clique-width) for each path sequence, and a constant number of labels for the rest of the graph, building the whole graph column by column, which bounds the clique-width of the resulting graph. 
We show that these cuts are exactly enough to transition these six states. 

In fact, we add $5n+1$ copies of these sequences of path gadgets for each variable, making each such sequence $(5n+1)m$ long. This ensures that in any solution, there is one copy, where all gadgets on each sequence are assigned the same state. Building on this, we show that the given instance is a satisfiable, if and only if the resulting graph admits a feedback vertex set of some fixed size.

\subparagraph{Connected Feedback Vertex Set.}
For the \Cfvsp problem, we follow the same general schema. However, in order to achieve the base $18$, we combine an acyclicity state with a connectivity state, resulting in $6\cdot 3 = 18$ states per label, where an acyclicity state defines how a label appears in an induced forest, while a connectivity state defines the connectivity of a label in the solution set itself. However, a challenge arises, since states like $(\stnone, \stnone)$ are non-realizable, while the upper bound realizes them though the existential representation technique mentioned above, and hence, cannot be trimmed.

In order to overcome this challenge, we make use of the fact, that connectivity representation from \cite{DBLP:conf/icalp/BojikianK24} is a two-step existential representation, that turns a set of $4$ states into $3$ required states for connectivity. Hence, we make use of the ``fourth'' redundant state, and make careful choice of which $3$ connectivity states to combine with each acyclicity state. This does not only results in $18$ realizable states, but also results in combinations that allow for simpler realization, allowing to realize each state using three ``boundary'' vertices instead of four in a more straightforward approach.

\section{Preliminaries}\label{sec:preliminaries}

For a computable function $f$ and a parameter $k$, we denote by $\ostar(f(k))$ a running time $\Oh(f(k) \cdot n^c)$ for a constant value $c$ that does not depend on $f$, where $n$ is the size of input.
For a positive integer $k$, we define $[k]:=\{1,\dots, k\}$ and $[k]_0 := [k] \cup \{0\}$. We also denote by the square brackets $[P]$ the Iverson bracket, which evaluates to $1$ if the predicate $P$ is true, and to $0$ otherwise.
Finally, for a mapping $\alpha :U\rightarrow V$, and a pair of elements $(a,b)$ with $b\in V$, we define the \emph{extension} $\alpha[a\mapsto b]:U\cup\{a\}\rightarrow V$ using the square brackets as well, where $\alpha[a\mapsto b](a) = b$ and $\alpha[a\mapsto b](u) = \alpha(u)$ for all $u\in U\setminus\{a\}$. For a subset $S\subseteq U$, we denote by $\alpha|_S$ the restriction of $\alpha$ to $S$.

Given a mapping $\weightf\colon U\rightarrow \mathbb{F}$ for some set $U$ and a ring $\mathbb{F}$, that is explicitly defined as a weight function, we define the weight of a set $S\subseteq U$ as $\weightf(S) = \sum_{u\in S} \weightf(u)$.

\subparagraph{Graphs and clique-width.}

In this work we deal with undirected graphs only. We assume that the inputs are simple undirected graphs. However, in the upper bound, we will extend the definition of a clique-expression into mutligraphs for ease of representation.
When we say we \emph{identify} two vertices of a graph, we mean we remove both vertices and add a single vertex adjacent to all their neighbors. In particular, in a multigraph, this might create loops or multiedges between vertices.

A \emph{labeled graph} is a graph $G=(V, E)$ together with a \emph{labeling function} $\lab\colon V\rightarrow \mathbb{N}$. We usually omit the function $\lab$ and assume that it is implicitly given with $G$. We say that $G$ is \emph{$k$-labeled}, if it holds that $\lab(v)\leq k$ for all $v\in V$. A \emph{labeled forest} is a labeled graph that is a forest.
We define a \emph{clique expression} $\mu$ as a well-formed expression defined by the following operations on labeled graphs:
\begin{itemize}
    \item \emph{Introduce vertex} $i(v)$ for $i\in\mathbb{N}^+$. This operation constructs a graph containing a single vertex and assigns label $i$ to this vertex.
    \item The \emph{relabel} operation $\relabel{i}{j}(G)$ for $i,j \in \mathbb{N}^+$, $i\neq j$. This operation changes the labels of all vertices in $G$ labeled $i$ to the label $j$.
    \item The \emph{join} operation $\clqadd i j (G)$ for $i,j\in\mathbb{N}^+, i\neq j$. The constructed graph results from $G$ by adding all edges between the vertices labeled $i$ and the vertices labeled $j$, i.e.
    \[\clqadd{i}{j}(G) = (V, E \cup \{\{u, v\}\colon\lab(u)=i \land \lab(v)=j\}).\]
    \item The \emph{union} operation $G_1 \clqunion G_2$. The resulting graph is the disjoint union of $G_1$ and $G_2$.
\end{itemize}

We denote the graph resulting from a clique expression $\mu$ by $G_{\mu}$, and the constructed labeling function by $\lab_{\mu}$. We associate with a clique expression $\mu$ a syntax tree $\syntaxtree_{\mu}$ (we omit $\mu$ when clear from context) in the natural way, and associate with each node $x\in V(\syntaxtree)$ the corresponding operation. For $x\in V(\syntaxtree)$, the subtree $\syntaxtree_x$ rooted at $x$ induces a subexpression $\mu_x$. We define $G_x = G_{\mu_x}$, $V_x = V(G_x)$, $E_x=E(G_x)$ and $\lab_x = \lab_{\mu_x}$. Given a set of vertices $S\subseteq V$ and a node $x\in \nodes$, we denote $S_x = S\cap V_x$.

We say that a clique expression $\mu$ is a \emph{$k$-expression} if $G_x$ is a $k$-labeled graph for all $x\in V(\syntaxtree)$. We define the clique-width of a graph $G$ (denoted by $\cw(G)$) as the smallest value $k$ such that there exists a $k$-expression $\mu$ with $G_{\mu}$ isomorphic to $G$.
We can assume without loss of generality, that any given $k$-expression defining a graph $G=(V,E)$ uses at most $O(|V|)$ union operations, and at most $O(|V|k^2)$ unary operations~\cite{DBLP:journals/tcs/BergougnouxK19, CourcelleO00}. A linear $k$-expression $\mu$ is a $k$-expression, such that $\syntaxtree_{\mu}$ is a caterpillar, i.e.\ each union node has at least one child that corresponds to an introduce vertex operation.

Finally, we call a clique expression \emph{irredundant}, if for any join operation $\clqadd{i}{j}(\mu')$ in $\mu$ it holds that there are no edges between the vertices labeled $i$ and the vertices labeled $j$ in $G_{\mu'}$.

\subparagraph{Vectors and multisets.}
For two sets $U,R$, a vectors $v\in \mathbb{R}^U$, and an element $x\in U$, we use both notations $v_x$ and $v[x]$ to index the element $x$ in $v$ depending on the context. We use the square-bracket notation $v[x]$ primarily when $R$ is some ring ($\bin$) and $v$ is some table indexed by footprints of our algorithm. We also consider vectors as implicit mappings. Hence, for a set $S\subseteq R$, $v^{-1}(S)$ denotes the set of all indices $x\in U$ such that $v_x \in S$.
Finally, when $R$ is a ring with zero element $0$, we define the \emph{support} of $v$ as the set $U\setminus v^{-1}(0)$, i.e.\ the set of all indices that map to non-zero elements of $R$.

In this work we deal with multisets. Let $U$ be some ground set. We denote by $\mset(U)$ the family of all multisets over $U$. We define the multiplicity function $\#_S : U\rightarrow \mathbb{N}$ for a multiset $S$, that assigns to each element $u\in U$ its multiplicity in $S$. Each multiset is uniquely defined by its multiplicity function. We say that $p$ belongs to $S$, i.e.\ $p\in S$ if it holds that $\#_S(p)\geq 1$.
We define the \emph{union} of two multisets, given by the sum of the multiplicities of their elements, i.e.\ for $R := S\cup T$ it holds that $\#_R(u) = \#_S(u) + \#_T(u)$ for each $u\in U$.
We define the \emph{difference} of two multiset $R = S \setminus T$, where for each $u\in U$ it holds that
$\#_R(u) = \max\{0, \#_S(u) - \#_T(u)\}$.

\subparagraph{Convolutions, lattices and join product.}
Let $\statejoin:S\times S\rightarrow S\cup\{\perp\}$ be a binary operation over a set $S$, where $\perp$ does not belong to $S$, and represents a \emph{bad join}. Given two tables $A,B\in \mathbb{F}^S$ for some ring $\mathbb{F}$, the convolution of $A$ and $B$ over $\statejoin$ (or over $S$ if $\statejoin$ is clear from context) is defined as the table $C = A\stateconv B$, where 
for all $x\in S$ it holds that 
\[C[x] = \sumstack{y,z\in S\\y\statejoin z = x} A[y]\cdot B[z].\]

\begin{observation}\label{obs:iso-same-time}
    Let $U$ and $V$ be two sets with a join operation defined over each of them, and let $\phi$ be an isomorphism between $U$ and $V$ that preserves the join operation, such that both $\phi$ and its inverse can be computed in polynomial time. If the convolution over $V$ can be computed in time $f$, then the convolution over $U$ can be computed in time $\ostar(f+|U|)$.
\end{observation}

We call an ordering $\preceq$ over a set $S$ a join-semilattice, if for all elements $a,b\in S$ there exists a unique least upper bound $a\lor b$ (called the join of $a$ and $b$). Given two tables $T_1,T_2\in \mathbb{F}^{S}$ for some ring $\mathbb{F}$, we define the join product of $T_1$ and $T_2$ as the table
\[
T_1 \joinprod T_2 [s] = \sumstack{a,b\in S\\a\lor b = s} T_1[a]\cdot T_2[b].
\]

\begin{remark}
    The join product of a lattice is the convolution defined by its join operation.
\end{remark}

\section{Clique-width Upper bound}\label{sec:ub}

\subsection{Fusion clique expressions}

We note that, in comparison to the typical settings where clique expressions are defined over simple labeled graphs, we deal with labeled multigraphs along this upper bound. Hence, a join might create a double edge between two vertices, if an edge already exists, resulting in a cycle of length two. However, as mentioned earlier, we can assume that the given expression is irredundant, and therefore, each edge is added only once. Hence, this results in the same graph independent of whether we deal with simple graphs or multigraphs.

In fact, we will use a slightly modified version of clique expressions that will allow deterministic representation at join nodes and fast convolution operations at union nodes.
We call these expressions \emph{fusion clique expressions}, defined over the so called extended labeled graphs, where a $k$-extended labeled graph is a $k$-labeled graph with additionally a single vertex labeled $0$ (called the \emph{zero-vertex}). The operation $i(v)$ create two isolated vertices, $v$ labeled $i$ and $v_0$ labeled $0$. Join and relabel are defined in the same way, whereas for the union operation, after unifying the two extended labeled graphs we identify the two vertices labeled $0$ into a single vertex. 
Therefore, the graph resulting from a fusion clique expression is simply the graph resulting from adding $v_0$ as an isolated vertex to the graph resulting from the original clique expression without any modifications.

Fusion clique expressions are in fact a special case of fusion-tree expressions introduced by Fürer~\cite{DBLP:conf/latin/Furer14}, where we only apply the fusion operation to the label $0$, directly after each union operation, and we never include this label in any join or relabel operations.
The reason we define fusion clique expressions is that they simplify representation allowing to create an additional connected component where we accumulate adjacencies in a representation of a partial solution. This will become clear in the next sections.

Given a $k$-labeled graph $G$, we define the \emph{canonical clique expression} $\mu_G$ of $G$ as the linear clique expression defined as follows: Let $v_1,\dots v_n$ be the vertices of $G$. First, we introduce the vertices $v_1,\dots v_n$ assigning the label $k+i$ to $v_i$, and we unify each vertex $v_i$ with the part introduced so far. After that, for each edge $\{i,j\}\in E(G)$, we add this edge by the join $\clqadd{k+i}{k+j}$. Finally, we relabel each vertex $v_i$ to its final label by calling $\relabel{k+i}{\lab(v_i)}$. Note that the resulting expression is a linear $(n+k)$-expression.

In the rest of this upper bound, we exclusively consider extended labeled graphs and fusion clique expressions. Therefore, we might omit the words extended and fusion, and refer to them as labeled forests and clique expressions respectively.

Let $G = (V, E)$ be the input graph with $n$ vertices and $m$ edges, and $\target$ be the target size of a feedback vertex set. Let $\mu$ be a $k$-expression of $G$, for some value $k\in\mathbb{N}$, and let $\syntaxtree$ be the corresponding syntax tree of $\mu$. Let $\nodes := V(\syntaxtree)$ and $r\in\nodes$ be the root of $\syntaxtree$. We define the ground set $U = [k]$, and call it the set of labels.
Let us also fix an integer $\W$, and a weight function $\weightf:V\rightarrow [\W]$ that both will be chosen later, where $\W$ is bounded polynomially in $n$. We extend $\weightf$ with $\weightf(v_0) = 0$.
Our goal is count (modulo $2$) the number of feedback vertex sets of size $\target$ and weight $\weight$ for each value of $\weight \in [\W n]_0$ in $G_r$ that exclude the vertex $v_0$, which are exactly the feedback vertex sets of size $\target$ and weight $\weight$ in $G$, since $v_0$ is isolated in $G_r$.
This will allow us to apply the isolation lemma, by choose $\W$ large enough, and choosing $\weightf$ uniformly at random from $[\W]$ for each vertex independently.

In fact, instead of counting feedback vertex sets in $G_r$, we will count the number of induced forests of size $\budget$ and weight $\weight$ for all values of $\budget\in[n]_0$ and $\weight \in [n\cdot W]_0$ that contain the vertex $v_0$. Clearly, a set $X\subseteq V$ is a feedback vertex set of size $\target$ and weight $\weight$ in $G_r$ if and only if $V_r\setminus X$ is a forest of size $|V_r| - \target$. Therefore, the number of feedback vertex sets of size $\target$ and weight $\weight$ in $G_r$ that exclude the vertex $v_0$ is equal to the number of forests of size $|V_r| - \target = n+1-\target$ and weight $\weightf(V_r\setminus X)$ in $G_r$ that contain the vertex $v_0$. This motivates the following definition of a partial solution.

\begin{definition}
A \emph{partial solution} $S\subseteq V_x$ at a node $x\in\nodes$ is a vertex set that contains the zero-vertex, such that $S$ induces a forest in $G_x$. Hence, a partial solution is a labeled forest.
\end{definition}

Along this upper bound, we will use the letter $\budget$ to index values in $[n]_0$, and $\weight$ to index values in $[n\cdot \W]_0$. We skip repeating their definitions to avoid redundancy.

\subsection{Connectivity patterns}

Since the number of partial solutions is huge, we cannot keep track of each partial solution individually. Instead, we will keep footprints defined by these partial solutions, and count for each footprints the number of partial solutions of each size and weight having this footprint. We prove that this suffices to count all partial solutions correctly along $\syntaxtree$ using dynamic programming. 

Intuitively, one can represent connectivity, by representing each connected component by its labels, and by representing a partial solution by the set of all representations of all its connected components. However, for acyclicity representation the situation is different, as two connected components with different number of vertices labeled a specific label $i$ are not equivalent, and two partial solutions with different number of copies of the same component are also not compatible. For example let $S,S'$ be two solutions both inducing two connected components in $G_x$, one of them is a single vertex labeled $i$ and the other in $S$ is a single vertex labeled $j$ while in $S'$ is a single edge whose endpoints are both labeled $j$. Then by applying $\clqadd{i}{j}(\mu)$ one gets a valid partial solution for $X$ but not for $X'$. Hence, it seems that keeping multiplicity is necessary for acyclicity representation. However, we show that it suffices to keep multiplicity upper-bounded by two both for the labels appearing in each connected component, and for the connected components themselves. We will use vectors for counting the former, and multisets over these vectors for the latter. We first extend our notation for multisets to capture this multiplicity upper bound.

\begin{definition}
Let $U$ be some ground set. We define $\tset(U) \subseteq \mset(U)$ as the family of multisets with maximal multiplicity 2.
Given a multiset $R\in\mset(U)$ we define the multiset $R\tdown\in\tset(U)$ as the multiset resulting from $R$ by upper bounding the multiplicity of each element $u\in U$ by $2$, i.e.\ $\#_{R\tdown}(u) = \min\{2, \#_R(u)\}$.
We also define the \emph{two-union} operation $\tcup$ over multisets, where $S\tcup T := (S\cup T)\tdown$.
\end{definition}

\begin{definition}\label{def:two-sum}
    Given two vectors $u,v\in Z$, we define the \emph{two-sum} of $u$ and $v$ (denoted $u\tplus v$), as the vector $w\in Z$, where for $j\in[k]_0$ it holds that $w_j = \min\{2, u_j + v_j\}$.
\end{definition}

\begin{definition}\label{def:pattern}
    Let $Z = ([k]_0)^{\{0,1,2\}}$ be the family of all vectors assigning a value $0$, $1$ or $2$ to each label $i\in[k]_0$. A \emph{pattern} $p$ is a multiset of elements from $Z$, with multiplicity at most $2$, such that there exists exactly one vector $v\in Z$ with $v_0 = 1$ and $\#_p(v) = 1$, and it holds that $w_0 = 0$ for all other vectors $w\in p$ different from $v$. 
    We call $v$ the \emph{zero-vector} of $p$ and denoted by $z(p)$.
    We denote by $\Pat \subseteq \mset(Z)$ the family of all patterns.
    We also define the operator $\tot_i: \Pat \rightarrow \mathbb{N}$ for $i\in[k]_0$ as $\tot_i(p) := \sum_{v\in p} v_i$.
    We define the set $\lbs(p) := \{i\in[k]\colon \tot_i(p) > 0\}$.

    Given a labeled graph $H$, we define the \emph{type} of $H$ (denoted $\labv(H)\in Z$) as the vector that assigns to each index $j\in[k]_0$ the number of vertices labeled $j$ in $H$, upper bounded by $2$, i.e.
    \[
        \big(\labv(H)\big)_j := \min\big\{2, \big|\lab_H^{-1}(j)\big|\big\}.
    \]

    Given a labeled forest $F$,
    let $C_1,\dots C_{\ell}$ be the connected components of $F$.
    We define $\pat(F)$ the pattern corresponding to $F$ as 
    $
    \big\langle \labv(C_1), \ldots, \labv(C_\ell)\big\rangle\tdown
    $.
    The condition on the zero-vector is preserved by the assumption that there exists a single vertex labeled $0$ in any (extended) labeled graph $G$.
    Given a family of labeled forests $\mathcal{F}$, we define
    $
    \pat(\mathcal{F}) := \bigdelta_{F\in\mathcal{F}} \big\{\pat(F)\big\}
    $.

    Given a pattern $p$, we define the \emph{canonical forest} $F_{p}$ of $p$ as the labeled forest containing a simple path $C_i$ for each vector $v^{(i)}$ of $p$, where for each $j\in[k]_0$, let $c:=v^{(i)}_j$, then we add $c$ vertices labeled $j$ to $C_i$.
\end{definition}

\begin{figure}[ht]
    \centering
    \begin{subfigure}[b]{.3\textwidth}
        \centering
        \includegraphics[width=.7\textwidth]{media/labeled-graphs.pdf}
        \caption{A $2$-labeled forest $F$}
    \end{subfigure}
    \hfill
    \begin{subfigure}[b]{.35\textwidth}
        \centering
        $\left\langle
        \begin{pmatrix}
            0 \\
            2 \\
            0
        \end{pmatrix},
        \begin{pmatrix}
            1 \\
            1 \\
            1
        \end{pmatrix},
        \begin{pmatrix}
            0 \\
            0 \\
            1
        \end{pmatrix},
        \begin{pmatrix}
            0 \\
            0 \\
            1
        \end{pmatrix}
        \right\rangle
        $

        \vspace{1em}
        \caption{The pattern $p:=\pat_F$}
    \end{subfigure}
    \hfill
    \begin{subfigure}[b]{.3\textwidth}
        \centering
        \includegraphics[width=.7\textwidth]{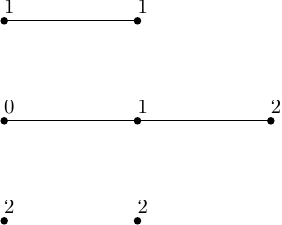}
        \caption{The canonical forest $F_p$.}
    \end{subfigure}
    \hfill
    \caption{A $k$-labeled forest $F$ ($k=2$), its corresponding pattern and the corresponding canonical forest. Note that we upper bound the multiplicity of each label in each connected component by $2$, and the multiplicity of each connected component by $2$ as well.}
    \label{fig:three-subfigs}
\end{figure}

In the following, we show that these patterns carry enough information about partial solutions to correctly count the number of partial solutions using dynamic programming over $\syntaxtree$. However, the number of all patterns is still very large, and hence, it is infeasible to index our dynamic programming scheme with all patterns. Therefore, we will even reduce the family of all patterns into a smaller representative family, that preserves the parity of the number of partial solutions over $\syntaxtree$.

\begin{definition}
    A \emph{clique extension} $\cext$ is a clique expression such that there exists a single clique-operation in $\cext$ with a variable operand $x$. Hence, $\cext$ can be seen as a mapping that takes a clique expression $\mu$ as input, and outputs the clique expression $\cext(\mu)$ obtained by replacing the variable operand $x$ in $\nu$ with $\mu$. Hence, $\syntaxtree_{\cext(\mu)}$ results from $\syntaxtree_{\cext}$ by plugging the syntax tree $\syntaxtree_{\mu}$ at the node corresponding to $x$, i.e.\ $\syntaxtree_{\mu}$ is a subtree of $\syntaxtree_{\cext(\mu)}$.

    Given a clique extension $\tau$, we define $G_{\tau'}$ and $V_{\tau'}$ for a subextension $\tau'$ of $\tau$ in a similar way to clique expressions, where for $v_x$ the node corresponding to the variable $x$ we define $G_{v_x}$ as the empty labeled graph.
\end{definition}

\begin{definition}\label{def:partial-compat-equiv}
    Let $\cext$ be some clique extension and $F'$ be a partial solution of $G_{\cext}$. For a labeled forest $F$, we say that $F$ is \emph{compatible} with $F'$ under $\cext$ (denoted $F\cmptb_{\cext}F'$) if $F\cup F'$ induces an acyclic subgraph in $G_{\cext(\mu_F)}$ where $\mu_F$ is the canonical clique expression of $F$. We say that $F'$ is compatible with a pattern $p$ under $\cext$ if $F_p\cmptb_{\cext}F'$, where $F_p$ is the canonical forest of $p$.

    Given two families of labeled forests $\mathcal{F}_1$, $\mathcal{F}_2$, we say that $\mathcal{F}_1$ is \emph{equivalent} to $\mathcal{F}_2$ ($\mathcal{F}_1\pquiv \mathcal{F}_2$) if it holds for each clique extension $\cext$ and each partial solution $F'$ of $G_{\cext}$ that
    \[
    |\big\{F_1\in\mathcal{F}_1\colon F_1\cmptb_{\cext}F'\big\}|\bquiv 
    |\big\{F_2\in\mathcal{F}_2\colon F_2\cmptb_{\cext}F'\big\}|.
    \]
    We say two labeled forests $F_1, F_2$, are \emph{equivalent} ($F_1\pquiv F_2$), if it holds that $\{F_1\}\pquiv \{F_2\}$, i.e.\ if $F_1$ is compatible with a partial solution $F'$ of an extension $\cext$ if and only if $F_2$ is compatible with $F'$ under $\cext$.
    Finally, we say that two families of patterns $P_1, P_2$ are equivalent, if their families of canonical forests are equivalent, i.e.\ if it holds that 
    $\{F_p\colon p\in P_1\}\pquiv \{F_p\colon p\in P_2\}$.
\end{definition}

It is not hard to see that $\pquiv$ is an equivalence relation over labeled forests. Now we aim to prove some equivalence relations over labeled forests. In order to achieve this, we start by establishing some technical lemmas.

\begin{lemma}\label{lem:ub-cycle-structure}
    Let $\cext$ be a clique extension and $F'$ be a partial solution of $G_{\cext}$. Let $F$ be a labeled forest and $C$ be a simple cycle in $\tilde{G} := G_{\cext(\mu_F)}[F\cup F']$. Then There exists a simple cycle $C'$ in $\tilde{G}$ such that for each connected component $D$ of $C'$ it holds that the vertices of $D\cap C'$ appear as a continuous segment on $C'$, and for each vector $v\in Z$ there exist at most two connected components $D,D'$ of $\tilde{G}$ intersecting $C'$ with $\labv(D) = \labv(D') = v$.
\end{lemma}

\begin{proof}
Let $D$ be a connected component, whose intersection with $C$ does not appear as a continuous segment. We can assume that the first vertex of $C$ belongs to $D$ but not the last by rotating $C$. Let $C = v_1,\dots v_n$, and let $v_2$ be the last vertex of $D$ appearing on $C$. Then we can replace the part of the cycle between $v_1$ and $v_2$ by a path inside $D$, as $D$ is a connected component, whereas $v_2,\dots,v_n,v_1$ is a path from $v_2$ to $v_1$ that intersects $D$ only in its endpoints. Hence, we get a cycle where vertices of $D$ appear as a continuous segment. Note that we have replaced a whole part of $C$ between two vertices of $D$. Therefore, for some other connected component $D'$ different from $D$ whose vertices already appear as a continuous segment, we have either removed all vertices of $D'$ from $C$, or kept them untouched. Hence, we can repeat this process until all connected components appear as continuous segments.

Now let $v\in Z$, and let $D_1,\dots D_{\ell}$ be the connected components of $\tilde{G}$ that intersect $C$ such that $\labv(D_i) = v$ for all $i\in[\ell]$, given in the order they appear on $C$. Assume that $\ell\geq 3$. Let $w$ be the first vertex of $D_2$ appearing on $C$, and $u$ be the vertex preceding $w$ on $C$. Let $i$ be the label of $w$ in $F$. Note that $\{u,w\}\notin E(F)$, since $D_2$ is a connected component of $F$. Moreover, since $\labv(D_2)=\labv(D_{\ell})$, it must hold that $D_{\ell}$ contains a vertex $w'$ labeled $i$ in $F$. Hence, the edge $\{u,w'\}$ must exist in $\tilde{G}$. We can then replace the part of $C$ between $u$ and $w$ with the edge $\{u,w'\}$, which results in a cycle $C'$ that intersects at most two connected components of $\tilde{G}$ of type $v$. Since we only delete vertices from the cycle, the resulting cycle must preserve both continuous segments and the upper bound on the number of components of each type intersecting the cycle. Hence, we can apply this process for each type $v\in Z$ consecutively, until we get a cycle $C'$ that intersects at most two connected components of $\tilde{G}$ of each type.
\end{proof}

\begin{lemma}\label{lem:forest-equiv-pat}
    It holds for each labeled forest $F$ and $p := \pat(F)$ that $F \pquiv F_p$.
\end{lemma}
\begin{proof}
    Let $\cext$ be some clique extension, and let $F'$ be a partial solution of $G_{\cext}$. Let $\tilde{G} = G_{\cext(\mu_{F})}$, and $\tilde{G}_p = G_{\cext(\mu_{F_p})}$.
    
    First assume that $F_p \not\cmptb F'$, then there exists a cycle $C$ in $\tilde{G}_p$.
    Let us fix an injective mapping $\iota$ from the connected components of $F_p$ to the connected components of $F$ that maps each connected component of $F_p$ to a connected component of the same type in $F$. This is possible by the definition fo $\pat(F)$. Then one can turn $C$ into a cycle in $\tilde{G}$ by replacing each maximal segment $s$ of $C$ intersecting a connected component $D$ of $F_p$ with a path in $\iota(D)$, whose endpoints have the same labels as the endpoints of $s$. This implies that $F\not\cmptb F'$.

    Now assume that $F \not \cmptb F'$, and let $C$ be a simple cycle in $\tilde G$. By \cref{lem:ub-cycle-structure}, we can assume that the vertices of each connected component of $F$ intersecting $C$ appear as a continuous segment on $C$, and that for each type $v\in Z$ there are at most two connected components of $F$ with $\labv(D) = v$. 
    Let $\mathcal{D}$ be the family of the connected components of $F$ that intersect $C$. Again we define the injective mapping $\iota$ between $\mathcal{D}$ and the connected components of $F_p$ that maps each component in $\mathcal{D}$ to a component of the same type in $F_p$. This is possible, since for each $v \in Z$, $F_p$ contains two components of type $v$ if $F$ contains at least two, and $F_p$ contains one component of this type if $F$ contains one, and since $\mathcal{D}$ contains at most two components of each type.

    Hence, one can turn $C$ into a cycle in $\tilde{G}_p$ by replacing each maximal segment $s$ of $C$ intersecting a component $D\in\mathcal{D}$ with a path in $\iota(D)$, whose endpoints have the same labels as the endpoints of $s$. This implies that $F_p\not\cmptb F'$.
\end{proof}

\begin{corollary}\label{cor:equiv-if-equiv-pat}
    Given two labeled forests $F_1, F_2$ such that $\pat(F_1) \pquiv \pat(F_2)$, then it holds that $F_1$ is equivalent to $F_2$.
\end{corollary}
\begin{proof}
    Let $p_1= \pat(F_1)$, and $p_2=\pat(F_2)$. It holds by definition that $F_{p_1} \pquiv F_{p_2}$, and by \cref{lem:forest-equiv-pat} that $F_{p_1} \pquiv F_1$ and $F_{p_2} \pquiv F_2$. It follows by transitivity of $\pquiv$ that $F_1 \pquiv F_2$.
\end{proof}

\subsection{Nice patterns and Pattern representation}

So far we have shown that patterns already carry enough information to represent partial solutions over $\syntaxtree$. However, since the number of patterns is very large for our algorithm, we define a special family of patterns, called the family of \emph{very nice patterns} $\CSP$ of size $6^k$. We show that one can find for any family of patterns representing a family of partial solutions in $G_x$ for some node $x\in\nodes$, a family of very nice patterns that preserves the parity of the number extensions of these partial solutions into solutions in the whole graph $G$. This allows us to 
restrict the dynamic programming tables to patterns in $\CSP$.

\begin{definition}
    A pattern $p$ is \emph{nice}, if every vector of $p$ different from the zero-vector is a unit vector.
    A pattern $p$ is \emph{very nice}, if it is nice and it holds in addition for each $i\in[k]$ that 
    $\sum_{v\in p} v_i \leq 2$.
    We denote by $\CSP\subseteq \Pat$ the family of all very nice patterns.
\end{definition}

\begin{definition}\label{def:csp-states-bijection}
    We define the set of states $\simpstates=\{\stnone, \stdisc,\stddisc, \stconn, \stplus, \stdconn\}$, where $\stconn$ stands for Connectivity and $\stdisc$ stands for disconnectivity. The star stands for double (dis-)connectivity, $C^+$ stands for the combination of both, and $\stnone$ stands for neither.

    Let $\mathcal{S} = \simpstates^{k}$ be the family of all vectors that assign a state $s\in S$ to each label, and let $\phi_{CS}$ be the canonical bijection from $\CSP$ to $\states$ that assigns to a pattern $p\in \CSP$ the vector $s:=\phi_{CS}(p)$, where for $i\in[k]$ the state $s_i$ as defined as follows:
    \begin{itemize}
        \item $s_i = \stnone$ if $v_i = 0$ for all $v\in p$,
        \item $s_i = \stdisc$ if $\#_p(\langle\idv_i\rangle) = 1$, and $z(p)_i = 0$,
        \item $s_i = \stddisc$ if $\#_p(\langle\idv_i\rangle) = 2$, and $z(p)_i = 0$,
        \item $s_i = \stplus$ if $\#_p(\langle\idv_i\rangle) = 1$, and $z(p)_i = 1$,
        \item $s_i = \stconn$ if $\#_p(\langle\idv_i\rangle) = 0$, and $z(p)_i = 1$,
        \item $s_i = \stdconn$ if $\#_p(\langle\idv_i\rangle) = 0$, and $z(p)_i = 2$.
    \end{itemize}
\end{definition}

\begin{observation}\label{obs:csp-states}
    It follows from the bijection $\phi_{CS}$ that $|\CSP| = 6^k$.
\end{observation}

Even though we index the dynamic programming tables by the family $\CSP$, we will use the bijection $\phi_{CS}$ implicitly, by calling $s_i$ the state of the label $i$ in $p$. We also use this bijection to define a fast convolution scheme that will allow us to process a union node more efficiently.

Now we show that for any family of patterns, one can efficiently find an \emph{equivalent} family of very nice patterns.
In the rest of this subsection, let $p_1, p_2, p_3, p_4 \in \Pat$ be four patterns, and let $v_1, v_2, v_3 \in Z$, and $R\subseteq Z$ such that
\begin{itemize}
    \item $p_1 = \langle v_1, v_2 \tplus v_3\rangle\tcup R$,
    \item $p_2 = \langle v_2, v_1 \tplus v_3\rangle\tcup R$,
    \item $p_3 = \langle v_3, v_1 \tplus v_2\rangle\tcup R$, and
    \item $p_4 = \langle v_1 \tplus v_2 \tplus v_3\rangle\tcup R$.
\end{itemize}

\begin{lemma}\label{lem:rep-vec-2-of-3}
    It holds for each clique extension $\cext$ and each partial solution $F'$ of $G_{\cext}$ that if $F'$ is compatible with one of $p_1,p_2,p_3$, then it is compatible with at least two of them.
\end{lemma}

\begin{proof}
    For $i\in[3]$, let $F_i$ be the canonical forest of $p_i$, and let $G_i = G_{\cext(\mu_{F_i})}[F_i\cup F']$. 
    Finally, Let $D_i$ be the component of $G_i$ corresponding to $v_i$, and $D'_i$ be the component corresponding to $v_{i'} \tplus v_{i''}$, where $i'$ and $i''$ are the two different values in $[3]$ different from $i$.
    
    For the sake of contradiction, assume that $F'$ is compatible with $p_1$, but it is not compatible with $p_2$ or with $p_3$. We will show that this leads to a contradiction.
    The cases where $F'$ is compatible with $p_2$ only, or with $p_3$ only are symmetric, which would prove the lemma.

    Hence, we assume that $G_2$ and $G_3$ contain cycles, but $G_1$ doesn't. Let $C_2, C_3$ be simple cycles in $G_2$ and $G_3$ respectively. If $C_2$ does not intersect $D'_2$, then it can be turned into a cycle in $G_1$, since each connected component of $G_2$ other than $D'_2$ is a subgraph of a connected component of $G_1$. Formally speaking, such a cycle can be turned into a cycle in $G_1$ by replacing each segment of the cycle, intersecting the component corresponding to $v_2$ by the same part in the component corresponding to $v_2 \tplus v_3$ in $G_1$. 
    Hence, we can assume that $C_2$ intersects $D'_2$. By \cref{lem:ub-cycle-structure}, we assume that $C_2$ starts with a vertex of $D'_2$ and only leaves $D'_2$ once, i.e.\ the vertices of the intersection form a prefix of $C_2$.
    Let $u_2$ and $u'_2$ be the first and the last vertices on $C_2$ that belong to $D'_2$, and let $i_2 = \lab(u_2)$ and $j_2 = \lab(u'_2)$.
    Hence, the rest of $C_2$ builds a path from $u_2$ to $u'_2$ that intersects $D_2$ only in its endpoints.
    If $(v_1)_{i_2}$ and $(v_1)_{j_2}$ are both positive (or $(v_3)_{i_2}$ and $(v_3)_{j_2}$ are both positive), then we can replace the path between $u_2$ and $u'_2$ with a path in $D_1$ (or $D'_1$ resp.) getting a cycle in $G_1$. Hence, we can assume that $(v_1)_{i_2} > 0$ and $(v_3)_{j_2} > 0$ (as the other case is symmetric).
    Analogously, we assume that $C_3$ intersects $D'_3$, and that all vertices of $D'_3$ in $C_3$ appear as a prefix of $C_3$ starting with $u_3$ and ending in $u'_3$ labeled $i_3$ and $j_3$ respectively. This already implies that $G_3$ contains a path from $u_3$ to $u'_3$ that intersects $D'_3$ only in its endpoints. Moreover we can assume that $(v_1)_{i_3} > 0$ and $(v_2)_{j_3} > 0$.

    Note that $u'_2$ and $u'_3$ correspond to two different vertices $w_2,w_3\in D'_1$ respectively, where $w_2$ has the same label as $u'_2$ and $w_3$ has the same label as $u'_3$. Moreover, $G_1$ contains paths $P'_2$ and $P'_3$ from $D_1$ to $w_2$ and $w_3$ respectively. Since $D'_1$ is a connected component, the concatenation of these two paths with the path between $w_2$ and $w'_2$ in $D'_1$ yields a cycle in $G_1$ which contradicts our assumption that $F'$ is compatible with $p_1$ but not with $p_2$ or $p_3$.
\end{proof}

\begin{lemma}\label{lem:rep-vec-4-weak}
    It holds for each clique extension $\cext$ and each partial solution $F'$ of $G_{\cext}$ that $F'$ is compatible with all patterns $p_1, p_2$ and $p_3$, if and only if it is compatible with $p_4$ as well.
\end{lemma}

\begin{proof}
    For $i\in[4]$, let $G_i := G_{\cext(\mu_{F_i})}[F_i\cup F']$. For $i\in[3]$, let $D_i$ be the connected component of $G_i$ corresponding to $v_i$, and $D'_i$ be the component corresponding to $v_{i'} \tplus v_{i''}$, where $i'$ and $i''$ are the two different values in $[3]$ different from $i$. Let $D_4$ be the connected component of $G_4$ corresponding to $v_1 + v_2 + v_3$.
    First assume that there exists $i\in[3]$ such that $F'$ is not compatible with $p_i$, and let $C$ be a cycle in $G_i$. Then we can turn $C$ into a cycle in $G_4$ by replacing any segment of $C$ in $D_i$ or $D'_i$ with a corresponding segment of $D_4$ that starts and ends with vertices of the same label as the original segment.

    For the other direction, assume that $F'$ is compatible with all $p_1$, $p_2$ and $p_3$, and for the sake of contradiction, assume that $F'$ is not compatible with $p_4$. Let $C$ be a cycle in $G_4$. If $C$ does not intersect $D_4$, then it is not hard to see that this can be turned into a cycle in $G_i$ for all $i\in[3]$ disjoint from $D_i$ and $D'_i$, since all four graphs agree outside of these components. Hence, we can assume that $C$ intersects $D_4$. By \cref{lem:ub-cycle-structure}, we can assume that the vertices of $D_4$ in $C$ appear as a prefix of $C$. Let $u$ be the first vertex of this segment and $u'$ be the last vertex, and let $\ell$ and $\ell'$ be the labels of $u$ and $u'$ respectively. If there exists $j\in[3]$ with $v_j\geq \idv_{\ell} + \idv_{\ell'}$, then we can replace the path between $u$ and $u'$ with a path in $D_j$ between two vertices of the same labels as $u$ and $u'$ getting a cycle in $G_j$. Otherwise, there exist two different values $j,j'\in[3]$, such that $(v_j)_{\ell} > 0$ and $(v_{j'})_{\ell'} > 0$. Let $j''$ be the unique value in $[3]\setminus\{j,j'\}$. Then $D'_{j''}$ contains two different vertices $w$ and $w'$ labeled $\ell$ and $\ell'$ respectively. Hence, we can replace the path between $u$ and $u'$ with a path in $D'_{j''}$ between $w$ and $w'$, which results in a cycle in $G_{j''}$. This shows that $G_{j''}$ contains a cycle, which contradicts the assumption that $F'$ is compatible with $p_1, p_2$ and $p_3$.
\end{proof}

\begin{corollary}\label{cor:rep-vec}
    It holds that $\{p_1\}\pquiv \{p_2\} \Delta \{p_3\} \Delta \{p_4\}$.
\end{corollary}

\begin{proof}
    Let $\cext$ be some clique extension, and let $F'$ be a partial solution of $G_{\cext}$. We show that $F'$ is compatible with an even number of patterns among $p_1, p_2, p_3, p_4$. The proof follows from the definition of representation.

    If $F'$ is compatible with none of the patterns, or all of them, then we are done. Otherwise, $F'$ must miss $p_4$ by \cref{lem:rep-vec-4-weak}, and by the same lemma, at least one of $p_1, p_2$ and $p_3$. But since $F'$ is compatible with at least one of them, it holds by \cref{lem:rep-vec-2-of-3} that $F'$ is compatible with at least two of them. Hence, $F'$ is compatible with exactly two of $p_1, p_2, p_3$. It follows in all cases that $F'$ is compatible with an even number of patterns out of $p_1, p_2, p_3, p_4$. Therefore, $F'$ is compatible with $p_1$ if and only if it is compatible with an odd number of patterns out of $p_2, p_3, p_4$. Note that the symmetric difference does not change the parity of the number of compatible patterns, since we only remove pairs of identical patterns.
\end{proof}

\begin{procedure}
    Given a pattern $p\in \Pat$, a vector $v\in p$ different from $z(p)$ and an index $j\in[k]$, we define the operation $\redind(p,v,j)$ as follows: 
    if $v_j = 0$, or $v$ is a unit vector, then we return $\{p\}$.
    Otherwise, let $z := z(p)$, $v_0 := v - \idv_j$, and $R := p\setminus\langle v, z\rangle$. 
    Then the operation returns the family $\redind(p_2, v_0, j)\Delta\{p_3\}\Delta\{p_4\}$, where $p_2, p_3, p_4$ are defined as follows:
    \begin{itemize}
        \item $p_2 = \langle z \tplus \idv_j, v_0\rangle\tcup R$,
        \item $p_3 = \langle z \tplus v_0, \idv_j\rangle\tcup R$, and
        \item $p_4 = \langle z \tplus v_0 \tplus \idv_j\rangle\tcup R$.
    \end{itemize}
\end{procedure}

\begin{lemma}\label{lem:reduce-index}
    The operation $\redind(p,v,j)$ runs in polynomial time and outputs a family of patterns $P$ equivalent to $p$, where $|P|\leq 5$. Moreover, it holds for each $p'\in P$ that $p'$ results from $p$ by replacing $v$ with either the singleton $\idv_j$ or with a vector $v'$ that results from $v$ by setting $v'_j = 0$, and by possibly adding positive values to $z(p)$.
\end{lemma}

\begin{proof}
    The equivalence holds from \cref{cor:rep-vec}, by choosing $v_1 = z(p)$, $v_2 = v_0$, $v_3 = \idv_j$, and $R = p\setminus\langle v, z\rangle$, as it follows for $p_1:=p$ that $p_1, p_2, p_3, p_4$ match the definitions provided in the beginning of this subsection. The running time and the output size follow from the observation that $p_3$ and $p_4$ already admit the desired form, and in $p_2$ we decrease $v_j$ by $1$, where $v_j$ is upper bounded by $2$. Hence, the recursion has depth at most $2$.
\end{proof}

\begin{procedure}
    Given a family of patterns $P\subseteq \Pat$, let $\rednice(P)$ be the operation defined by the following procedure:
    Let $P_0 := P$, and for $i\in[k]$, we define $P_i$ from $P_{i-1}$ as follows:
    First, we define $P_i^0 = P_{i-1}$. Then we define $P_i^j$ from $P_i^{j-1}$, where for $p\in P_i^{j-1}$, let $R_p = \redind(p, v, i)$ for some non-unit vector $v \in p$ different from $z(p)$ with $v_i \geq 1$ if such a pair $v,i$ exists, or $R_p = \{p\}$ otherwise. Then $P_i^j := \bigdelta_{p\in P_i^{j-1}} R_p$.
    Let $i^*$ be the smallest index such that $P_i^{i^*} = P_{i}^{i^*+1}$. Then we define $P_i = P_i^{i^*}$. The operation outputs $P_k$.
\end{procedure}

\begin{lemma}\label{lem:reduce-to-nice}
    The operation $\rednice(P)$ is finite, and it outputs a nice family of patterns $P'$ equivalent to $P$.
\end{lemma}

\begin{proof}
    We show by induction over $i\in[k]_0$ that $P_i$ is a family of patterns equivalent to $P_{i-1}$ (and hence to $P$), such that for all $j\leq i$ and for all $p\in P_i$ and all $v\in p$ where $v_j > 0$ we have $v_j = z(p)$ or $v_j$ is a singleton. The base case holds trivially, since $P_0 = P$.
    For the induction step, assume that the claim holds for $P_{i-1}$.
    Then by \cref{lem:reduce-index} it holds that $P_i^j$ results from $P_{i-1}$ by replacing some vectors $v$ with either $\idv_i$ or by a vector $v'$ that results from $v$ setting $v'_i = 0$. Since we do not change any other index in these vectors, and since we apply this operation exhaustively, the claim must hold for $P_i$. 
    
    The process is finite, since we deal with finite patterns only, and it holds that the maximum number of non-identity vectors $v$ in a pattern $p \in P_i^j$ with $v_i > 0$ is strictly smaller than the maximum number of such vectors in a pattern $p \in P_i^{j-1}$ for each $j\geq 1$.
\end{proof}

\begin{procedure}
    We define the operation $\operatorname{CleanPattern}(p)$ that takes as input a nice pattern $p$ and returns a pattern $p'$, where $p'$ results from $p$ by reducing the multiplicity of $\idv_i$ by $\max\{0, \tot_i(p) -2\}$ for each $i\in[k]$, i.e.\ $p'$ is a nice pattern with $z(p) = z(p')$, and $\tot_i(p') = \min\{2, \tot_i(p)\}$ for all $i\in[k]$.
\end{procedure}

Intuitively, this operation reduces the total number of times each label appear in a nice pattern to at most $2$ by removing identity vectors only.

\begin{lemma}\label{lem:clean-pattern}
    The operation $\operatorname{CleanPattern}(p)$ runs in polynomial time and outputs a very nice pattern $p'$ equivalent to $p$.
\end{lemma}

\begin{proof}
    The running time follows from the fact that we only access and modify the multiplicity of the identity vectors of the pattern, and check the values of $z(p)_i$ for some values $i\in[k]$. Now we show that $p$ and $p'$ are equivalent.
    
    Let $\cext$ be some clique extension, and let $F'$ be a partial solution of $G_{\cext}$.
    Let $G_p := G_{\cext(\mu_{F_p})}[F'\cup F_p]$, and $G_q := G_{\cext(\mu_{F_q})}[F'\cup F_{q}]$
    If $F'$ is compatible with $F_p$, then it is clearly compatible with $F_q$, since $G_q$ is a subgraph of $G_p$.
    
    Now assume that $F'$ is not compatible with $F_p$, and let $C$ be a cycle in $G_p$.
    Since $p$ is a nice pattern, each connected component of $F_p$ except for the one corresponding to $z(p)$ consists of a single vertex.
    By \cref{lem:ub-cycle-structure} we can assume that the intersection of $C$ with the component $D$ corresponding to $z(p)$ forms a continuous segment on $C$. Moreover, if this segment is not empty, we can rotate the cycle and assume that it starts with this segment. If there exists $i\in[k]$ such that $C$ visits at least three vertices of $F_p$ labeled $i$. Let $u,u'$ be the first two occurrences and $w$ the last one. Then it must also hold that $w$ is an isolated vertex in $F_p$, since $p$ is a nice pattern, and $z(p)_i \leq 2$.
 
    Let $w'$ be the vertex following $w$ on the cycle. Then both edges $\{u, w'\}$ and $\{u',w'\}$ exist in $G_p$, as $\{w,w'\}$ is added by the clique extension, and all vertices $u, u', w$ have the same label in $F_p$. Let $x,x'$ be the vertices following $u$ and $u'$ on the cycle respectively.
    Let $C'$ be the cycle resulting from $C$ by removing the part of the cycle between $x$ and $w$ if $\{u,w'\}$ is not the edge preceding $u$ on the cycle, or the part between $x'$ and $w$ otherwise. Then $C'$ is a cycle resulting from $C$ by removing vertices only, and it contains at most two vertices of $D$ labeled $i$. We repeat this process for all such labels $i$ getting a cycle $C'$ where each label appears at most twice. The cycle $C'$ corresponds to a cycle in $G_q$ since $z(p) = z(q)$.
\end{proof}

\begin{procedure}
    Given a pattern $p\in \Pat$, we define the operation $\redpat(p)$, where for $P:=\rednice(\{p\})$, the operation outputs the family
    $\bigdelta_{p\in P} \operatorname{CleanPattern}(p)$.
\end{procedure}

\begin{corollary}\label{cor:reduce}
    Given a pattern $p\in \Pat$, the operation $\redpat(p)$ outputs a very nice family of patterns equivalent to $p$.
\end{corollary}

\begin{proof}
    This follows from \cref{lem:reduce-to-nice} and \cref{lem:clean-pattern}, and by transitivity of representation.
\end{proof}

Now we define pattern operations that will allow us to modify a pattern corresponding to a partial solution in a way that corresponds to applying a clique-operation to that partial solution. These operations form the backbone of our dynamic programming algorithm.

\begin{definition}\label{def:pat-join}
    Given $i,j\in[k]$ with $i\neq j$, we define the predicate $\acyc_{i,j}(p)$ over all patterns $p\in\Pat$, where $\acyc_{i,j}(p)$ is true, if $\tot_i(p) = 0$, $\tot_j(p) = 0$, or the following two conditions hold:
    \begin{itemize}
        \item $v_i+v_j \leq 1$ for all $v\in p$,
        \item $\min\{\tot_i(p), \tot_j(p)\} \leq 1$.
    \end{itemize}
    Now we define the \emph{join} operation $\patadd_{i,j}:\TP\rightarrow \Pat\cup\{\badpat\}$, where $\badpat$ stands for a ``bad pattern'', as follows:
    \[
        \patadd_{i,j}(p)=
        \begin{cases}
            \badpat &\colon \lnot \acyc_{i,j}(p),\\
            p  &\colon \{i,j\}\not\subseteq \lbs(p),\\
            p' &\colon \text{Otherwise}.
        \end{cases}
    \]
    where $p'$ results from $p$ as follows: let $S_i = \langle v\in p\colon v_i > 0\rangle$ and $S_j = \langle v\in p\colon v_j > 0\rangle$, then $p'$ results from $p$ by removing all vectors of $S_i\cup S_j$ from $p$ and adding their two-sum, i.e.\
    \[
    p' := \Big(p\setminus \big(S_i\cup S_j\big)\Big) \tcup \Big\langle \sum_{v\in S_i} v \tplus \sum_{v\in S_j} v \Big\rangle.
    \]
\end{definition}

\begin{lemma}\label{lem:pat-join-acyclic}
    Let $F$ be a labeled forest, and let $p = \pat(F)$. Let $p'=\patadd_{i,j}(p)$, and $F' = \clqadd{i}{j}(F)$. Then it holds that $F'$ is acyclic if and only if $\acyc_{i,j}(p)$ holds. Moreover, assuming that $\acyc_{i,j}(p)$ holds, then it holds that $p' = \pat(F')$.
\end{lemma}

\begin{proof}
    First, assume that $\acyc_{i,j}(p)$ is false. We show that $F'$ must contain a cycle. Hence, it must hold that $\tot_i(p)\geq 1$ and $\tot_j(p)\geq 1$. If there exists $v\in p$ with $v_i + v_j \geq 2$, then there exists a connected component $D$ of $F$ with $\big|\big\{v\in D\colon \lab(v) \in \{i,j\}\big\}\big|\geq 2$. If $D$ contains a vertex $u$ labeled $i$, and a vertex $u'$ labeled $j$, then the edge $\{u,u'\}$ together with the path between $u$ and $u'$ in $D$ form a cycle in $F'$. Otherwise there must exist two distinct vertices $u$ and $u'$ in $D$ assigned the same label. Without loss of generality, let this label be $i$. Since $\tot_j(p)\geq 1$, there must exist a vertex $w$ in $F$ labeled $j$, but then the path $u,w,u'$ together with the path between $u$ and $u'$ in $D$ form a cycle in $F'$. Hence, we assume that $v_i+v_j\leq 1$ for all $v\in p$. Then it must hold that $\min\{\tot_i(p),\tot_j(p)\}\geq 2$. Let $u,u'$ be two distinct vertices labeled $i$, and $w,w'$ be two distinct vertices labeled $j$ in $F$. Then $u,w,u',w'$ is a cycle in $F'$.

    Now assume that $\acyc_{i,j}(p)$ holds. For the sake of contradiction, assume that $F'$ contains a cycle $C$. Since $F$ is a forest, $C$ must use an edge $\{u, w\}$ added by the operation $\clqadd{i}{j}(F)$, where $\lab(u) = i$ and $\lab(w) = j$. If $u$ and $w$ belong to the same connected component, then there must exist a vector $v\in p$ with $v_i + v_j \geq 2$, which contradicts our assumption. So we assume that $u$ and $w$ belong to two different components. Hence, there must exist a simple path from $u$ to $w$ in $F'$ not using the edge $\{u,w\}$. But since $v_i+v_j\leq 1$ the components of $u$ and $w$ do not contain any other vertices labeled $u$ or $v$. Hence, the path from $u$ to $w$ cannot visit any other vertices of their components in $F$, since otherwise these components would contain other vertices labeled $i$ or $j$. Hence, $u$ must have another neighbor labeled $j$ and $w$ must have another neighbor labeled $i$. But this contradicts the assumption that $\min\{\tot_i(p), \tot_j(p)\} \leq 1$. Hence, $F'$ must be acyclic.

    Assuming that $\acyc_{i,j}(p)$ holds, $p'=\pat(F')$ follows from observing that $\clqadd{i}{j}(F)$ combines all connected components of $F$ containing vertices labeled $i$ and $j$ into a single component, if both labels appear in the image of $\lab_F$, and keeps $F$ as it is otherwise.
\end{proof}

\begin{lemma}\label{lem:pat-join-time}
    Let $p \in \CSP$ be a very nice pattern, and let $i,j\in[k]$ with $i\neq j$ such that $\acyc_{i,j}(p)$ holds. Then it holds for $p' := \patadd_{i,j}(p)$ that $\redpat(p')$ runs in polynomial time and outputs at most $7$ patterns.
\end{lemma}

\begin{proof}
    If $\tot_i(p) = 0$ or $\tot_j(p) = 0$, then $p' = p$, and the claim follows trivially. Otherwise, if $z(p)_i \geq 1$ or $z(p)_j \geq 1$, then it holds that $p'$ results from $p$ by removing all identity vectors $\idv_i$ and $\idv_j$ from $p$, and adding them to $z(p)$. In this case $p'$ is also a very nice pattern.
    
    Now we assume that $z(p)_i = z(p)_j = 0$. Hence, $p'$ contains a single non-identity vector $v$ different from $z(p')$ with $v_k = 0$ for all $k\neq \{i,j\}$. Hence, $\rednice(p')$ only makes calls $\redind(q, w, k)$ with $k\in\{i,j\}$ and $q\setminus \langle w\rangle$ is a very nice pattern.
    This implies that all calls of $\redind(q, w, k)$ produce at least two very nice patterns.
    It also holds by $\acyc_{i,j}(p)$ that $\tot_i(p)+\tot_j(p)\leq 3$.
    It follows that the recursion tree of the $\redind(q,w,k)$ calls, made by $\rednice(p')$, has depth at most $3$, width $3$, and each node has at last two children that are leaves. Hence, the operation runs in polynomial time and outputs at most $7$ patterns.
\end{proof}

\begin{definition}\label{def:pat-relabel}
    Given a vector $v\in Z$, and $i,j\in[k]$, $i\neq j$, we define the vector $v_{i\rightarrow j}$ with:
    \[
        (v_{i\rightarrow j})_{\ell} = \begin{cases}
            0 & \colon \ell = i,\\
            v_i + v_j & \colon \ell = j,\\
            v_{\ell} & \colon \text{otherwise}.
        \end{cases}
    \]
    Given a pattern $p\in\Pat$, we define the relabel operator $p_{i\rightarrow j}$ as $\langle v_{i\rightarrow j} \mid v \in p \rangle\tdown$.
\end{definition}

\begin{definition}\label{def:pat-union}
    Given two patterns $p, q\in\Pat$, we define the operator $\patunion$ where $p\patunion q$ results from the union $p\tcup q$ by replacing the vectors $z(p)$ and $z(q)$ with $z(p) \tplus z(q) - \idv_0$, i.e.\
    \[
        p\patunion q = \Big((p\cup q)\setminus \langle z(p), z(q)\rangle\Big) \tcup \langle z(p) + z(q) - \idv_0\rangle.
    \]
\end{definition}

\begin{observation}\label{obs:pat-union-relabel}
    Let $F,F'$ be labeled forests, and let $p = \pat(F)$, $p' = \pat(F')$.
    Then it holds that $p\patunion p' = \pat(F\clqunion F')$, and $p_{i\rightarrow j} = \pat(\relabel{i}{j}(F))$ for all $i,j\in[k]$ with $i\neq j$.
\end{observation}

\subsection{Algorithm}

Now we present the algorithm, given as the recursive formulas of our dynamic programming tables. We follow with correctness proof, and we show that these tables can be computed in the claimed time.

\begin{algorithm}\label{algorithm}
We define the tables $T_x^{\budget,\weight} : \CSP \rightarrow \mathbb{N}$ for each $x\in\nodes$ and values $\budget$, $\weight$, recursively over $\nodes$ as follows:
\begin{itemize}
    \item Introduce vertex node $\mu_x = i(v)$: We set $T_{x}^{0,0}[\langle\idv_0\rangle]$ and $T_{x}^{1,\weightf(v)}[\langle \idv_0, \idv_i\rangle]$ to $1$, and $T_{x}^{\budget, \weight}[p]$ to $0$ for all other patterns $p\in \CSP$ and values $\budget,\weight$.
    
    \item Relabel node $\mu_x = \relabel{i}{j}(\mu_{x'})$: We define 
    $T_x^{\budget,\weight}[p] = \sumstack{p' \in \CSP\\ p'_{i\rightarrow j} = p} T_{x'}^{\budget,\weight}[p']$.

    \item Join node $\mu_x = \clqadd{i}{j}(\mu_{x'})$: We define
    $T_x^{\budget,\weight}[p] = \sumstack{p'\in\Pat\\p\in\redpat(p')} \sumstack{p'' \in \CSP\\ \patadd_{i,j}(p'') = p'} T_{x'}^{\budget,\weight}[p'']$.

    \item Union node $\mu_x = \mu_{x_1} \clqunion \mu_{x_2}$: We define
    \begin{equation}\label{eq:alg-union}
    T_x^{\budget,\weight}[p] = \sumstack{\budget_1+\budget_2 = \budget\\ \weight_1+\weight_2 = \weight}\sumstack{p_1,p_2\in\CSP\\ p = p_1\patunion p_2} T_{x_1}^{\budget_1,\weight_1}[p_1] \cdot T_{x_2}^{\budget_2,\weight_2}[p_2].
    \end{equation}
\end{itemize}
\end{algorithm}

\begin{definition}
    Let $\tilde{T}_x^{\budget,\weight}:\Pat\rightarrow \mathbb{N}$ be the tables defined over $\Pat$ that follow the same recursive definitions of $T_x^{\budget,\weight}$, but sum over patterns $p'\in\Pat$ instead of $\CSP$, and for join nodes skip the outer sum, i.e.\ for join nodes we sum over patterns $p'\in\Pat$ such that $\patadd_{i,j}(p') = p$.
\end{definition}

\begin{lemma}\label{lem:alg-counts-sol}
    It holds that $\tilde{T}_x^{\budget,\weight}[p]$ is congruent (modulo $2$) to the number of partial solutions $F$ in $G_x$ of size $\budget + 1$ and weight $\weight$ such that $\pat(F) = p$ for each pattern $p\in \Pat$.
\end{lemma}

\begin{proof}
    We prove this lemma by induction over $\nodes$.
    Note that the additional unit in $\budget+1$ counts for the vertex $v_0$, as the tables above do not count for this vertex, which allows a cleaner formula for a union node. For a leaf node $x$, it holds that there exists two partial solutions, the one containing $v_0$ only of size $1$ and weight $0$, and the one containing both $v$ and $v_0$ of size two and weight $\weightf(v)$. This corresponds exactly to the definition of $\tilde{T}_x$ at an introduce vertex node.

    Now let $x$ be some non-leaf node, and assume that the claim holds for all children of $x$. Let $F\subseteq V_x$, and let $F_x = G_x[F]$.
    For a relabel node $x$, $F$ is a partial solution at $x$ if and only if it is a partial solution at $x'$. Moreover, it holds by \cref{obs:pat-union-relabel} that $\pat(F_x) = \pat(\relabel{i}{j}(F_{x'}))$. Hence, the claim follows by the induction hypothesis.

    For a join node, it holds that $F$ is a partial solution in $G_{x}$ if and only if it is a partial solution in $G_{x'}$ and 
    $\clqadd{i}{j}(F_x)$ is acyclic. It hold by \cref{lem:pat-join-acyclic} that the later is the case if and only if $\acyc_{i,j}(\pat(F_{x'}))$ holds. Moreover, it holds by the same lemma, assuming that $F_x$ is acyclic, that $\pat(F_x) = \patadd_{i,j}(\pat(F_{x'}))$. Hence, the claim follows by the induction hypothesis.

    Finally, for a union node $x$, since $v_0$ is an isolated vertex in $G_x$, it holds that $G_x$ results from the disjoint union of $G_{x_1}\setminus v_0$ and $G_{x_2} \setminus v_0$ by adding $v_0$ as an isolated vertex.
    Hence, every partial solution $F$ at $x$ is the disjoint union of two partial solutions $F_1$ and $F_2$ at $x_1$ and $x_2$ respectively, with $|F| = |F_1| + |F_2| - 1$, as we count $v_0$ twice. Hence, it holds that $|F|-1 = (|F_1|-1) + (|F_2|-1)$. It also holds that $\weightf(F) = \weightf(F_1) + \weightf(F_2)$, since $\weightf(v_0) = 0$. The claim follows by applying the induction hypothesis at $x_1$ and $x_2$, since it holds from \cref{obs:pat-union-relabel} that $\pat(F) = \pat(F_1)\patunion \pat(F_2)$.
\end{proof}

\begin{lemma}\label{lem:alg-counts-ext}
    For each clique extension $\cext$ and each partial solution $F'$ of $G_{\cext}$, let $S_{\cext, F'}\subseteq \CSP$ be the set of all very nice patterns $p$ such that $F_p$ is compatible with $F'$. Then it holds that the number of partial solutions $F$ of $G_x$ of size $\budget + 1$ and weight $\weight$ such that $F\cup F'$ induces a forest in $G_{\cext(\mu_F)}$ is congruent (modulo $2$) to the sum $\sum\limits_{p\in S_{{\cext,F'}}}T_x^{\budget, \weight}[p]$.
\end{lemma}

\begin{proof}
    Let $\tilde{S}_{\cext, F'}$ be the extension of $S_{\cext, F'}$ to $\Pat$, i.e.\ $\tilde{S}_{\cext, F'}=\{p\in\Pat\colon F_p\cmptb_{\cext} F'\}$.
    It holds by \cref{cor:equiv-if-equiv-pat} that $F_p$ is compatible with $F'$ if and only if each partial solution $F$ with $\pat(F)=p$ is compatible with $F'$. Hence, it follows from \cref{lem:alg-counts-sol} that the number of partial solutions of size $\budget +1$ and weight $\weight$ compatible with $F'$ is congruent to the sum
    $\sum_{p\in \tilde{S}_{{\cext,F'}}}\tilde{T}_x^{\budget, \weight}[p]$.
    We claim that for each $x\in\nodes$, for all values $\budget, \weight$, for each clique extension $\cext$, and each partial solution $F'$ of $\cext$ the following congruence holds:
    \[
    \sumstack{p \in \CSP\\F_p \cmptb_{\cext} F'}T_x^{\budget,\weight}[p]
    \bquiv
    \sumstack{p \in \CSP\\F_p \cmptb_{\cext} F'}\tilde{T}_x^{\budget,\weight}[p].
    \]

    We prove the claim by induction over $x\in \nodes$. 
    Let $T=T_x^{\budget, \weight}$ and $\tilde{T} = \tilde{T}_x^{\budget, \weight}$. 
    The claim holds trivially for an introduce node, as $\tilde{T}[p]$ equals $T[p]$ for all patterns $p\in\CSP$, and equals $0$ for all other patterns $p\in \Pat\setminus \CSP$. Now let $x$ be a non-leaf node, and assume that the claim holds for the children of $x$.
    For a relabel node ($\mu_x = \relabel{i}{j}(\mu_{x'})$), 
    let $T_0 = T_{x'}^{\budget, \weight}$, and $\tilde{T}_0 = T_{x'}^{\budget,\weight}$.
    Let $\cext'$ be the clique extension that results from $\cext$ by replacing the variable $x$ with $\relabel{i}{j}(x)$. Then the following must holds:
    \begin{claim}\label{claim:rep-manip-expr}
    It holds for each labeled forest $F_0$ that $F := \relabel{i}{j}(F_0)$ is compatible with $F'$ under $\cext$ if and only if $F_0$ is compatible with $F'$ under $\cext'$.
    \end{claim}
    This follows from the fact that $\cext(F)$ and $\cext'(F_0)$ are the exact same expressions.
    It follows that
    \[
    \sumstack{p \in \CSP\\F_p \cmptb_{\cext} F'}T[p]
    \bquiv
    \sumstack{p' \in \CSP\\F_{p'_{i\rightarrow j}} \cmptb_{\cext} F'}T_0[p']
    \bquiv
    \sumstack{p' \in \CSP\\\relabel{i}{j}(F_{p'}) \cmptb_{\cext} F'}T_0[p']
    \bquiv
    \sumstack{p' \in \CSP\\F_{p'} \cmptb_{\cext'} F'}T_0[p'],
    \]
    where the first congruence holds by the definition of the tables $T^{\budget, \weight}$, the second holds by \cref{obs:pat-union-relabel}, and the last one holds by \cref{claim:rep-manip-expr}.
    Analogously, it holds that
    \[\sumstack{p \in \CSP\\F_p \cmptb_{\cext} F'}\tilde{T}[p]
    \bquiv
    \sumstack{p' \in \CSP\\F_{p'} \cmptb_{\cext'} F'}\tilde{T}_0[p'].\]
    It also holds by induction hypothesis, that
    \[
    \sumstack{p' \in \CSP\\F_{p'} \cmptb_{\cext'} F'}\tilde{T}_0[p']
    \bquiv
    \sumstack{p' \in \CSP\\F_{p'} \cmptb_{\cext'} F'}T_0[p'].
    \]
    It follows that
    \[
    \sumstack{p \in \CSP\\F_p \cmptb_{\cext} F'}T[p]
    \bquiv
    \sumstack{p \in \CSP\\F_p \cmptb_{\cext} F'}\tilde{T}[p].
    \]
        
    For a join node ($\mu_x = \clqadd{i}{j}(\mu_{x'})$),
    let $T_0 = T_{x'}^{\budget, \weight}$, and $\tilde{T}_0 = T_{x'}^{\budget,\weight}$. We define the table $\hat{T}\colon \Pat\rightarrow\bin$, where for $p\in\Pat$ we define
    $\hat{T}[p] = \sum\limits_{\patadd_{i,j}(p')=p} T_0[p']$.
    Let $\cext'$ be the clique extension that results from $\cext$ by replacing the variable $x$ with $\clqadd{i}{j}(x)$. Again, it holds for a labeled forest $F_0$ that $F := \clqadd{i}{j}(F_0)$ is compatible with $F'$ under $\cext$ if and only if $F_0$ is compatible with $F'$ under $\cext'$. If follows, in an analogous way to a relabel node, and by using\cref{lem:pat-join-acyclic}, that
    \[
    \sumstack{p \in \CSP\\F_p \cmptb_{\cext} F'}\hat{T}[p]
    \bquiv
    \sumstack{p' \in \CSP\\F_{\patadd_{i, j}(p')} \cmptb_{\cext} F'}T_0[p']
    \bquiv
    \sumstack{p' \in \CSP\\\clqadd{i}{j}(F_{p'}) \cmptb_{\cext} F'}T_0[p']
    \bquiv
    \sumstack{p' \in \CSP\\F_{p'} \cmptb_{\cext'} F'}T_0[p'].
    \]
    Analogously, it holds that
    \[\sumstack{p \in \CSP\\F_p \cmptb_{\cext} F'}\tilde{T}[p]
    \bquiv
    \sumstack{p' \in \CSP\\F_{p'} \cmptb_{\cext'} F'}\tilde{T}_0[p'].\]
    It also holds by induction hypothesis, that
    \[
    \sumstack{p' \in \CSP\\F_{p'} \cmptb_{\cext'} F'}\tilde{T}_0[p']
    \bquiv
    \sumstack{p' \in \CSP\\F_{p'} \cmptb_{\cext'} F'}T_0[p'].
    \]
    It follows that
    \[
    \sumstack{p \in \CSP\\F_p \cmptb_{\cext} F'}\hat{T}[p]
    \bquiv
    \sumstack{p \in \CSP\\F_p \cmptb_{\cext} F'}\tilde{T}[p].
    \]
    On the other hand, it holds by \cref{cor:reduce} that the support of $T$ represents the support of $\hat{T}$, and hence, that 
    \[
    \sumstack{p \in \CSP\\F_p \cmptb_{\cext} F'}T[p]
    \bquiv
    \sumstack{p \in \CSP\\F_p \cmptb_{\cext} F'}\hat{T}[p].
    \]
    Hence, it holds that
    \[
    \sumstack{p \in \CSP\\F_p \cmptb_{\cext} F'}T[p]
    \bquiv
    \sumstack{p \in \CSP\\F_p \cmptb_{\cext} F'}\tilde{T}[p].
    \]

    Finally, for union node, let $x_1$, $x_2$ be the children of $x$.
    We define $\tau_p(x) := \tau(x \clqunion F_{p})$. Clearly, it holds that $F\clqunion F_p \cmptb_{\tau} F'$ if and only if $F\cmptb_{\tau_p} F'$.
    It follows that
    \begin{align*}
        \sumstack{p\in\CSP\\F_p\cmptb F'}T[p]
        &\bquiv \sumstack{\budget_1+\budget_2 =\budget\\\weight_1+\weight_2=\weight}
        \sumstack{p_1,p_2\in\CSP\\F_{p_1 \patunion p_2}\cmptb F'}
        T_{x_1}^{\budget_1,\weight_1}[p_1]\cdot
        T_{x_2}^{\budget_2,\weight_2}[p_2]
        \\
        &\bquiv \sumstack{\budget_1+\budget_2 =\budget\\\weight_1+\weight_2=\weight}
        \sumstack{p_2\in\CSP}T_{x_2}^{\budget_2,\weight_2}[p_2]
        \sumstack{p_1\in\CSP\\F_{p_1\patunion p_2} \cmptb F'}
        T_{x_1}^{\budget_1,\weight_1}
        \\
        &\bquiv \sumstack{\budget_1+\budget_2 =\budget\\\weight_1+\weight_2=\weight}
        \sumstack{p_2\in\CSP}T_{x_2}^{\budget_2,\weight_2}[p_2]
        \sumstack{p_1\in\CSP\\F_{p_1}\clqunion F_{p_2} \cmptb F'}
        T_{x_1}^{\budget_1,\weight_1}
        \\
        &\bquiv \sumstack{\budget_1+\budget_2 =\budget\\\weight_1+\weight_2=\weight}
        \sumstack{p_2\in\CSP}T_{x_2}^{\budget_2,\weight_2}[p_2]
        \sumstack{p_1\in\CSP\\F_{p_1}\cmptb_{\tau_{p_2}} F'}
        T_{x_1}^{\budget_1,\weight_1}
        \\
        &\bquiv\sumstack{\budget_1+\budget_2 =\budget\\\weight_1+\weight_2=\weight}
        \sumstack{p_2\in\CSP}T_{x_2}^{\budget_2,\weight_2}[p_2]
        \sumstack{p_1\in\Pat\\F_{p_1}\cmptb_{\tau_{p_2}} F'}
        \tilde{T}_{x_1}^{\budget_1,\weight_1}
        \\
        &\bquiv\sumstack{\budget_1+\budget_2 =\budget\\\weight_1+\weight_2=\weight}
        \sumstack{p_2\in\CSP}T_{x_2}^{\budget_2,\weight_2}[p_2]
        \sumstack{p_1\in\CSP\\F_{p_1}\clqunion F_{p_2} \cmptb F'}
        \tilde{T}_{x_1}^{\budget_1,\weight_1}
        \\
        &\bquiv\sumstack{\budget_1+\budget_2 =\budget\\\weight_1+\weight_2=\weight}
        \sumstack{p_1\in\CSP}\tilde{T}_{x_1}^{\budget_1,\weight_1}[p_1]
        \sumstack{p_2\in\CSP\\F_{p_1}\clqunion F_{p_2} \cmptb F'}
        T_{x_2}^{\budget_2,\weight_2}
        \\
        &\bquiv\sumstack{\budget_1+\budget_2 =\budget\\\weight_1+\weight_2=\weight}
        \sumstack{p_1\in\CSP}\tilde{T}_{x_1}^{\budget_1,\weight_1}[p_1]
        \sumstack{p_2\in\CSP\\F_{p_1} \cmptb_{\tau_{p_1}} F'}
        T_{x_2}^{\budget_2,\weight_2}
        \\
        &\bquiv\sumstack{\budget_1+\budget_2 =\budget\\\weight_1+\weight_2=\weight}
        \sumstack{p_1\in\CSP}\tilde{T}_{x_1}^{\budget_1,\weight_1}[p_1]
        \sumstack{p_2\in\CSP\\F_{p_1} \cmptb_{\tau_{p_1}} F'}
        \tilde{T}_{x_2}^{\budget_2,\weight_2}
        \\
        &\bquiv \sumstack{\budget_1+\budget_2 =\budget\\\weight_1+\weight_2=\weight}
        \sumstack{p_1,p_2\in\Pat\\F_{p_1}\clqunion F_{p_2} \cmptb F'}
        \tilde{T}_{x_1}^{\budget_1,\weight_1}[p_1]\cdot
        \tilde{T}_{x_2}^{\budget_2,\weight_2}[p_2]
        \\
        &\bquiv \sumstack{\budget_1+\budget_2 =\budget\\\weight_1+\weight_2=\weight}
        \sumstack{p_1,p_2\in\Pat\\F_{p_1 \patunion p_2} \cmptb F'}
        \tilde{T}_{x_1}^{\budget_1,\weight_1}[p_1]\cdot
        \tilde{T}_{x_2}^{\budget_2,\weight_2}[p_2]
        \bquiv \sumstack{p\in\Pat\\F_p\cmptb F'}\tilde{T}[p],
    \end{align*}
    where the first and the last congruences hold by the definitions of the tables $T$ and $\tilde{T}$ respectively, and the third and the second to last by \cref{obs:pat-union-relabel}.
\end{proof}

\begin{corollary}\label{cor:alg-counts-fvs}
    The graph $G$ admits an odd number of feedback vertex sets of size $\budget$ and weight $\weight$ if and only if it holds that $\sum\limits_{p\in \CSP} T_r^{n-\budget,\weightf(V)-\weight}[p] \bquiv 1$.
\end{corollary}

\begin{proof}
    It holds that the number of feedback vertex sets of size $\budget$ and weight $\weight$ in $G$ is equal to the number of partial solutions $F$ in $G_r$ of size $n + 1 - \budget$ and weight $\weightf(V) - \weight$.
    The corollary follows then from \cref{lem:alg-counts-ext} by observing that each partial solution is compatible with the partial solution $\{v_0\}$, and hence, the number of these solutions is congruent (modulo $2$) to the sum over all patterns in $S_{\cext, \emptyset} = \CSP$, where $\cext(x) = x$.
\end{proof}

\subsection{Running time}

Now we bound the running time of the algorithm.
First we show how to efficiently process a union node $x\in\nodes$ using a fast convolution technique, assuming that all tables $T_{x'}^{\budget,\weight}$ are given for all children $x'$ of $x$. We start by formally introducing the notion of convolution, and presenting a specific binary operation $\statejoin$ over the set of states (\cref{def:csp-states-bijection}), whose convolution over $\bin$ intuitively corresponds to the union formula of our algorithm.

\begin{definition}
    We define the binary operation $\simplejoin:S\times S\rightarrow S$ given by the following table:
    \begin{center}
    \begin{tabular}{|c|c|c|c|c|c|c|}
    \hline
    \simplejoin &\stnone&\stdisc&\stddisc&\stconn&\stplus&\stdconn\\
    \hline
    \stnone&\stnone&\stdisc&\stddisc&\stconn&\stplus&\stdconn\\
    \hline
    \stdisc&\stdisc&\stddisc&\stddisc&\stplus&\stplus&\stdconn\\
    \hline
    \stddisc&\stddisc&\stddisc&\stddisc&\stplus&\stplus&\stdconn\\
    \hline
    \stconn&\stconn&\stplus&\stplus&\stdconn&\stdconn&\stdconn\\
    \hline
    \stplus&\stplus&\stplus&\stplus&\stdconn&\stdconn&\stdconn\\
    \hline
    \stdconn&\stdconn&\stdconn&\stdconn&\stdconn&\stdconn&\stdconn\\
    \hline
\end{tabular}
\end{center}

We define the binary operation $\statejoin$ over $\mathcal{S}$ as the $k$th power of $\simplejoin$.
\end{definition}

The following result follows from observing that the operation $\patunion$ over $\CSP$ acts locally on each label, and that for every single label it coincides with the operation $\simplejoin$.

\begin{observation}\label{obs:csp-states-iso}
    The bijection $\phi_{CS}$ defined in \cref{def:csp-states-bijection} is an isomorphism between the operation $\patunion$ restricted to $\CSP$ and the operation $\statejoin$ over $\states$. That is, for all $p_1, p_2 \in \CSP$, it holds that
    $\phi_{CS}(p_1 \patunion p_2) = \phi_{CS}(p_1) \statejoin \phi_{\CSP}(p_2)$.
\end{observation}

Now we show that the convolution over this operation corresponds to the recursive formula defined at a union node in our algorithm. 

\begin{lemma}\label{lem:conv-implies-union}
    If the convolution $\stateconv$ of tables in $\states^{\bin}$ can be computed in time $f$, then all tables $T_x^{\budget,\weight}$ for all values $\budget$ and $\weight$ at a union node $x\in \nodes$ can be computed in time $\ostar(f + 6^k)$.
\end{lemma}

\begin{proof}
    In order to compute the tables $T_x^{\budget, \weight}$ at a union node (\cref{eq:alg-union} of \cref{algorithm}), we iterate over all values $\budget, \weight$. For each, we iterate over all values $\budget_1,\budget_2, \weight_1, \weight_2$ with $\budget_1 + \budget_2 = \budget$ and $\weight_1 + \weight_2 = \weight$, and we compute the convolution of $T_{x_1}^{\budget_1,\weight_1}$ and $T_{x_2}^{\budget_2,\weight_2}$ over $\patunion$. We add the resulting table to $T_x^{\budget,\weight}$. The correctness follows directly by the definition of the tables $T_x^{\budget,\weight}$.
    
    It holds by \cref{obs:csp-states-iso}, that $\phi_{CS}$ is an isomorphism between the operation $\patunion$ restricted to $\CSP$ and the operation $\statejoin$ over $\states$. Hence, it follows by \cref{obs:iso-same-time} that the convolution of $T_{x_1}^{\budget_1,\weight_1}$ and $T_{x_2}^{\budget_2,\weight_2}$ over $\patunion$ can be computed in time $\ostar(f+6^k)$. Since we only iterate over a polynomial number of values of $\budget$, $\weight$, $\budget_1$, $\budget_2$, $\weight_1$, and $\weight_2$, we get a total running time of $\ostar(f + 6^k)$ to compute all tables $T_x^{\budget,\weight}$ at the union node $x$.
\end{proof}

At the end of this section we show that this convolution can be computed in time $\ostar(6^k)$ indeed proving the following lemma:

\begin{lemma}\label{lem:conv-time}
    The convolution $\stateconv$ of two tables in $\bin^{\states}$ can be computed in time $\ostar(6^k)$.
\end{lemma}

Now we bound the running time of the entire algorithm by proving the following lemma:

\begin{lemma}\label{lem:time-nodes}
All tables $T_x^{\budget,\weight}$ for $x\in\nodes$ and all values $\budget, \weight$ can be computed in time $\ostar(6^k)$.
\end{lemma}

\begin{proof}
    We achieve this by computing the tables $T_x^{\budget, \weight}$ as defined in \cref{algorithm} by a bottom up dynamic programming over $\syntaxtree$. Therefore, we can assume for each node $x\in\nodes$, that all tables $T_{x'}^{\budget, \weight}$ are already given for all children $x'$ of $x$. Since we assume that $|\nodes|$ is bounded polynomially in $n$, it suffices to prove this bound individually for each node $x\in\nodes$.
    We start by bounding the running time for introduce, relabel and join nodes.

    For an introduce node $x$, we iterate over all values $\budget, \weight$ and all patterns $p\in\CSP$ and set the value of $T_x^{\budget,\weight}[p]$ individually.
    For a relabel node $\mu_x = \relabel{i}{j}(\mu_{x'})$, we iterate over all patterns $p'\in\CSP$ and all values $\budget,\weight$, we compute the pattern $p = p'_{i\rightarrow j}$, and we add $T_x'^{\budget,\weight}[p']$ to $T_x^{\budget,\weight}[p]$.
    For a join node $\mu_x = \clqadd{i}{j}(\mu_{x'})$, we iterate over all patterns $p''\in\Pat$ and all $\budget,\weight$, and compute the pattern $p' = \patadd_{i,j}(p'')$. Then we add $T_{x''}^{\budget,\weight}[p'']$ to all entries $T_x^{\budget,\weight}[p]$ for all $p\in \redpat(p')$.

    In total, for each node $x$ we iterate over a polynomial number of values $\budget, \weight$, and it holds by \cref{obs:csp-states} that $|\CSP| = 6^k$. Finally, it holds by \cref{lem:pat-join-time} that $|\redpat(p')| \leq 7$, and that $\redpat(p')$ can be computed in polynomial time for a pattern $p'$ that results from a very nice patterns $p''\in\CSP$ by a join operation. Hence, all tables $T_x^{\budget,\weight}$ at introduce, relabel or join nodes can be computed in time $\ostar(6^k)$.
    Finally, for a union node $x\in\nodes$, the running time follows by inserting the bound $\ostar(6^k)$ from \cref{lem:conv-time} into \cref{lem:conv-implies-union}.
\end{proof}

Before we prove \cref{lem:conv-time}, let us show how this already implies the claimed algorithms in \cref{theo:count-cw} and \cref{theo:ub}.

\begin{proof}[Proof of~\cref{theo:count-cw}]
    We fixes $\W = 1$ and $\weightf(v) = 1$ for all $v\in V$. Hence a feedback vertex set of size $\budget$ has weight $\budget$ as well. The algorithm computes all tables $T_x^{\budget,\weight}$ for each $x\in\nodes$ and all values $\budget, \weight$ as defined in \cref{algorithm}, and outputs the value $s:= \sum_{p\in \CSP} T_r^{n-\target,n-\target}[p]$, where $r$ is the root node of $\syntaxtree$.
    It holds by \cref{cor:alg-counts-fvs} that $s$ is equal to the number of feedback vertex sets of size $\target$ in $G$ modulo $2$. It follows from \cref{lem:time-nodes} that this algorithm runs in time $\ostar(6^k)$.
\end{proof}

\subsection{Isolation Lemma}\label{sec:isolation}

So far, we have shown that the algorithm counts (modulo $2$) the number of solutions of size $\budget$ and weight $\weight$. Now we reduce the decision version of \Fvsp to the counting (modulo $2$) version using the isolation lemma~\cite{DBLP:journals/combinatorica/MulmuleyVV87}. We show that by a careful choice of $\W$, and by choosing the weight function $\weightf$ uniformly at random, we can isolate a unique minimum weight solution of a specific size $\target$ with high probability, if such a solution exists.

\begin{definition}\label{def:isolation}
    A weight function $\weightf:U\rightarrow \mathbb{Z}$ \emph{isolates} a set family $\mathcal{F}\subseteq 2^U$, if there exists a unique $S'\subseteq \mathcal{F}$ with $\weightf(S') = \min_{S\in\mathcal{F}}\weightf(S)$.
\end{definition}

\begin{lemma}[\cite{DBLP:journals/combinatorica/MulmuleyVV87}]\label{lem:iso}
    Let $\mathcal{F}\subseteq 2^U$ be a set family over a universe $U$ with $|\mathcal{F}|>0$, and let $N>|U|$ be an integer. For each $u\in U$, choose a weight $\weightf(u)\in \{1,2,\dots N\}$ uniformly and independently at random. Then it holds that $\pr[\weightf \text{ isolates } \mathcal{F}]\geq 1-|U|/N$.
\end{lemma}

\begin{lemma}\label{lem:iso-imply-single-sol}
    Assume that $G$ contains a feedback vertex set of size $\target$.
    Let $\W = 2 |V|$, and let us fix $\weightf$ by choosing $\weightf(v)$ independently and uniformly at random in $[\W]$ for each vertex $v \in V$. Let $Q$ be a minimum-weight feedback vertex set of $G$ of size $\target$. Then $Q$ is unique with probability at least $1/2$.
\end{lemma}

\begin{proof}
    Let $\mathcal{Q}\subseteq 2^V$ be the family of all feedback vertex sets of size $\target$ in $G$. Since we choose $\weightf(v)$ independently and uniformly at random for each $v\in V$, it follows by \cref{lem:iso} that $\weightf$ isolates $\mathcal{Q}$ with probability at least $1 - |V| / \W = 1 - |V| / (2|V|) = 1/2$.
\end{proof}

\begin{proof}[Proof of~\cref{theo:ub}]
    The algorithm first fixes $\W = 2\cdot |V|$ as chosen in \cref{lem:iso-imply-single-sol} and fixes $\weightf$ by choosing the weight of each vertex independently and uniformly at random in $[\W]$. Then it computes all tables $T_x^{\budget,\weight}$ for each $x\in\nodes$ and all values $\budget, \weight$ as defined in \cref{algorithm}. The algorithm then accepts if $s_{\weight} := \sum_{p\in \CSP} T_r^{\target,\weight}[p] \bquiv 1$ for any value $\weight$ where $r$ is the root node of $\syntaxtree$, or it rejects otherwise.

    If $G$ does not admit a feedback vertex set of size $\target$, then it follows from \cref{cor:alg-counts-fvs} that the sum $s_{\weight}$ is even for all values $\weight$, and hence the algorithm rejects. Otherwise, we assume that $G$ admits a feedback vertex set of size $\target$, and let $\weight_0$ be the minimum weight of such a set.
    Then it follows from \cref{lem:iso-imply-single-sol} that there exists a unique feedback vertex set of size $\target$ and weight $\weight_0$ with probability at least $1/2$. It follows in this case that $s_{\weight_0} \bquiv 1$ and hence, the algorithm accepts. The running time of the algorithm follows from \cref{lem:time-nodes}.
\end{proof}

\subsection{Fast union operation}

All that is left is to prove \cref{lem:conv-time}, i.e., to show that the convolution over $\statejoin$ can be computed in time $\ostar(6^k)$. In order to do so, we apply a multilayered version of the fast Zeta transform~\cite{DBLP:conf/stoc/BjorklundHKK07} combined with the ``Count and Filter'' technique of van Rooij~\cite{DBLP:conf/birthday/Rooij20}. We start with formal definitions:

\begin{definition}
    Let $(S, \leq)$ be a partially ordered set, and $\mathbb{F}$ be some ring. We define the Zeta transformation $\zeta_{\leq}\colon \mathbb{F}^S \rightarrow \mathbb{F}^S$ over $\leq$, where for $f: S\rightarrow \mathbb{F}$ we define
    $\zeta_{\leq}f(a) = \sum_{b\leq a} f(b)$.
    Similarly, we define the Mobius transformation $\mu_{\leq}$ over $\leq$ as 
    $\mu_{\leq}f(a) = \sum_{b\leq a} \mu_{\leq}(b,a) f(a)$,
    where $\mu_{\leq}(b,a) = 1$ if $b=a$, or $\sum_{b < c \leq a}\mu(c,a)$.
\end{definition}

We cite the following results from van Rooij~\cite{DBLP:conf/birthday/Rooij20}:

\begin{lemma}[{\cite[Proposition 5]{DBLP:conf/birthday/Rooij20}}]
    \label{lem:mobius-zeta-time}
    Let $\leq$ be a join-semilattice over a finite set $S$ of size $L$, and let $\preceq$ be the $k$th power of $\leq$. Then both Zeta and Mobius transformations over $\preceq$ can be computed in time $\ostar(L^k)$.
\end{lemma}

\begin{lemma}[{\cite[Proof of Lemma 4]{DBLP:conf/birthday/Rooij20}}]
    \label{lem:mobius-zero-one}
    It holds for all $x,y \in S$ that $\sum_{x\leq z\leq y} \mu_{\leq}(z,y) = 1$ if $x=y$ and $0$ otherwise.
\end{lemma}

\begin{corollary}[{\cite[Lemma 4]{DBLP:conf/birthday/Rooij20}}]
    \label{cor:mobius-zeta}
    It holds for each mapping $f\colon S\rightarrow \mathbb{F}$ that $\mu(\zeta(f)) = f$.
\end{corollary}

Now we define a lattice $\preceq$ that reflects the ordering structure of our union operation.

\begin{definition}
    We define the lattice $\leq$ over $S_0$ as given by the following Hasse diagram:

    \begin{center}
    \begin{tikzpicture}[y=.6cm, x=.6cm]
        \pic{ord-o};
    \end{tikzpicture}
    \end{center}

    Let $\preceq$ be the $k$th power of this lattice.
\end{definition}

Note that the operation $\lor_{\leq}$ is quite close to the join operation $\simplejoin$, as they only differ on the pair $(\stdisc, \stdisc)$ and the pairs in $\{\stconn, \stplus\}\times\{\stconn, \stplus\}$. The reason for this is is that lattices' joins fall short of capturing the multiplicity of specific states without introducing additional states. Therefore, we will aim to decompose this lattice into three separate lattices, such that, if applied in the right order, one can filter out ``bad'' combinations between them, correctly simulating the join operation $\statejoin$.

\begin{definition}\label{def:three-lattices}
    We define the orderings $\leq_1, \leq_2, \leq_3$ over $S_0$ given by the following Hasse diagrams:
    \begin{center}
    \begin{minipage}[t]{0.3\textwidth}
    \centering
    \begin{tikzpicture}[y=.6cm, x=.6cm]
        \pic{ord-1};
    \end{tikzpicture}
    \end{minipage}%
    ~
    \begin{minipage}[t]{0.3\textwidth}
    \centering
    \begin{tikzpicture}[y=.6cm, x=.6cm]
        \pic{ord-2};
    \end{tikzpicture}
    \end{minipage}%
    ~
    \begin{minipage}[t]{0.3\textwidth}
    \centering
    \begin{tikzpicture}[y=.6cm, x=.6cm]
        \pic{ord-3};
    \end{tikzpicture}
    \end{minipage}%
    \end{center}

    We denote by $\preceq_{1}, \preceq_{2}, \preceq_{3}$ the $k$th powers of the lattices $\leq_1, \leq_2, \leq_3$ respectively. For $i\in\{1,2,3\}$, let $\lor_i = \lor_{\preceq_i}$, $\zeta_i := \zeta_{\preceq_i}$ and $\mu_i := \mu_{\preceq_i}$ be the join, Zeta and Mobius transformations of these orderings respectively.
\end{definition}

Now we are ready to present the main part of this section, where we show that by applying the Zeta transformations over these three orderings in the right order, and by correctly filtering these transformations, we transform the convolution problem into a pointwise product of two tables. The crux of this proof is to show that by transforming this product back in the reverse order of these transformations, we get exactly the convolution of the two original tables.

\begin{definition}
    We define the sets $\classC := \{\stconn, \stplus\}$, $\classDC := \{\stconn^*\}$ and $\classD := \{\stdisc\}$.
    We also define $\classA := S_0\setminus \{\stdisc, \stdconn\}$.
    Given a mapping $f\colon \states\rightarrow \bin$ and a set $X\subseteq S_0$ (mainly one of the sets $\classC, \classDC, \classD, \classA$), we define the mapping $f|^X_i$, where for $x\in \states$ we define
    \[f|^X_i(x) := f(x) \cdot \big[|f^{-1}(X)|=i\big].\]
    Let $f^{(1)} = \zeta_1 f$, $f^{(1)}_i = f^{(1)}|^{\classC}_i$, $f^{(2)}_i = \zeta_2 f^{(1)}_i$, $f^{(2)}_{i,j} = f^{(2)}_i|^{\classD}_j$ and $f^{(3)}_{i,j} = \zeta_3 f^{(2)}_{i,j}$.
\end{definition}

Intuitively, we will use $\classC$ and $\classD$ to bound the number of times these states appear on both sides of a join operation, while the partition of $S_0$ into $\classDC, \classA, \classD$ corresponds to the decomposition of $\leq$ into the orderings $\leq_1, \leq_2$ and $\leq_3$ respectively.

Before we proceed with the next lemma, we present some more notation and provide some intuition that aids understanding.

\begin{definition}
    Given a state $x\in\states$, we define $C_x = x^{-1}(\classC)$, $D_x = x^{-1}(\classD)$, $C^*_x = x^{-1}(\classDC)$, and $A_x = x^{-1}(\classA)$ as the sets of indices assigned the corresponding states.
    Given a state $y\in\states$ and a set $X\subseteq S_0$, we define $y_X$ as the restriction of $y$ to the preimage of $X$.
\end{definition}

\begin{observation}
    Let $X$ be either $\classC$ or $\classD$, and $y,z\in \states$ be two states. Then $X_z$ are the set of indices of $z$ assigned a state in $X$, and $y_{X_z}$ is the restriction of $y$ to these indices. Since $X_{y_{X_z}}$ is the set of indices of $y_{X_z}$, assigned states in $X$, it follows that $X_{y_{X_z}} = X_y \cap X_x$.
\end{observation}

\begin{lemma}\label{lem:conv-zeta-form}
    It holds for all $x\in\states$ that
    $f^{(3)}_{i,j}(x) = \sum\limits_{y\preceq x} \Big[\big|C_{y_{C_x}}\big| =i\Big]\Big[\big|D_{y_{D_x}}\big| = j\Big] f(y)$.
\end{lemma}

\begin{proof}
    It holds that $f^{(1)}(x) = \sum\limits_{y\preceq_1 x}f(y)$, i.e.\ $f(y)$ appears in the sum, if $y$ only differs from $x$ at positions in $C^*_x$. Therefore, it holds that
    \[
    f^{(2)}_i(x) = 
    \sum_{\substack{
    y_{C^*_x}\preceq_1 x_{C^*_x}, y_{A_x} \preceq_2 x_{A_x}\\
    y_{D_x} = x_{D_x}}} 
    \big[\big|C_{y_{C_x}}\big| =i\big] f(y),
    \]
    and hence,
    \[
    f^{(2)}_{i,j}(x) = 
    \sum_{\substack{
    y_{C^*_x}\preceq_1 x_{C^*_x}, y_{A_x} \preceq_2 x_{A_x}\\
    y_{D_x} = x_{D_x}}} 
    \big[\big|C_{y_{C_x}}\big| =i\big]
    \big[\big|D_{y_{D_x}}\big| =j\big] f(y).
    \]
    The claim follows from the fact that $x\preceq y$ holds if and only if $y_{C^*_x}\preceq_1 x_{C^*_x}$, $y_{A_x} \preceq_2 x_{A_x}$ and $y_{D_x} \preceq_3 x_{D_x}$ hold.
\end{proof}

For two mappings $f,g\colon \states\rightarrow \bin$, we denote by $f\cdot g$ the pointwise product of the two mappings, i.e.\ $\big(f\cdot g\big)(x) = f(x) \cdot g(x)$ for all $x\in\states$.
From now on, we fix two arbitrary mappings $f,g\colon \states\rightarrow \bin$, and an arbitrary state $x\in \states$. For $i,j\in[k]$, we define
\[h^{(3)}_{i,j} = \sumstack{i_1+i_2=i\\j_1+j_2=j} f^{(3)}_{i_1,j_1} \cdot g^{(3)}_{i_2,j_2}.\]
Now we show that by reversing the sequence of applications of $\zeta$ on the resulting product, one obtains exactly the convolution $\stateconv$.
We start by defining this reversing process.

\begin{definition}\label{def:reverse-conv}
    Given a mapping $h^{(3)}_{i,j}\colon \states\rightarrow \bin$, we define $h^{(2)}_{i,j} := \mu_3 h^{(3)}_{i,j}$. We define $h^{(2)}_i$ for each values $i\in[k]$, where $h^{(2)}_i(z) := h^{(2)}_{i,j}(z)$ for $j=\big|D_z\big|$. We define $h^{(1)}_i := \mu_2 h^{(2)}_i$ and $h^{(1)}$, where we define $h^{(1)}(z) := h^{(1)}_{i}(z)$ for $i=\big|C_z\big|$. Finally, we define $h = \mu_1 h^{(1)}$.
\end{definition}

The following observation is a direct consequence of \cref{lem:conv-zeta-form}.

\begin{observation}\label{obs:conv-prod-form}
    It holds that 
    \[h^{(3)}_{i,j}(x) = \sum\limits_{y,z\preceq x} 
    \Big[\big|C_{y_{C_x}}\big|+\big|C_{z_{C_x}}\big| = i\Big]
    \Big[\big|D_{y_{D_x}}\big|+\big|D_{z_{D_x}}\big| = j\Big] 
     f(y)\cdot g(z).\]
\end{observation}

\begin{lemma}\label{lem:conv-h2-form}
    It holds that
    \[
    h^{(2)}_i(x) = \sumstack{y,z\preceq x\\y_D\statejoin z_D = x_D} \Big[\big|C_{y_{C_x}}\big| + \big|C_{z_{C_x}}\big| =i\Big] f(y)\cdot g(z).
    \]
\end{lemma}

\begin{proof}
    In order to prove the claim we first show the correctness of following equality:
    It holds for $z\preceq y \preceq_3 x$, and $X\in\{\classC,\classD\}$ that
    \begin{equation}\label{eq:conv-h2-y-implies-x}
        z^{-1}(X)\cap y^{-1}(X)= z^{-1}(X) \cap x^{-1}(X).
    \end{equation}

    \noindent ``$\subseteq$'': It follows from $y\preceq_3 x$ that $y^{-1}(X) \subseteq x^{-1}(X)$, since all states $\stdisc, \stconn$ and $\stplus$ are maximal in the ordering $\preceq_3$. Hence, $z^{-1}(X)\cap y^{-1}(X) \subseteq z^{-1}(X) \cap x^{-1}(X)$.

    \noindent ``$\supseteq$'': For $X=\classC$, the claim follows from $y\preceq_3 x$, as the elements of $\classC$ are minimal in the ordering $\preceq_3$. For $X=\classD$, it holds for $\ell\in x^{-1}(\classD) \cap z^{-1}(\classD)$ that $y_{\ell} \in\{\stnone,\stdisc\}$. But then it follows from $z\preceq y$ and from $z_{\ell} = \stdisc$ that $y_{\ell} = \stdisc$. Hence, $z^{-1}(X)\cap x^{-1}(X)\subseteq y^{-1}(X)$. The claim follows.

    Note that \cref{eq:conv-h2-y-implies-x} can also be written as
    \begin{equation}\label{eq:conv-h2-y-implies-x-2}
        D_{z_{D_y}} = D_{z_D} \quad \text{and} \quad C_{z_{C_y}} = C_{z_C}.
    \end{equation}

    Now we are ready to prove the lemma. It holds that
    \begin{align*}
        &h^{(2)}_{i,j}(x) = \mu_3 h^{(3)}_{i,j}(x)\\
        &= \sum\limits_{y\preceq_3 x} \mu_3(y, x)
        \sum\limits_{z_1, z_2 \preceq y}
        \Big[ \big|C_{(z_1)_{C_y}}\big| + \big|C_{(z_2)_{C_y}}\big| = i \Big] 
        \Big[ \big|D_{(z_1)_{D_y}}\big| + \big|D_{(z_2)_{D_y}}\big| = j \Big]
        f(z_1) g(z_2) \\
        &= \sum\limits_{y\preceq_3 x} \mu_3(y, x)
        \sum\limits_{z_1\lor z_2 = y' \preceq y}
        \Big[ \big|C_{(z_1)_{C}}\big| + \big|C_{(z_2)_{C}}\big| = i \Big] 
        \Big[ \big|D_{(z_1)_{D}}\big| + \big|D_{(z_2)_{D}}\big| = j \Big]
        f(z_1) g(z_2) \\
        &\begin{aligned}
        = \sumstack{y'\preceq x\\ z_1\lor z_2 = y'}
        \Big[ \big|C_{(z_1)_{C}}\big| + \big|C_{(z_2)_{C}}\big| = i \Big] 
        &\Big[ \big|D_{(z_1)_{D}}\big| + \big|D_{(z_2)_{D}}\big| = j \Big]\\
        &\cdot f(z_1) \cdot g(z_2)
        \sum\limits_{y'\preceq y \preceq_3 x}\mu_3(y, x)
        \end{aligned}
        \\
        &\begin{aligned}
        =\sumstack{y'\preceq x\\ z_1\lor z_2 = y'}
        \Big[ \big|C_{(z_1)_{C}}\big| + \big|C_{(z_2)_{C}}\big| = i \Big] 
        &\Big[ \big|D_{(z_1)_{D}}\big| + \big|D_{(z_2)_{D}}\big| = j \Big]\\
        &\cdot f(z_1) \cdot g(z_2) \cdot
        \sumstack{y'_{C^*}\preceq y_{C^*} = x_{C^*}\\y'_{A}\preceq y_{A}=x_{A}}\sum\limits_{y'_{D}\preceq_3 y_{D} \preceq_3 x_{D}}\mu_3(y_D, x_D)
        \end{aligned}
        \\
        &= \sumstack{y'\preceq x\\ y'_D = x_D}\sum\limits_{z_1\lor z_2 = y'}
        \Big[ \big|C_{(z_1)_{C}}\big| + \big|C_{(z_2)_{C}}\big| = i \Big] 
        \Big[ \big|D_{(z_1)_{D}}\big| + \big|D_{(z_2)_{D}}\big| = j \Big]
        f(z_1) \cdot g(z_2)
        \\
        &= \sumstack{z_1,z_2\preceq x\\(z_1)_D\lor_3 (z_2)_D = x_D}
        \Big[ \big|C_{(z_1)_{C}}\big| + \big|C_{(z_2)_{C}}\big| = i \Big] 
        \Big[ \big|D_{(z_1)_{D}}\big| + \big|D_{(z_2)_{D}}\big| = j \Big]
        f(z_1) \cdot g(z_2).
    \end{align*}
    The third equality holds by \cref{eq:conv-h2-y-implies-x-2}.
    In the fifth equality, we partition $y',y$ and $x$ 
    into their restrictions to the sets $C^*, A, D$, since $y_{C^*} = x_{C^*}$ and $y_A = x_A$ must hold for $y\preceq_3 x$.
    The equality then holds by noting that $\mu_3(y, x) = \mu_3(y_D, x_D)$ for $y\preceq_3 x$. The last equality holds by \cref{lem:mobius-zero-one}.

    Since it holds that $h^{(2)}_i(x) = h^{(2)}_{i,j}(x)$ for $j=x^{-1}(\classD)$, it holds that
    \begin{equation}
        h^{(2)}_i(x) = 
        \sumstack{z_1,z_2\preceq x\\(z_1)_D\lor_3 (z_2)_D = x_D}
        \Big[ \big|C_{(z_1)_{C}}\big| + \big|C_{(z_2)_{C}}\big| = i \Big] 
        \Big[ \big|D_{(z_1)_{D}}\big| + \big|D_{(z_2)_{D}}\big| = |D_x| \Big]
        f(z_1) \cdot g(z_2).
    \end{equation}
    We claim that $(z_1)_D\lor_3 (z_2)_D = x_D$ and $|D_{(z_1)_D}| + |D_{(z_2)_D}| = |D_x|$ hold, if and only if it holds that $(z_1)_D\statejoin (z_2)_D = x_D$, as the only way for $(z_1)_{\ell}\lor (z_2)_{\ell} = \stdisc$ but $(z_1)_{\ell}\simplejoin (z_2)_{\ell} \neq \stdisc$ is if $(z_1)_{\ell} = (z_2)_{\ell} = \stdisc$. But then there must exist another index $\ell'\in D$ with both $(z_1)_{\ell'}$ and $(z_2)_{\ell'}$ different from $\stnone$, which contradicts the assumption that $(z_1)_D\lor_3 (z_2)_D = x_D$. Hence, it must hold that $(z_1)_D\statejoin (z_2)_D = x_D$ and the lemma follows.
\end{proof}

\begin{lemma}\label{lem:conv-h1-form}
    It holds that
    \[
    h^{(1)}(x) = \sumstack{y,z\preceq x\\ y_D\statejoin z_D = x_D\\ y_A \statejoin z_A = x_A}f(y)\cdot g(z).
    \]
\end{lemma}

\begin{proof}
    We first prove the following equality: It holds for all $z\preceq y \preceq_2 x$ that
    \begin{equation}\label{eq:conv-h1-y-implies-x}
        \quad C_{z_{C_y}} = C_{z_C}.
    \end{equation}

    \noindent ``$\subseteq$'': It follows from $y\preceq_2 x$ that $C_y \subseteq C$ since $\stplus$ is maximal in $\preceq_2$, and $\stplus$ is the only state larger than $\stconn$ under $\preceq_2$. The claim follows.

    \noindent ``$\supseteq$'': Let $\ell\in x^{-1}(\classD) \cap z^{-1}(\classD)$ it follows from $z\preceq y$ that $y_{\ell} \in\{\stnone,\stdisc, \stdconn\}$. But then it follows from $y\preceq_2 x$ that $y_{\ell}\in\{\stdisc, \stconn\}$. Hence, $\ell \in C_y$. The claim follows.

    Now we prove the lemma. It holds that
    \begin{align*}
        &h^{(1)}_i(x) = \mu_2 h^{(2)}_i(x)\\
        &= \sum\limits_{y\preceq_2 x} \mu_2(y, x)\sumstack{z_1, z_2\preceq y\\ (z_1)_D\statejoin (z_2)_D = y_D} \Big[\big|C_{(z_1)_{C_y}}\big| + \big|C_{(z_2)_{C_y}}\big| = i\Big] f(z_1) \cdot g(z_2)\\
        &= \sum\limits_{y\preceq_2 x} \mu_2(y, x)
        \sumstack{z_1\lor z_2=y'\preceq y\\ (z_1)_D\statejoin (z_2)_D = x_D} 
        \Big[\big|C_{(z_1)_C}\big| + \big|C_{(z_2)_C}\big| = i\Big]
        f(z_1) \cdot g(z_2)\\
        &= \sum\limits_{y'\prec x}\sumstack{z_1\lor z_2=y'\\(z_1)_D\statejoin (z_2)_D = x_D}
        \Big[\big|C_{(z_1)_C}\big| + \big|C_{(z_2)_C}\big| = i\Big]
        f(z_1) \cdot g(z_2)
        \sumstack{y' \preceq y \preceq_2 x} 
        \mu_2(y, x)\\
        &
        \begin{multlined}
        =
        \sum\limits_{y'\prec x}\sumstack{z_1\lor z_2=y'\\(z_1)_D\statejoin (z_2)_D = x_D}
        \Big[\big|C_{(z_1)_C}\big| + \big|C_{(z_2)_C}\big| = i\Big]
        f(z_1) \cdot g(z_2)\\
        \sumstack{y'_D=y_D=x_D\\y'_{C^*}\leq y_{C^*}=x_{C^*}}
        \sumstack{y'_A \preceq_2 y_A \preceq_2 x_A} 
        \mu_2(y_A, x_A)
        \end{multlined}\\
        &= \sumstack{y'\prec x\\y'_A=x_A}\sumstack{z_1\lor z_2=y'\\(z_1)_D\statejoin (z_2)_D = x_D}
        \Big[\big|C_{(z_1)_C}\big| + \big|C_{(z_2)_C}\big| = i\Big]
        f(z_1) \cdot g(z_2)\\
        &= \sumstack{z_1,z_2\prec x\\(z_1)_A\lor (z_2)_A=x_A\\(z_1)_D\statejoin(z_2)_D = x_D}
        \Big[\big|C_{(z_1)_C}\big| + \big|C_{(z_2)_C}\big| = i\Big]
        f(z_1) \cdot g(z_2).
    \end{align*}
    The third equality holds by \cref{eq:conv-h1-y-implies-x}.
    In the fifth equality, we partition $y',y$ and $x$ 
    into their restrictions to the sets $C^*, A, D$, where $y'_D = y_D = x_D$ holds by definition, and $y_{C^*} = x_{C^*}$ holds for $y\preceq_2 x$.
    The equality then holds since it holds that $\mu_2(y, x) = \mu_2(y_A, x_A)$ for $y\preceq_2 x$. Finally, the last equality holds by \cref{lem:mobius-zero-one}.

    Since it holds that $h^{(1)}(x) = \mu_2 h^{(1)}_i(x)$ for $i=\big|x^{-1}(\classC)\big|$, it follows that
    \begin{equation}
        h^{(1)}(x) = 
        \sumstack{z_1,z_2\preceq x\\(z_1)_A\lor_3 (z_2)_A = x_A\\(z_1)_D\statejoin(z_2)_D = x_D}
        \Big[ \big|C_{(z_1)_{C}}\big| + \big|C_{(z_2)_{C}}\big| = |C_x| \Big] 
        f(z_1) \cdot g(z_2).
    \end{equation}
    We claim that $(z_1)_A\lor_2 (z_2)_A = x_A$ and $|C_{(z_1)_C}| + |C_{(z_2)_C}| = |x_C|$ hold, if and only if it holds that $(z_1)_A\statejoin (z_2)_A = x_A$, as the only way for $(z_1)_{\ell}\lor (z_2)_{\ell} \in \classC$ but $(z_1)_{\ell}\simplejoin (z_2)_{\ell} \notin \classC$ is if $(z_1)_{\ell},(z_2)_{\ell} \in \classC$. But then there must exist another index $\ell'\in D$ with $(z_1)_{\ell'}, (z_2)_{\ell'} \notin \classC$, which contradicts the assumption that $(z_1)_A\lor_3 (z_2)_A = x_A$. Hence, it must hold that $(z_1)_A\statejoin (z_2)_A = x_A$ and the lemma follows.
\end{proof}

The following result follows directly from the definitions of $\statejoin$, $h^{(1)}$ and $\zeta_1$ by \cref{cor:mobius-zeta}.

\begin{corollary}\label{cor:conv-h1-is-zeta}
    It holds that $h^{(1)} = \zeta_1 (f\stateconv g)$, and hence, $f\stateconv g = \mu_1 h^{(1)}$.
\end{corollary}

\begin{proof}[Proof of~\cref{lem:conv-time}]
    Given the mappings $f$ and $g$, we start by computing the mappings $f^{(1)}, f^{(1)}_i, f^{(2)}_i, f^{(2)}_{i,j}$ and $f^{(3)}_{i,j}$ as defined above for all values $i,j\in[k]$. Similarly, we compute the mappings $g^{(1)}, g^{(1)}_i, g^{(2)}_i, g^{(2)}_{i,j}$ and $g^{(3)}_{i,j}$. Each of these computations can be done in time $\ostar(6^k)$ by~\cref{lem:mobius-zeta-time}.
    We then compute the mapping $h=f\cdot g$ in $Oh(6^k)$ time by iterating over all states $x\in\states$ and computing $h(x)=f(x)\cdot g(x)$. Finally, we compute the mappings $h^{(2)}_{i,j}, h^{(2)}_i, h^{(1)}_i, h^{(1)}$ and $h$ for all values of $i,j\in[k]$ as defined in \cref{def:reverse-conv}. Each of these computations can be done in time $\ostar(6^k)$ by \cref{lem:mobius-zeta-time}. 
    Hence, the whole process runs in time $\ostar(6^k)$, and it holds by \cref{cor:conv-h1-is-zeta} that $h = f\stateconv g$. 
\end{proof}

\section{Lower bound}\label{sec:lb}

In this section, we prove \cref{theo:lower-bound}, showing that our upper bound is essentially tight. Our lower bound is conditioned on the Strong Exponential-Time Hypothesis (SETH) stated as follows:

\begin{conjecture}[SETH~\cite{DBLP:journals/jcss/ImpagliazzoP01,DBLP:journals/jcss/ImpagliazzoPZ01}]
    For any positive value $\delta$ there exists an integer $d$ such that the $d$-SAT problem cannot be solved in time $\ostar((2-\delta)^n)$, where $n$ denotes the number of the variables.
\end{conjecture}

We prove our claim by a reduction from the $q$-CSP-$B$ problem based on the following results:

\begin{theorem}[{\cite[Theorem 3.1]{DBLP:journals/siamdm/Lampis20}}]
    \label{theo:lampis-csp-lb}
    For any $B \geq 2$, $\varepsilon >0$ we have the following: if the SETH is true, then there exists a $q$ such that $n$-variable $q$-CSP-$B$ cannot be solved in time $\ostar((B-\varepsilon)^n)$.
\end{theorem}

First, let us define the $n$-variable $q$-CSP-$B$ problem. Given is a set of variables $\var = \{v_1,\dots v_n\}$ with values in $[B]$ and a set $\mathscr{C} = \{C_1,\dots C_m\}$ of $q$-constraints, where a constraint $C_j = (V_j, S_j)$ is a pair consisting an ordered tuple of $q$ variables $V_j\in \var^q$ and a set of $q$-tuples of values $S_j\subseteq [B]^q$.

A partial assignment $\alpha$ is a mapping from a subset $S\subseteq \var$ to $[B]$. For a partial assignment $\alpha:S\rightarrow[B]$, let $\tilde{\alpha}:\var\rightarrow [B]\cup \{\perp\}$ be the extension of $\alpha$ to all variables in $\var$ with $\alpha(v) =\perp$ for $v\notin S$. Then we say that a partial assignment $\alpha$ satisfies a constraint $C_j$ if the tuple $\tilde{\alpha}(V_j)$ is in $S_j$. The goal of the $n$-variable $q$-CSP-$B$ problem is to find an assignment $\alpha:\var\rightarrow[B]$ satisfying all the constraints of $\mathscr{C}$ or to declare that no such assignment exists.

From now on, fix $B=6$.
Our lower bound follows similar tight lower bounds for structural parameters \cite{DBLP:journals/corr/bojikianfghs25, DBLP:journals/siamdm/Lampis20} based on SETH. We provide a parameterized reduction from the $q$-CSP-$B$ problem to the \Fvsp problem for each value of $q$, bounding the value of the parameter in the resulting graph by a constant above $n$. We show that an algorithm with running time $\ostar((6-\epsilon)^{\lcw})$ for any positive value $\epsilon$ implies an algorithm with running time $\ostar((6-\epsilon)^{n})$ for the $q$-CSP-$B$ problem for any value of $q$, which contradicts SETH. Hence, assuming SETH, the lower bound must hold.

\subsection{Construction of the graph}

Let $I$ be the given instance of the $q$-CSP-$B$ problem with $n$ variables and $m$ constraints.
A \emph{deletion edge} $\tri{u}{v}$ between two vertices $u$ and $v$ in a graph $G$ for two distinct vertices $u,v\in V(G)$ consists of the edge $\{u,v\}$, and a new vertex $w_{u,v}$ adjacent to $u$ and $v$ only.
We call the vertex $w_{u, v}$ a \emph{deletion vertex}. For a set of vertices $S\subseteq V$, we denote by $\inctri{S}$ the set consisting of $S$ and all deletion vertices introduced by deletion edges between the vertices in $S$.

The graph $G$ consists of copies of fixed components of constant sizes, called gadgets. We distinguish two kinds of gadgets in $G$, namely a \emph{path gadget} and a \emph{constraint gadget}. We distinguish a single vertex $r$ in $G$ called the \emph{root} vertex and two more vertices $g_1$ and $g_2$ in $G$, called \emph{guards}.

More concretely, $G$ consists of $n$ path-like structures called \emph{path sequences} (paths for short), each consisting of $c = (5n + 1)m$ consecutive path gadgets. Consecutive path gadgets are connected by adding bicliques between two vertices of each path gadget (its \emph{exit vertices}) and two vertices (called the \emph{entry vertices}) of the path gadget following it. These small cuts give a path sequence its narrow path-like structure. The entry vertices of the first path gadget of each path sequence are adjacent to $g_1$, while the exit vertices of the last path gadget are adjacent to $g_2$.

For a path sequence $P_i$, let $\pgad{i}{1},\dots \pgad{i}{c}$ be the path gadgets on $P_i$ from left to right. We divide $P_i$ into $(5n + 1)$ groups of size $m$ each, called \emph{segments}, where the $\ell$th segment $\seg{i}{\ell}$ consists of the path gadgets $\pgad{i}{(\ell-1)m + 1}, \dots, \pgad{i}{\ell\cdot m}$. For $\ell\in[5n + 1]$, we call the set consisting of the $j$th segment of each path sequences a \emph{section} $\sect{\ell}$. Finally, for $j\in [c]$, we call the set consisting of the $j$-th path gadget of each path sequence a \emph{column}, i.e., for $j\in[c]$ we have $\col{j} = \{\pgad{1}{j},\dots, \pgad{n}{j}\}$. See \cref{fig:lb-construction} for illustration.

\begin{figure}
    \centering
    \begin{subfigure}[b]{0.3\textwidth}
        \centering
        \includegraphics[width=\textwidth]{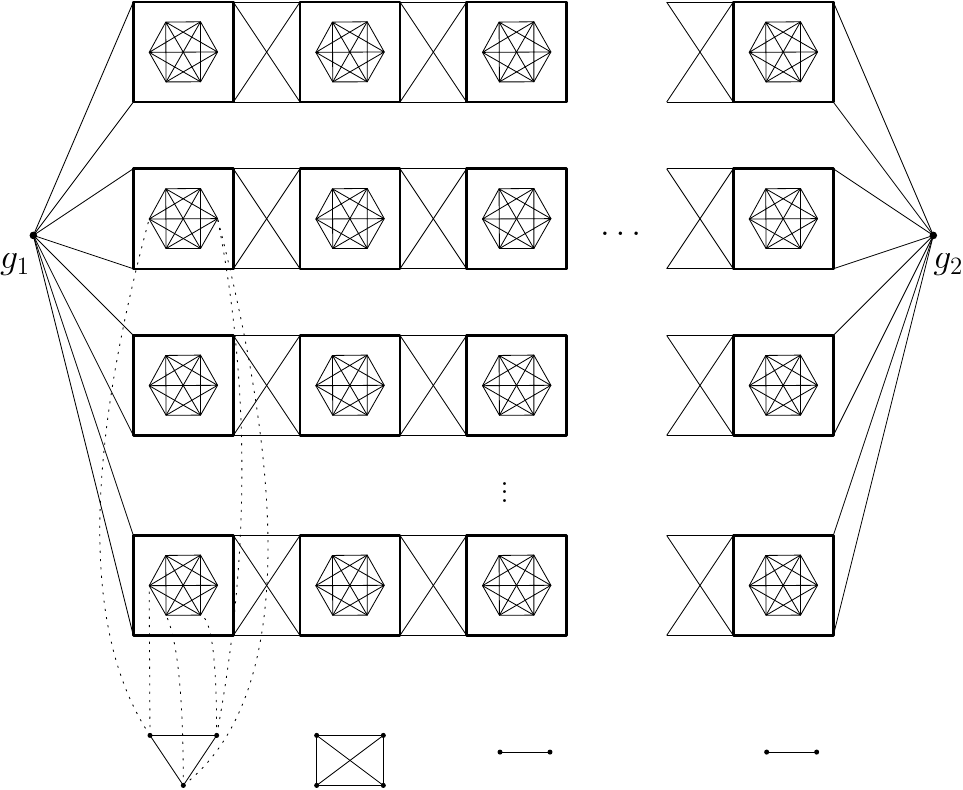}
        
    \end{subfigure}
    \hfill
    \begin{subfigure}[b]{0.3\textwidth}
        \centering
        \includegraphics[width=\textwidth]{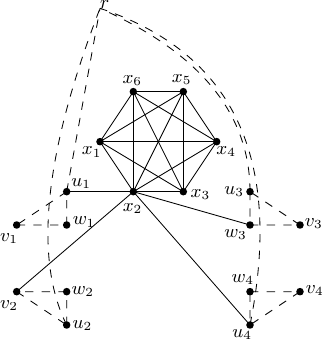}
    \end{subfigure}
    \caption{On the left, the graph $G$ consisting of $n$ path sequences. We depict a path gadget as a square with its clique vertices drawn inside of it. Each column corresponds to a constraint with a constraint gadget hanging to it. The deletion edges between the first constraint gadget and the first column are drawn as dotted lines, where the rest are omitted for clarity. From the drawn gadgets, it can be seen that this graph $G$ corresponds to an instance where the first constraint admits two satisfying assignments, the second has four, and the third has two. The variables underlying the first constraint are the second and the last variables, as the deletion edges are drawn accordingly.
    On the right, we depict a path gadget. Normal edges are depicted as dashed lines, while deletion edges are depicted as solid lines. The deletion edges are drawn between the cycles $C_i$ and only the clique vertex $x_2$ for clarity. Note that a deletion edge to $v_2$ means that if the path gadget has state $2$ then $v_2$ must be in the solution (and hence, state $\stnone$), a deletion edge to $u_1$ means that $v_1$ is disconnected from $r$ inside the gadget (and hence, state $\stdisc$), while a deletion edge to $w_3$ means that $v_3$ is connected to $r$ through $u_3$ (state $\stconn$).
    \label{fig:lb-construction}}
\end{figure}

We attach a constraint gadget to each column.
Let $\states := \{\stnone, \stdisc, \stddisc, \stconn, \stplus, \stdconn\}$ be the same set of states defined in the upper bound.
A solution defines a state in each path gadget, given by the intersection of the solution with its vertices. Each constraint gadget in a section corresponds to a different constraint, and it ensures that the states defined by a solution in its column corresponds to a satisfying assignment of this constraint. As noted above, each path consists of $5n+1$ segments. This ensures that for each solution $S$ there exists some section $\sect{\ell}$, where it holds for each segment of this section, that all path gadgets of this segment admit the same state under $S$. Such a state assignment corresponds to a satisfying assignment of the instance $I$. Therefore, each solution $S$ corresponds to a satisfying assignment of $I$.

\subparagraph{Path gadget $\pgadget$.} A path gadget $\pgadget$ consists of $6$ vertices $x_1,\dots, x_6$, called the clique $K$ and of four cycles of length $3$ each $C_1,\dots, C_4$.
All clique vertices are adjacent to the root $r$, and there is a deletion edge between each pair of them.
Each cycle $C_i$ consists of the vertices $v_i,u_i,w_i$, where $u_i$ is adjacent to $r$. The vertices $v_1$ and $v_2$ serve as the \emph{entry vertices} of $\pgadget$, and we call $C_1,C_2$ the entry cycles of $\pgadget$. Similarly, $v_3$ and $v_4$ serve as the \emph{exit vertices} and we call $C_3,C_4$ the exit cycles of $\pgadget$.

We define the transition mapping $t:[B]\rightarrow \{\stnone, \stdisc, \stconn\}^4$ given by the following mappings:

\begin{tabularx}{.8\textwidth}{XXX}
    $1\mapsto (\stnone, \stnone, \stconn, \stconn)$,&
    $2\mapsto (\stdisc, \stnone, \stconn, \stdisc)$,&
    $3\mapsto (\stdisc, \stdisc, \stconn, \stnone)$\\
    $4\mapsto (\stconn, \stnone, \stdisc, \stdisc)$,&
    $5\mapsto (\stconn, \stdisc, \stdisc, \stnone)$,&
    $6\mapsto (\stconn, \stconn, \stnone, \stnone)$.
\end{tabularx}

Intuitively, each state of $\states$ can be decomposed into two simple states in $\{\stnone, \stdisc, \stconn\}$. Hence, each element in the image of $t$ corresponds the transition between two states. A solution $S$ defines a simple state in each cycle $C_i$ of a path gadget $\pgadget$ given by the intersection of $S$ with the vertices of $C_i$.
We construct path gadgets in such a way that the states defined by any solution $S$ in the entry and exit cycles form a valid transition (i.e., an element in the image of $t$), and we call the corresponding index of this transition the \emph{state} of this path gadget defined by the solution $S$.

For each value $o\in[6]$ and each value $i\in [4]$, we add a deletion edge between the clique vertex $x_o$ and the cycle $C_i$ as follows: Let $Z := \big(t(o)\big)_i$ be the $i$th simple state in the $o$th transition. We add a deletion edge between $x_o$ and $v_i$ if $Z = \stnone$, between $x_o$ and $u_i$ if $Z=\stdisc$, and between $x_o$ and $w_i$ if $Z=\stconn$. This concludes the construction of a path gadget. For a specific path gadget $\hat{\pgadget}$, and an arbitrary path gadget vertex $v\in\pgadget$ defined above, we denote by $v(\hat{\pgadget})$ the copy of $v$ in the context of the gadget $\hat{\pgadget}$.

We add all edges between the exit vertices of each path gadget and the entry vertices of the following path gadget. We also add all edges between $g_1$ and the entry vertices of the first path gadget of each path sequence, and all edges between $g_2$ and the exit vertices of the last path gadget of each path sequence.

\subparagraph{constraint gadget \cgadget.}
For $j\in[c]$, we attache a constraint gadget $\cgad{j}$ to the $j$th column, corresponding to the $h$th constraint $C_{h}$ where $j = \ell m + h$ for some integer value $\ell$. Let $S_h = \{s_1,\dots s_k\}$. Then $\cgadget$ consists of the $k$ vertices $z_1\dots z_k$, called the \emph{clique vertices} as well, with a deletion edge between each pair of them. For each $\ell\in[k]$ and each $p\in [q]$, let $v_i := (V_h)_p$ be the $p$th variable in the constraint $C_h$, and let $y := (s_{\ell})_p$ be the value assigned to $v_i$ by $s_{\ell}$. Then we add a deletion edge between $z_{\ell}$ and $x_{y'}(\pgad{i}{j})$ for all $y'\in[B]\setminus\{y\}$, where $x_{y'}(\pgad{i}{j})$ is the copy of the clique vertex $x_{y'}$ in the path gadget $\pgad{i}{j}$. We call $k$ the \emph{order} of the constraint gadget $\cgadget$.
This completes the description of a constraint gadget, and with which the description of the graph $G$.

\subparagraph{Budget.}
Now we define the target value $\budget$ in the resulting instance $(G, \budget)$. We define a family of pairwise disjoint components $\lbfamily \subseteq 2^V$, and we define the mapping $\lbound:\lbfamily\rightarrow\mathbb{N}$ as a lower bound on the size of a feedback vertex set in each of these components. Finally, we define $\budget$ as the sum of $\lbound(X)$ for all $X\in \lbfamily$.
By the tightness of $\budget$, as we shall see in \cref{lem:lb-b-tight}, it holds that any solution of size $\budget$ must contain exactly $\lbound(X)$ vertices in each such component $X\in\lbfamily$, and no other vertices.

Now we describe the family $\lbfamily$. From each path gadget $\pgadget$ in $G$, the family $\lbfamily$ contains each cycle $\inctri{C_i}$ together with its deletion edges for $i\in[4]$, and of the clique vertices $\inctri{K}$ together with the deletion edges between them. The family $\lbfamily$ also contains each constraint gadget in $G$. We define the mapping $\lbound$ as follows: we set $\lbound(\inctri{C_{\ell}}) = 1$ for each cycle of a path gadget, and $\lbound(\inctri{K}) = 5$ for the clique vertices of a path gadget. For each constraint gadget $\cgadget$ of order $k$ in $G$, we define $\lbound(\cgadget) = k-1$.

\subsection{Proof of correctness}

\begin{lemma}\label{lem:lb-cw-bound}
    There exists a constant $k_0$ such that for $k=n+k_0$, a linear $k$-expression of $G$ can be constructed in polynomial time.
\end{lemma}

\begin{proof}
    We describe a linear $k$-expression of $G$. Let $x_0, x_1, \dots$ be the nodes of $\syntaxtree$ building a simple path, and $x'_i$ be the private neighbor of $x_i$ (if exists) in the caterpillar structure of the syntax tree of a linear $k$-expression. Since we only use these nodes $x'_i$ to introduce a new vertex of some label $\ell$, and then unify it at $x_i$ with the labeled graph built so far (i.e.\ $G_{x_{i-1}}$), we will just ignore these nodes and say that at $x_i$ we add a vertex labeled $\ell$ to $G_{x_{i-1}}$. Hence, we consider $x_0, x_1, \dots$ as timestamps of building $G$, and denote $G_{x_i}$ by $G_i$. Similarly, let $\mu_i = \mu_{x_i}$, $V_i = V_{x_i}$, $E_i = E_{x_i}$, and $\lab_i = \lab_{x_i}$.

    We reserve the first $n$ labels for later, and use the label $\lforget := n+1$ to label vertices whose incident edges have all been introduced before the current timestamp. 
    When we say we forget a vertex $u$ at some timestamp $x_i$, we mean that we relabel (the label of) this vertex to $\lforget$, i.e.\ it holds that
    $\mu_i := \relabel{\lab_{i-1}(u)}{\lforget}(\mu_{i-1})$, where we only apply this operation when the vertex $u$ has a unique label in $G_{i-1}$.
    We also reserve the label $\ltri := n+2$ to create deletion edges. Let $x_i$ be some timestamp. For $u,v\in V_x$, let $i=\lab_x(u)$ and $j=\lab_x(v)$. When we say that we add a deletion edge between $u$ and $v$, we assume first that $\rlab_x(i) = \{u\}$ and $\rlab_x(j) = \{v\}$ (both have unique labels). In this case, we mean that we add a new vertex $w_{u,v}$ labeled $\ltri$, and we add all edges between $u$, $v$ and $w_{u,v}$, and then we relabel $w_{u,v}$ to $\lforget$. Since this is the only place where we use the label $\ltri$ in the construction, it holds at any given timestamp $x_i$ that $|\rlab_i(\ltri)|\leq 1$.

    Now we are ready to describe the linear $k$-expression of $G$. First, we introduce the vertices $r, g_1$ and $g_2$ with labels $\lroot = n'+3$, $\lgl = n'+4$, and $\lgr=n'+5$ respectively. These are the only vertices that will be assigned these labels at any timestamp, and will never be relabeled. Hence, for a vertex $u$ of a unique label at any timestamp, we can make it adjacent to any of them through a join operation without changing the adjacency of any other vertices of the graph.

    We divide the rest of the construction into $c$ phases. In each phase we add a new column to the construction, one by one, adding the $j$th column in the $j$th phase. Let $x_{h}$ be the timestamp at the end of phase $j$. We will preserve the following invariant: 
    It holds that $\lab_{h}(r) = \lroot$, $\lab_{h}(g_1) = \lgl$, $\lab_{h}(g_2) = \lgr$, and for each $i\in[n]$, that $\lab_{h}(v^3(X_i^j)) = \lab_{h}(v^4(X_i^j)) = i$.
    All other vertices in $V_{h}$ are labeled $\lforget$ at $x_h$.

    Now we describe the $j$th phase for each index $j\in [c]$. The phase starts by introducing the vertices of the constraint gadget $\cgad{j}$, each having its own unique label (starting from $n+6$). We add all deletion edges between all these vertices. After that we build the path gadgets of the current column one by one, where we build the $i$th path gadget of the current column as follows: First we create all vertices of the path gadget assigning a unique label to each vertex. We add the edges inside the path gadget, and between the path gadget and the root as described in the construction of the graph. We also add the deletion edges between the path gadget and the constraint gadget in its column. This is possible since all these vertices have unique labels. We also add all edges between the entry vertices of the path gadget and the exit vertices of the preceding path gadget (having label $i$), if this is not the first path gadget, or with $g_1$ otherwise. After that we forget the exit vertices of the preceding path gadget, and we relabel the exit vertices of the current path gadget to the label $i$, and forget all other vertices of the current path gadget.

    When we finish the construction of each column, we forget its constraint gadget. When we finish the construction of the last column, we join the label $\lgr$ of $g_2$ with all labels $i\in[n]$, introducing the edges between $g_2$ and the exit vertices of the last path gadget in each path sequence.
    Clearly, this results in the graph $G$ as described above. At each timestamp $x$, we use at most $n$ labels for the exit vertices, we also use the $5$ reserved labels $\lforget$, $\ltri$, $\lroot$, $\lgl$, $\lgr$, and we uniquely label the vertices of at most one constraint gadget and one path gadget, whose sizes are both constant.
\end{proof}

\begin{lemma} \label{lem:lb-sol-if-sat}
    If $I$ is satisfiable, then $G$ admits a feedback vertex set of size $\budget$.
\end{lemma}

\begin{proof}
    Let $\alpha$ be a satisfying assignment of $I$. We define the solution $\sol \subseteq V$ of size $\budget$ as the set of the following vertices:
    From each path gadget $\pgad{i}{j}$ with $\alpha(v_i) = y$ we include in $\sol$ all vertices $x_{y'}$ for $y'\in[B]\setminus\{y\}$. For each $\ell\in[4]$ we also include $v_{\ell}$ if $t(y)_\ell = \stnone$, $u_{\ell}$ if $t(y)_\ell = \stdisc$, and $w_{\ell}$ if $t(y)_\ell = \stconn$.
    From each constraint gadget $\cgad{j}$ corresponding to a constraint $C_h$, we include all vertices $z_{\ell'}$ of the constraint gadget except for the vertex $z_{\ell}$ where $s_{\ell} = \alpha(V_h)$. The vertex $z_{\ell}$ exists, since $\alpha$ satisfies $C_h$.

    We claim that $\sol$ is a feedback vertex set of $G$ of size $\budget$. The size follows from the fact that we pack exactly $\lbound(X)$ vertices in each component $X \in \lbfamily$. Now we show that $G\setminus \sol$ is a forest.
    By the choice of $\sol$, it contains at least one endpoint of each deletion edge between the clique vertices, and a vertex of each cycle $C_i$ of each path gadget. It also contains a least one endpoint of each deletion edge between the vertices of a constraint gadget. Hence, cycles can only appear as deletion edges between a clique vertex and a vertex of a constraint gadget, or as a closed walk between passing through two consecutive path gadgets. Any other closed walk in the graph must contain a cycle of one of these forms as a subsequence.

    First, we exclude the first case. Assume there exists such a cycle, and let $x_y, z_{\ell}$ be the endpoints of the corresponding deletion edge, where $x_y$ is a clique-vertex of a path gadget $\pgad{i}{j}$ and $z_{\ell}$ is a vertex of the constraint gadget $\cgad{j}$ in the same column. Since $x_y$ is not in the solution, it means that $\alpha(v_i) = y$. But then it follows from the choice of $\sol$ that $(S_{\ell})_i = y$, and hence, by the definition of a constraint gadget, there doesn't exist a deletion edge between $x_y$ and $z_{\ell}$ which is a contradiction.

    Finally, let $C$ be a cycle of the latter form. Let $\pgadget, \pgadget'$ be two consecutive path gadgets that intersect $C$. Let $x_y, x_{y'}$ be the clique vertices of $\pgadget$ and $\pgadget'$ respectively, excluded from $\sol$. By the choice of $\sol$, there are no edges left between the cycles $C_i$ and the clique vertices. Hence, any minimal such cycle is either of the form $v_{3}(\pgadget), v_{1}(\pgadget'), v_{4}(\pgadget), v_{2}(\pgadget')$, i.e.\ between exit and entry vertices only, or is of the form $r,u_i(\pgadget), v_i(\pgadget), v_{i'}(\pgadget'), u_{i'}(\pgadget')$. Both cases can be excluded directly by inspecting the transitions mapping $t$, since the former requires a tuple of four $\stnone$, and the latter requires a tuple containing $\stconn$ both in its first two positions and in its latter two positions.
\end{proof}

Now we aim to prove the other direction of this lemma. Let $\sol$ be a fixed feedback vertex set of $G$ of size $\budget$. Let $\overline{\sol} = V\setminus \sol$ and $\hat{G} = G\setminus\sol$. For a set $X\subseteq V$, we denote $\sol_X = \sol \cap X$.

\begin{lemma}\label{lem:lb-b-tight}
    It holds that $\sol = \dot{\bigcup\limits_{X\in\lbfamily}}\sol_X$, and for each $X\in \lbfamily$ that $|\sol_X| = \lbound(X)$.
\end{lemma}

\begin{proof}
    The lower bound $|\sol_X| \geq \lbound(X)$ holds for each family $X\in\lbfamily$, since deletion edges and the cycles $C_i$ are all triangles. If the solution contains an addition vertex in one of the component of $\lbfamily$, or a vertex from outside these components, then it must hold that $|\sol| > \budget$, which is a contradiction.
\end{proof}

If $\sol$ contains a vertex $w_{u,v}$ added by a deletion edge $\tri{u}{v}$, then the set $\big(\sol\setminus\{w_{u,v}\}\big)\cup \{u\}$ is a solution of at most the same size, since any cycle containing $w_{u,v}$ must also contain $u$. Hence, from now on, we can assume that $\sol$ does ot contain any vertices added by deletion edges. That means $\sol$ contains at least one endpoint of each deletion edge, all but at most one clique-vertex in each path gadget, and all but one vertex of the clique vertices of each constraint gadget.

\begin{definition}\label{def:lb-pgad-state}
    Let $\pgadget$ be some path gadget, and $z_s$ be the only clique vertex of $\pgadget$ not in $\sol$ for some value $s \in [6]$. Then we call $s$ the \emph{state} of $\pgadget$ defined by $\sol$ (denoted by $\statef_{\sol}(\pgadget)$). We omit $\sol$ when clear from context.
\end{definition}

\begin{lemma}\label{lem:lb-state-imply-v}
    For a path gadget $\pgadget$, we define $G_{\pgadget}^{\sol} = \hat{G}\big[\pgadget\cup\{r\}\big]$. Let $s= \statef(\pgadget)$. For $i\in[4]$ we define the following predicates:
    \[
    \phi_i:=(v_i\in\sol),
    \quad
    \chi_i:=\lnot \phi_i \land \big(v_i,r \text{ are not connected in } G_{\pgadget}^{\sol}\big),
    \quad \text{and}\quad 
    \xi_i:=\lnot \phi_i \land \lnot \chi_i.
    \]
    Then all following implications hold:
    \begin{itemize}
        \item $s = 1 \implies \phi_1 \land \phi_2 \land \xi_3 \land \xi_4$.
        \item $s = 2 \implies \chi_1 \land \phi_2 \land \xi_3 \land \chi_4$.
        \item $s = 3 \implies \chi_1 \land \chi_2 \land \xi_3 \land \phi_4$.
        \item $s = 4 \implies \xi_1 \land \phi_2 \land \chi_3 \land \chi_4$.
        \item $s = 5 \implies \xi_1 \land \chi_2 \land \chi_3 \land \phi_4$.
        \item $s = 6 \implies \xi_1 \land \xi_2 \land \phi_3 \land \phi_4$.
    \end{itemize}
\end{lemma}

\begin{proof}
    It holds by definition, that the three predicates $\phi_i$, $\chi_i$, and $\xi_i$ are pairwise disjoint for each value $i\in[4]$. The predicate $\phi_i$ holds if and only if a $v_i\in \sol$, $\chi_i$ holds if and only if $u_i \in \sol$, and $\xi_i$ holds if and only if $w_i \in \sol$. The lemma follows from the assumption, that $\sol$ contains one vertex of $C_i$, an endpoint of each deletion edge between $C_i$ and the clique-vertices, and all clique vertices of $\pgadget$ except for one vertex.
\end{proof}

\begin{lemma}\label{lem:lb-states-transition-direction}
    For each two consecutive path gadgets $\pgadget_1$ and $\pgadget_2$, let $s_1$ be the state of $\pgadget_1$ and $s_2$ be the state of $\pgadget_2$. Then it holds that $s_2 \leq s_1$.
\end{lemma}

\begin{proof}
    Assume this is not the case, and let $\pgadget_1$ and $\pgadget_2$ be two consecutive path gadgets with states $s_1$ and $s_2$ respectively, such that $s_2 > s_1$.
    Let $t(s_1) = (p_1,p_2,p_3,p_4)$ and $t(s_2) = (q_1,q_2,q_3,q_4)$ be the transition tuples of $s_1$ and $s_2$ respectively. Let $P = \{p_3,p_4\}$ and $Q=\{q_1,q_2\}$.
    Then it holds by \cref{lem:lb-state-imply-v}, and since the cut between $\pgadget_1$ and $\pgadget_2$ is a complete bipartite graph, that if $\stconn \in P\cap Q$, then $v_3(\pgadget_1)$ is connected to $r$ in $G_{\pgadget_1}^{\sol}$ and $v_1(\pgadget_2)$ is connected to $r$ in $G_{\pgadget_2}^{\sol}$, which together with the edge $\{v_3(\pgadget_1),v_1(\pgadget_2)\}$ creates a cycle. Moreover, if $\stnone \notin P\cup Q$, then we get the cycle $v_3(\pgadget_1),v_1(\pgadget_2),v_4(\pgadget_1),v_2(\pgadget_2)$ in $\solg$. These two observations already imply contradictions for all pairs $(s_1,s_2)$ with $s_2 > s_1$ except for the case where $s_1 = 1$. In this case it holds that $p_3=p_4 = \stconn$. Hence, both vertices $v_3$ and $v_4$ of $\pgadget_1$ exist in $\solg$ and are connected to $r$ through $u_3$ and $u_4$ respectively. Therefore, if $Q\neq \{\stnone\}$, then we get the cycle consisting of the two paths $v_3,u_3,r$ and $v_4,u_4,r$ in $\pgadget_1$, together with the path $v_3(\pgadget_1), v_1(\pgadget_2), v_4(\pgadget_1)$ leading to a contradiction.
\end{proof}

\begin{lemma}\label{lem:lb-sat-if-sol}
    If $G$ admits a feedback vertex set of size $\budget$, then $I$ is satisfiable.
\end{lemma}

\begin{proof}
    Let us fix some path sequence, and consider the path gadgets $\pgad{i}{1} \dots \pgad{i}{c}$ on this path sequence. Let $s_{i_1} \dots s_{i_c}$ be the states of these path gadgets defined by $\sol$ in order. Then by \cref{lem:lb-states-transition-direction} it holds that $i_{j+1} \leq i_{j}$ for all $j\in[c-1]$. Since we have exactly $6$ states, there can exist at most $5$ different indices $j$ with $i_j \neq i_{j+1}$. Therefore, over all $n$ path sequences, there can exist at most $5n$ columns, such that some path gadget on each of these columns has a different state than the following path gadget, which in turn, implies that at most $5n$ sections exist with some path gadget in each of these sections having a different state than the following path gadget. Since the graph $G$ contains $5n + 1$ sections, there must exist at least one section in $G$, where each path gadget in this section has the same state as the following path gadget.
    
    Let us fix one such section $\sect{r}$ for $r \in [5n+1]$. For each segment $\seg{i}{r}$ in this section ($i\in [n]$), let $s_i$ be the state of all path gadgets on this segment. Then we claim that $\alpha:\var\rightarrow[B]$ with $\alpha(v_i) = s_i$ for all $i\in[n]$ is a satisfying assignment of $I$. In order to see this, let $C_j$ be some constraint for $j\in [m]$. Let us consider the $j$th column of the $r$th section. Since $\sol$ is a solution, there must exist a vertex $z_{\ell}$ of the constraint gadget attached to this column in the solution. Let $s_{\ell}$ be the tuple corresponding to $z_{\ell}$. Let $k\in[q]$ be some index, and let $v_i = (V_j)_k$ and $y = (s_{\ell})_k$. Then it must hold that all clique vertices $x_{y'}$ for $y'\neq y$ are in $\sol$, since $\solg$ is acyclic. This implies that the path gadget in this column that belongs to the $i$th path has state $y$, and hence, $\alpha(v_i) = y$. Since this holds for each $k\in[q]$, it follows that $\alpha(V_j) = s_{\ell}$, and hence, $\alpha$ satisfies $C_j$. Since this holds for each constraint $C_j$, it follows that $\alpha$ is a satisfying assignment of $I$.
\end{proof}

\begin{proof}[Proof of~\cref{theo:lower-bound}]
    Assume that there exists an algorithm $A$ that solves the \Fvsp problem in time $\ostar((6-\varepsilon)^k)$ for some positive value $\varepsilon>0$, when a linear $k$-expression $\mu$ is given with the input graph $G$.
    Then given the instance $I$ of the $d$-CSP-$B$ problem over $n$ variables for $B=6$ and some positive integer $d$, one can construct the graph $G$ together with the linear $k$-expression $\mu$ of $G$ and choose the value $\budget$ as described above, where both $\mu$ and $k$ are given by \cref{lem:lb-cw-bound}, and run $A$ on the resulting instance. It holds by \cref{lem:lb-sol-if-sat} and \cref{lem:lb-sat-if-sol} that $G$ admits a feedback vertex set of size $\budget$ if and only if $I$ is satisfiable.
    It holds by \cref{lem:lb-cw-bound} that $\mu$ is a linear $k$-expression of $G$ that can be constructed in polynomial time, where $k = n + k_0$ for some constant $k_0$. Hence, the algorithm runs in time
    $\ostar((6-\varepsilon)^{n + k_0}) = \ostar((6-\varepsilon)^{n})$ which contradicts SETH by \cref{theo:lampis-csp-lb}.
\end{proof}

\section{Counting algorithm for treewidth}\label{sec:tw}

In this section we show that, informally speaking, by restricting our dynamic programming algorithm to the states that appear in a tree decomposition, we get an algorithm that counts (modulo $2$) the number of feedback vertex sets of a fixed size in time $\ostar(3^{\tw})$ given a tree decomposition of width $\tw$ with input. This matches the best possible running time for the decision version of this problem under SETH~\cite{DBLP:journals/talg/CyganNPPRW22}, and hence, it is also tight, since a better algorithm for counting (modulo $2$) would also break this bound for the decision version using the isolation lemma.
We formalize this idea through a dynamic programming algorithm on a tree decomposition, where we index the dynamic programming tables by a smaller (restricted) family of patterns.
We start by restating the definitions of tree decompositions and tree width.

\begin{definition}
    A \emph{tree decomposition} of a graph $G$ is a pair $\mathcal{B} = (T, B)$, where $T$ is a tree and $B:V(T) \rightarrow 2^{V(G)}$ is a mapping that assigns to each node $x\in V(T)$ a bag $B(x)\subseteq V(G)$ (also denoted $B_x$) such that:
    \begin{itemize}
        \item For every vertex $v\in V(G)$ there is some $x\in V(T)$ with $v\in B(x)$,
        \item for every edge $\{u,v\}\in E(G)$ there is some $x\in V(T)$ with $u,v\in B(x)$,
        \item for every $v\in V(G)$, the set $\{x\in V(T)\colon v\in B(x)\}$ induces a connected subtree of $T$.
    \end{itemize}
    The \emph{width} of a tree decomposition is defined by $\max_{x\in V(T)}|B(x)|-1$, and the \emph{tree‐width} of $G$ is the minimum width over all such decompositions.

    We call a tree decomposition $(T,B)$ \emph{nice} if $T$ is rooted at a node $r$, $B(r)=\emptyset$, and it holds for every node $x$ of $T$ that either
    \begin{itemize}
        \item $x$ is a leaf node and $B(x)=\emptyset$,
        \item $x$ has a single child $x'$ and $B(x) = B(x')\cup \{v\}$ for $v\notin B(x')$ (introduce vertex node),
        \item $x$ has a single child $x'$ and $B(x) = B(x')\setminus \{v\}$ for some $v\in B(x')$ (forget vertex node),
        \item $x$ has two children $x_1,x_2$ and $B(x) = B(x_1) = B(x_2)$ (join node).
    \end{itemize}
    Finally, we define a \emph{very nice tree decomposition} as a nice tree decomposition with additional nodes $x$ with a single child $x'$, called introduce edge node, with $B(x) = B(x')$, such that for each edge $e$ there exists a unique introduce edge node $x$ corresponding to $e$ with $e\subseteq B(x)$.
    We also assume that the node introducing an edge $e=\{u,v\}$ is an ancestor of all join nodes containing both $u$ and $v$ in their bags.

    We define the graph $G_x=(V_x, E_x)$ for a node $x$ as the subgraph of $G$ containing all vertices and edges introduced in the subtree rooted at $x$. Hence, it holds for each join node $x$ that $B(x)$ induces an independent set in $G_x$. 
\end{definition}

\begin{remark}\label{rem:tw:dec-to-very-nice}
It is well-known that for every tree decomposition $\mathcal{B}$ there exists a very nice tree decomposition $\mathcal{B}'$ of the same graph, having the same width. Moreover, such a transformation from $\mathcal{B}$ to $\mathcal{B}'$ can be done in polynomial time~\cite{DBLP:books/sp/CyganFKLMPPS15}.
\end{remark}

\subsection{Boundaried graphs and graph labeling}

From now on, let $G'=(V',E')$ be the input graph over $n$ vertices, and $\mathcal{B}'$ be a tree decomposition of $G'$ of width $\tw'$. Let $\target$ be the target size of the feedback vertex sets we aim to count. We can assume by \cref{rem:tw:dec-to-very-nice} that $\mathcal{B}'$ is a very nice tree decomposition of $G'$ of width $\tw'$. We define $k=\tw'+1$ as the maximum size of a bag of $\mathcal{B}'$. Let $\mathcal{B}=(T,B)$ be the very nice tree decomposition resulting from $\mathcal{B}'$ by adding a new vertex $v_0$ (not in $G'$) to all bags. This turns the leaf nodes of $\mathcal{B}'$ into nodes that introduce the vertex $v_0$. It holds that $\mathcal{B}$ has width $k:=\tw'+1$, since we increase the size of all bags exactly by one. Let $\nodes := V(T)$ and let $r\in\nodes$ be the root of $T$. Then it holds that $B_r = \{v_0\}$.

It holds that the graph $G=(V,E)$ that corresponds to $\mathcal{B}$ results from $G'$ by adding an additional isolated vertex $v_0$. We aim to count (modulo $2$) the number of feedback vertex sets of size $\target$ in $G$ that exclude the vertex $v_0$, which are exactly the feedback vertex sets of $G'$ of size $\target$, since $v_0$ is isolated in $G$. In fact, similar to clique-width, we will count (modulo $2$) the number of induced forests of size $\budget$ in $G$ that include the vertex $v_0$, for all values of $\budget \in [n]_0$. It holds that the number of feedback vertex sets of size $\target$ in $G$ that exclude $v_0$ is equal to the number of induced forests of size $n + 1 - \target$ in $G$ that include $v_0$.
Again, this motivates the following definition of a partial solution: A \emph{partial solution} $S\subseteq V_x$ at a node $x\in \nodes$ is a set of vertices that contains $v_0$, such that $S$ induces a forest in $G_x$.

\begin{definition}
    We define the mapping $\phi:V\rightarrow [k]_0$ such that the restriction of $\phi$ to any bag $B_{x}$ is injective as follows:
    Let $x_1,\dots x_{n}$ be all forget nodes of $T$ sorted in some top-down ordering over $T'$, and let $x_i$ be the forget node of a vertex $v_i$ for all $i\in[n]$. Then we fix the mapping $\phi$ by iterating over all nodes $x_i$ in increasing order of $i$, and for each value of $i$ we fix $\phi(v_i)$ as the smallest value $\ell\in\mathbb{N}$ not in $\phi\big(B_{x_i}\setminus\{v_i\}\big)$.
    For $v\in V$, we call $\phi(v)$ the label of a vertex $v$.
\end{definition}

\begin{observation}\label{obs:labels}
    It holds for all $i\in[n]$ that the forget node of each vertex in $B_{x_i}$ is an ancestor of $x_i$, and hence, the label of each vertex in this bag, except for $v_i$, is already fixed, and the set $\phi\big(B_{x_i}\setminus\{v_i\}\big)$ is well-defined.
    Moreover, it holds that $v_0$ is the only vertex labeled $0$, and that $\phi(v) \leq k$ for all $v\in V$ since $k$ is the maximum size of a bag in $\mathcal{B}'$. The restriction of $\phi$ to each bag $B_{x}$ is injective, since for any two vertices $u,v$ sharing a bag $B_{x}$, the bag corresponding to the forget node of one of these vertices contains the other, and hence, by the choice of $\phi$ it holds that $\phi(u)\neq\phi(v)$.
\end{observation}

We will use the notion of graph \emph{boundary} to define representation over treewidth. Intuitively, a boundary of a graph is a special set of vertices $B$ that we use to extend this graph to a super graph, and hence, it is a separator in the super graph, that separates the original graph from the rest of the super graph. Mainly, in a tree decomposition, we will be interested in the boundary $B_{x}$ of the graph $G_x$ for each node $x$, as $G$ results from $G_x$ by adding vertices (and edges) separated from $G_x$ by $B_{x}$. It is well-known for many problems, that a boundary can carry enough information to correctly extend a partial solution to a solution in the whole graph. We will formalize this for the \Fvsp problem.

\begin{definition}\label{def:boundary}
    A \emph{boundaried graph} is a graph $G=(V,E)$ given together with a boundary set $B\subseteq V$ and an injective mapping $\phi:B\rightarrow[k]_0$, such that $B$ induces an independent set in $G$. A \emph{boundaried forest} is a boundaried graph that is a forest. We denote the boundaried graph $(G,B,\phi)$ by $G$ when $B$ and $\phi$ are clear from context.

    We say that two boundaried forests $(F_1, B_1, \phi_1)$ and $(F_2,B_2,\phi_2)$ are \emph{compatible} if it holds that $\phi_1$ and $\phi_2$ have equal images, and the following conditions hold: 
    Let $\pi$ be the unique label-preserving bijection from $B_1$ to $B_2$, i.e.\ $\phi_1(v) = \phi_2(\pi(v))$ for all $v\in B_1$, and let $H$ be the graph resulting from the union of $F_1$ and $F_2$ by identifying all pairs $(v,\pi(v))$ for $v\in B_1$. Then it must holds that $H$ is acyclic. We call the process of identifying these pairs of vertices the \emph{glueing} of $F_1$ and $F_2$. Note that the glueing operation might create double edges (cycles of length two) between two vertices.

    We say that two families of boundaried forests $\mathcal{F}_1$ and $\mathcal{F}_2$ are equivalent, if it holds for each boundaried forest $H$ that $H$ is compatible with an odd number of forests of $\mathcal{F}_1$ if and only if it is compatible with an odd number of forests of $\mathcal{F}_2$.
\end{definition}

\subsection{Treewidth patterns}

Now we characterize the needed conditions for acyclicity representation through the notion of ``patterns''. After that we define a more concise representation through a ``reduce'' operation resulting in a more compact set of patterns.

\begin{definition}\label{def:tw-pat}
    A \emph{(treewidth) pattern} $p$ is a partition of a subset $L_p$ of $[k]_0$ that contains the element $0$. We call $L_p$ the label set of $p$, and we call the set containing $0$ in $p$ the zero set of $p$ (denoted $Z_p$). Let $\TP$ be the family of all treewidth patterns. We also define one special pattern $\perp$ not in $\TP$ and call it the \emph{bad pattern}.
    We denote by $\nopat$ the pattern $\{\{0\}\}$ and call it the \emph{empty pattern}.
    We define the family $\CTP\subseteq \TP$ of all treewidth patterns, such that any set different from $Z_p$ is a singleton.
    
    For a pattern $p\in\TP$ and $i\in[k]$, we denote by $p\cup i$ the pattern $\perp$ if $i\in L_p$, or $p\cup\big\{\{i\}\big\}$ otherwise. We define $p\setminus i$ as $p$ itself, if $i\in L_p$, or as the pattern resulting from $p$ by removing $i$ from the set of $p$ that contains $i$. We remove the whole set $\{i\}$ if it is a singleton.

    Let $G=(V,E)$ be a boundaried graph with boundary $B\subseteq V$ containing a vertex $v_0\in B$. Let $\phi:B\rightarrow [k]_0$ be an injective mapping with $\phi(v_0)=0$. For a partial solution $S\subseteq V$ of $G$, we define the pattern $\pat_{G,B,\phi}(S)$ (omitting $G$, $B$ and $\phi$ when clear from context) of $G$ as follows: for each connected component $C$ of $G[S]$ that intersects $B$, we add the set $\phi(C\cap B)$ to $p$, i.e.\ $p$ contains a set for each connected component of $G\setminus S$, containing the labels that appear in the boundary of this component.

    It follows from the injectivity of $\phi$, that each label appears in at most one set of $\pat(S)$. Hence, $\pat(S)$ is a partition of a subset of $[k]_0$. Since $S$ is a partial solution, it holds that $v_0\in S$, and hence, there exists a set $X \in \pat(S)$ with $0\in X$. Hence, $\pat(S)$ is a pattern indeed. Given a boundaried forest $F$ with boundary $B$, we denote by $\pat_F$ the pattern $\pat_{F,B,\phi}(V(F))$.
    For a node $x\in\nodes$, we also denote $\pat_{G_x, B_{x}, \phi|_{B_{x}}}(S)$ by $\pat_x(S)$.
\end{definition}

\begin{definition}\label{def:tw-canonical-forest}
    Given a pattern $p\in\Pat$, we define the \emph{canonical (boundaried) forest} $(G_p=(V_p, E_p),B_p,\phi_p)$ of $p$, where for each $i\in L_p$ we add a vertex $v_i$, and for each set $S$ of $p$ we add another vertex $w_S$ and make it adjacent to all vertices $v_i$ with $i\in S$. Hence, $G_p$ is a disjoint union of stars, that contains one star for each set of $p$. We define $B := \{v_i\colon i\in L_p\}$ and $\phi(v_i) := i$ for all $i\in L_p$. It follows that $\pat(G_p) = p$.
\end{definition}

\begin{lemma}\label{lem:same-pattern-equiv-forests}
    It holds for any two labeled forests $(F_1,B_1,\phi_1)$ and $(F_2,B_2,\phi_2)$ with $\pat(F_1)=\pat(F_2)$ that $F_1$ is equivalent to $F_2$.
\end{lemma}

\begin{proof}
    Let $(H,B,\phi)$ be some boundaried forest compatible with $F_1$. We show that $H$ is compatible with $F_2$ as well. The other direction would then hold by symmetry. Let $G_1$ be the glueing of $F_1$ and $H$. Let $\pi_1$ be the bijection from $B_1$ to $B$ and $\pi$ the bijection from $B_1$ to $B_2$ as defined in \cref{def:tw-pat}. We define $\pi_2:=\pi_1 \circ \pi^{-1}$. Then $\pi_2$ is a bijection from $B_2$ to $B$. Let $G_2$ be the glueing of $F_2$ and $H$ given by $\pi_2$. We claim that $G_2$ is acyclic. Assuming otherwise, let $C$ be a simple cycle of $G_2$.
    Let $s_1,\dots s_{\ell}$ be the maximal segments of $C$ that belong to $F_2$, and let $D_1,\dots,D_{\ell}$ be connected components of $F_2$ that $s_1,\dots s_{\ell}$ belong to. Since it holds that $\pat(F_1)=\pat(F_2)$, it holds for each connected component $D$ of $F_2$ and for $X=\phi_2(D\cap B_2)$, that there exists a connected component $D'$ of $F_1$ with $\phi_1(D'\cap B_1) = X$. Hence, we can replace each segment $s_i$ connecting two vertices $u_2$ and $v_2$ of $F_2$ with a path of $F_1$ connecting the unique vertices $\pi(u_2)$ and $\pi(v_2)$, turning $C$ into a cycle $C'$ in $G_1$, which contradicts the assumption that $H$ is compatible with $F_1$.
\end{proof}

Now we define pattern operations that correspond to the different manipulations of a pattern corresponding to a partial solution under the different treewidth operations. We use these operations to introduce our algorithm and prove its correctness.

\begin{definition}\label{def:tw-patadd}
    For two different labels $i,j\in[k]$, we define the operation $\patadd_{i,j}\colon\TP\rightarrow\TP\cup\{\perp\}$ that takes as input a pattern $p$ and outputs $\perp$ if $i$ and $j$ belong to the same set of $p$, it outputs $p$ itself if $i$ or $j$ doesn't appear in $p$, or outputs the pattern that results from $p$ by merging the sets that contain $i$ and $j$ otherwise.
\end{definition}

\begin{observation}\label{obs:tw-patadd}
    Let $F$ be a labeled forest, and $p=\pat(F)$. Let $F'$ be the graph resulting from $F$ by adding the edge between the vertices labeled $i$ and $j$ if both exist, and let $p' = \patadd_{i,j}(p)$. Then $p' =\perp$ if and only if $F'$ contains a cycle, and $p':=\pat_{F'}$ otherwise.
\end{observation}

\begin{definition}\label{def:tw-join}
    We define the \emph{join} operation over patterns $\join:\TP\times\TP\rightarrow \TP\cup\{\perp\}$, where for $p,q\in\TP$ with $L_p = L_q$, $p\join q$ is given by the following procedure:
    Let $p_0 = p$ and let $S_1,\dots S_{\ell}$ be the sets of $q$. For each $i$ in $[\ell]$, we define $p_i$ from $p_{i-1}$ as follows: if $p_{i-1} = \perp$, or if there exists a set of $p_{i-1}$ containing two elements of $S_i$, then $p_i:=\perp$. Otherwise, $p_i$ results from $p_{i-1}$ by combining all sets that contain an element of $S_i$. We define $p\join q = p_{\ell}$. Finally, we define $p\join q = \perp$ if $L_p\neq L_q$.
\end{definition}

\begin{observation}
    Let $(G_1, B_1, \phi_1)$ and $(G_2, B_2, \phi_2)$ be two labeled forests where the images of $\phi_1$ and $\phi_2$ are equal. Let $p_1 = \pat_{G_1}$ and $p_2 = \pat_{G_2}$. Then it holds that $p_1\join p_2 = \perp$ if and only if the glueing $H$ of $G_1$ and $G_2$ contains cycles. Moreover, if $H$ is acyclic, then it holds that $\pat_H = p_1\join p_2$.
\end{observation}

\subsection{Pattern representation}

Now we introduce a special (restricted) family of patterns that represents all patterns, and show that this family suffices to represent acyclicity over a tree decomposition.

\begin{definition}\label{def:pat-compt-equiv}
    We say that two patterns $p,q$ are \emph{compatible} ($p\cmptb q$), if their canonical forests are. Given two families of treewidth patterns $S$ and $R$, we say that $R$ \emph{represents} $S$ ($R\sim S$), if their families of canonical forests are equivalent, i.e.\ if it holds for each pattern $r$ that the number of patterns of $S$ compatible with $r$ is congruent to the number of patterns of $R$ compatible with $r$.
\end{definition}

\begin{observation}\label{obs:tw-interwined-join-cmptb}
    It holds that representation is an equivalence relation. Moreover, it holds for three patterns $p,q,r \in \TP$ that $p\join q \cmptb r$ if and only if $p\cmptb q \join r$, since graph glueing is both commutative and associative.
\end{observation}

\begin{lemma}\label{lem:tw-delta-preserves-equiv}
    Let $S_1,S_2,R_1,R_2\subseteq \TP$ be four families of treewidth patterns such that $R_1$ represents $S_1$ and $R_2$ represents $S_2$. Then it holds that $R_1\Delta R_2$ represents $S_1\Delta S_2$.
\end{lemma}

\begin{proof}
    Let $q\in\TP$ be some a pattern. Let $s_1,s_2,r_1,r_2$ be the numbers of patterns in $S_1,S_2,R_1,R_2$ compatible with $q$ respectively. Then it holds by definition of representation that $s_1\bquiv r_1$ and $s_2\bquiv r_2$, hence it holds that $s_1+s_2\bquiv r_1+r_2$, but $s_1+s_2$ is congruent (modulo $2$) to the number of patterns of $S_1\Delta S_2$ compatible with $q$, and $r_1+r_2$ is congruent to the number of patterns of $R_1\Delta R_2$ compatible with $q$. Hence, these two numbers are congruent as well, and $R_1 \Delta R_2$ represents $S_1\Delta S_2$.
\end{proof}

Now we show that for each family of patterns $S\subseteq \TP$, there exists a family of patterns $R\subseteq \CTP$ that represents $S$. We achieve this through a \emph{reduce} operation that follows the same logic as $\rednice$ procedure defined for clique-width, but restricted to treewidth patterns. Hence, the proofs of representation follow the same logic, but applied to treewidth glueing operations instead of clique-width extension compatibility.

\begin{lemma}\label{lem:tw-rep-by-swap}
    Let $p=\{S_1,\dots S_n\}$ be some pattern for $n\geq 3$, and let $X = S_1$, $Y=S_2$, and $Z=S_3$ be three different sets of $p$. Let $R = p\setminus\{X,Y,Z\}$ be the rest of $p$. We define the following four patterns:
    \begin{itemize}
        \item $p_1 = R\cup\{X \cup Y, Z\}$,
        \item $p_2 = R\cup\{X \cup Z, Y\}$,
        \item $p_3 = R\cup\{Y \cup Z, X\}$,
        \item $p_4 = R\cup\{X \cup Y \cup Z\}$.
    \end{itemize}
    Then it holds that $\{p_2, p_3, p_4\}$ represents $\{p_1\}$.
\end{lemma}

\begin{proof}
    We show for any pattern $q\in\TP$ that if $q$ is compatible with one of $p_1,p_2,p_3$, then $q$ is compatible with at least two of them. We also show that $q$ is compatible with $p_4$ if and only if $q$ is compatible with all $p_1,p_2,p_3$. It follows from both claims for each pattern $q\in\TP$, that $q$ is compatible with an even number of patterns out of $p_1,p_2,p_3,p_4$ which proves the lemma.

    Now we prove the first claim. Assume that $q$ is compatible with neither $p_1$ nor $p_2$. We show that, in this case, it cannot be compatible with $p_3$ either. The other two cases are symmetric. Let $G_1, G_2, G_3$ be the glueing of $G_{p_1}$, $G_{p_2}$ and $G_{p_3}$ with $G_q$ respectively, and let $C_1$ be a cycle in $G_1$ and $C_2$ be a cycle in $G_2$. We show that $G_3$ contains a cycle proving that $p_3$ is not compatible with $q$.

    If $C_1$ does not visit vertices of the connected component corresponding to $X\cup Y$, then $C_1$ corresponds to a cycle in $G_3$ by replacing each maximal segment $s$ of $C$ belonging to a connected component of $G_1$ by a path on $G_3$ whose endpoints have the same labels as the endpoints of $s$. This is possible, since each set of $p_1$ other than $X\cup Y$ is a subset of a set of $p_3$. Similarly, if $C_2$ does not contain vertices in the connected component corresponding to $X\cup Z$, then $C_2$ corresponds to a cycle in $G_3$. Now we assume that $C_1$ passes through vertices in the connected component corresponding to $X\cup Y$ and $C_2$ passes through vertices in the connected component corresponding to $X\cup Z$. Hence, we can assume that $G_1$ contains a path between two vertices of the component corresponding to $X\cup Y$ that intersects this component only in its endpoints. If there exists such a path where the labels of both endpoints belong to the same set $X$ or $Y$, then $C_1$ corresponds to a cycle in $G_3$ by the same argument above. Hence, we assume that there exists a path with one endpoint labeled in $X$ and the other in $Y$. Similarly, we can assume that $G_2$ contains a path between two vertices whose labels are in $X$ and $Z$, that intersect the component corresponding to $X\cup Z$ only in its endpoints. However, these two paths correspond to paths in $G_3$ ending in the same labels, since each set of $p_1$ different from $X\cup Y$, and each set of $p_2$ different from $X\cup Z$ is a subset of a set of $p_3$. Hence, we get the cycle consisting of the segments $s_1,s_2,s_3,s_4$, where $s_1$ is a path between a vertex $v$ whose label belongs to $X$ and a vertex $u$ whose label belongs to $Y$, $s_3$ is a segment between a vertex $w$ whose label belongs to $Z$ and a vertex $v'$ whose label belongs to $X$, $s_2$ is a path between $u$ and $w$ that lays completely in the star corresponding to $Y\cup Z$ in $G_3$, and $s_4$ is a path between $v$ and $v'$ that lays in the star corresponding to $X$ in $G_3$. Hence, $G_3$ contains a cycle, and $p_3$ is not compatible with $q$.

    Now we prove the second claim. Let $G_4$ be the glueing of $G_{p_4}$ and $G_q$. If $G_i$ contains a cycle for any value $i\in[3]$, then one can turn this cycle into a cycle in $G_4$, since each set of $p_i$ is a subset of a set in $p_4$. On the other hand, if $G_4$ contains a cycle $C$, then we can first assume that this cycle intersects the component $S$ corresponding to $X\cup Y\cup Z$ in a single segment by short-cutting the cycle as follows: First rotate the cycle so that it doesn't begin with a vertex of $S$. Let $v_1$ be the first vertex on $C$ that belongs to $S$ and $v_2$ be the last such vertex. Then we can replace the cycle with another cycle that intersects $S$ in a single segment by replacing the part between $v_1$ and $v_2$ by a simple path in $S$.
    
    Let $s$ be this segment, and let $v_1$ be the first vertex of this segment, and $v_2$ be the last. Let $i_1,i_2$ be the labels of $v_1$ and $v_2$ respectively. Then one can replace $C$ by a cycle in $G_1$ if $\{i_1,i_2\} \subseteq X\cup Y$, by a cycle in $G_2$ if $\{i_1,i_2\}\subseteq X\cup Z$, or by a cycle in $G_3$ if $\{i_1,i_2\}\subseteq Y\cup Z$, by replacing $s$, and each other maximal segment on the cycle belonging to a connected component of $G_4$, by a path between two vertices of the same labels. This is possible since all other sets of $p_4$ different from $X\cup Y\cup Z$ are sets of $p_i$ for all $i\in[3]$.
\end{proof}

\begin{procedure}
    Given a pattern $p\in\TP$ and $i\in [k]$, we define the mapping $\redind(p,i)$ that outputs a set of patterns as follows: if $i$ does not belong to $p$, or if the set containing $i$ in $p$ is $Z_p$ or is a singleton, then we define $\redind(p,i):= p$. Otherwise, let $X$ be the set that contains $i$ in $p$, let $\tilde{X} = X\setminus \{i\}$, and let $p_0 = p\setminus\{X, Z_p\}$.
    Then we define $\redind(p,X,i)$ as the set containing the three patterns:

    \begin{itemize}
        \item $p_2 = p_0\cup\big\{Z_p\cup \{i\}, \tilde{X}\big\}$,
        \item $p_3 = p_0\cup\big\{Z_p \cup \tilde{X}, \{i\}\big\}$,
        \item $p_4 = p_0\cup\big\{Z_p \cup \tilde{X} \cup \{i\}\big\}$.
    \end{itemize}

    For $p\in\Pat$, we define $\redpat(p)$ as follows: Let $P_0 = \{p\}$, and let $i_1,\dots i_{\ell}$ be the indices in $[k]$ such that $i_j$ appears in $p$ in a non-singleton set different from $Z_p$ for all $j\in[\ell]$. Then we define $P_{j} := \bigdelta\limits_{q\in P_{j-1}} \redind(q, i_j)$ for each $j\in[\ell]$, and $\redpat(p):=P_{\ell}$.
\end{procedure}

\begin{lemma}\label{lem:tw-red-equiv-and-time}
    It holds for any pattern $p\in\TP$ that $\redpat(p)$ runs in time $\ostar(3^{\ell})$, where $\ell$ is the number of labels of $p$ that belong to a non-singleton set different from $Z_p$, and outputs a family of at most $3^{\ell}$ $\CTP$ patterns that represents $p$.
\end{lemma}

\begin{proof}
    The bound on the size of the resulting family and on the running time hold since the procedure runs in $\ell$ iterations, in each it replaces each pattern by at most $3$ patterns in polynomial time.

    It holds from \cref{lem:tw-rep-by-swap} that $\redind(p,i)$ represents $p$ for all $p\in\TP$ and all $i\in[k]$ by fixing $p_1 := p$.
    We prove by induction over $i\in[k]_0$, that $P_i$ represents $p$ and that for each $p\in P_i$ and each $j\leq i$, that $j$ appears in $p$ if and only if it appears in $p$, and that it can appear only in $Z_p$ or as a singleton.

    The claim is trivial for $j=0$ since $P_0 = \{p\}$. Now fix $j\in[\ell]$ and assume that the claim holds for $P_{j-1}$. Then it holds by \cref{lem:tw-delta-preserves-equiv} that $P_j$ represents $P_{j-1}$ and by transitivity of representation that $P_j$ represents $p$. Moreover, each pattern of $\redind(q,i_j)$ results from $q$ by moving $i_j$ either to the set $Z_q$, or to become a singleton, and possibly by merging the set containing $i_j$ with $Z_q$ which proves the claim.
\end{proof}

\subsection{The Algorithm}

Now we are ready to present our algorithm, as a bottom-up dynamic programming scheme over $T$ indexed by the family $\CTP$. Intuitively, the tables $P_x^{\budget}$ count (modulo $2$) the number of partial solutions of a fixed size at a node $x$, that are compatible with a pattern $p\in\CTP$.

\begin{algorithm}\label{alg:tw}
We define the tables $P_x^{\budget} \in \bin^{\CTP}$ for all nodes $x \in \nodes$, and all values $\budget\in[n]_0$ recursively as follows:
\begin{itemize}
    \item Introduce vertex $(v_0)$: We define $P_x^\budget[p] = 1$ if $\budget = 0$ and $p=\nopat$, and $0$ otherwise.
    
    \item Introduce vertex ($v\neq v_0$) with child $x'$: Let $i=\phi(v)$. We define $P_x^{\budget}[p] = P_{x'}^{\budget}[p\setminus i]$ if $\{i\}\in p$ or $i\notin L_p$, and $0$ otherwise.

    \item Forget vertex node ($v$) with child $x'$: Let $i=\phi(v)$. If $i \in L_p$, then we define $P_x^{\budget}[p] = 0$. Otherwise, we define
    \[
    P_x^{\budget}[p] = 
    \sumstack{p'\in\CTP\\p'\setminus i = p} P_{x'}^{\budget-[i\in L_{p'}]}[p'].
    \]

    \item Introduce edge ($\{u,v\}$) with child $x'$: Let $i=\phi(u)$ and $j=\phi(v)$. We define the tables $\tilde{P}_x^{\budget}\in \bin^{\TP}$ and $P_x^{\budget}\in\bin^{\CTP}$ with
    \[
    \tilde{P}_x^{\budget}[p] = \sumstack{q\in\CTP,\\\patadd_{i,j}(q) = p} P_{x'}^{\budget}[q]
    \quad\text{ and }\quad
    P_x^{\budget}[p] = \sumstack{q\in \TP,\\p\in\redpat(p)} \tilde{P}_x^{\budget}[q].
    \]

    \item Merge node with children $x_1, x_2$: For $p\in\CTP$ we define 
    \[
    P_x^{\budget}[p] = \sumstack{\budget_1+\budget_2=\budget}\sum\limits_{p_1\join p_2 = p} P_{x_1}^{\budget_1}[p_1] \cdot P_{x_2}^{\budget_2}[p_2].
    \]
\end{itemize}
\end{algorithm}

Now we show that the support of each table $P_x^{\budget}$ represents the family of all partial solutions of size $\budget$ at $x$. We prove this in two steps: First, we define the auxiliary tables $Q_x^{\budget}$ that count (modulo $2$) the number of partial solutions of size $\budget$ at $x$ having pattern $p$ for each pattern in $\TP$. In the second step, we show that the support of $P_x^{\budget}$ represents the support of $Q_x^{\budget}$ for all nodes $x\in\nodes$ and all values $\budget$.

\begin{definition}
    We define the tables $Q_x^{\budget}\in\bin^{\TP}$ for a nodes $x\in \nodes$ and $\budget\in[n]_0$ recursively as follows:
    \begin{itemize}
    \item Introduce vertex $(v_0)$: We set $Q_x^\budget[p] = 1$ if $\budget = 0$ and $p=\nopat$, and $0$ otherwise.
    
    \item Introduce vertex ($v\neq v_0$) with child $x'$: Let $i=\phi(v)$. We define $Q_x^{\budget}[p] = Q_{x'}^{\budget}[p\setminus i]$ if $\{i\}\in p$ or $i\notin L_p$, and $0$ otherwise.

    \item Forget vertex node ($v$) with child $x'$: Let $i=\phi(v)$.
    If $i\in L_p$, we define $Q_x^{\budget}[p] = 0$. Otherwise, we define
    \[
    Q_x^{\budget}[p] = 
    \sumstack{p'\in\TP\\p'\setminus i = p} Q_{x'}^{\budget-[i\in L_{p'}]}[p'].
    \]

    \item Introduce edge ($\{u,v\}$) with child $x'$: Let $i=\phi(u)$ and $j=\phi(v)$. We define
    \[
    Q_x^{\budget}[p] = \sumstack{q\in\TP,\\\patadd_{i,j}(q) = p} Q_{x'}^{\budget}[q].
    \]
    
    \item Merge node with children $x_1, x_2$: For $p\in\TP$ we define 
    \[Q_x^{\budget}[p] = \sumstack{\budget_1+\budget_2=\budget}\sum\limits_{p_1\join p_2 = p} Q_{x_1}^{\budget_1}[p_1] \cdot Q_{x_2}^{\budget_2}[p_2].\]
    \end{itemize}
\end{definition}

\begin{lemma}\label{lem:tw-Q-tables-count-sol}
    It holds for all nodes $x\in\nodes$, for all values $\budget$ and for all $p\in \TP$ that $Q_x^{\budget}[p]\bquiv 1$ if and only if there exists an odd number of partial solutions $X$, such that $|X\setminus B_x|=\budget$ and $\pat_x(X) = p$.
\end{lemma}

\begin{proof}
    We prove the claim by induction over the nodes $x\in\nodes$. For a node introducing the vertex $v_0$, the only solution is the empty solution, with $\pat_x(\emptyset) = \nopat$ which corresponds to the definition of the tables $Q_x^{\budget}$ at this node.

    For an introduce vertex node $v$ with $\phi(v)=i$, and $x'$ the child of $x$, each partial solution $X'$ with $|X'\setminus B_x = \budget$ at $x'$ and $\pat_{x'}(X') = p'$ is a partial solution at $x$, where $\pat_{x}(X')$ results from $p'$ by adding $i$ as a singleton. Moreover, $X'$ can be extended to a partial solution $X$ such that $|X\setminus B_{x'}| = |X'\setminus B_x|$ by adding $v$ to the partial solution. In this case it holds that $\pat_{x'}(X') =\pat_x(X)$. Since these two cases cover all partial solutions at $x$, and no partial solution at $x$ has the label $i$ in a non-singleton set, the correctness of the formula follows by induction hypothesis at $x'$.

    For a forget vertex node $v$ labeled $i$, each partial solution $X$ at $x$ is a partial solution at $x'$ where we exclude $i$ from the pattern since it is not in the boundary at $x$ anymore. Moreover, it holds that $|X\setminus B_x| = |X\setminus B_{x'}| + [v\in X]$ since $B_{x'}\setminus B_x = \{v\}$. The claim follows by induction hypothesis.
    
    For an add edge node $\{u,v\}$: Let $i,j$ be the labels of $u,v$ respectively. The lemma follows from \cref{obs:tw-patadd} and the induction hypothesis at the child node $x'$, since each partial solution $X$ at $x$ is a partial solution at $x'$ with $\pat_{x}(X) = \patadd_{i,j}(\pat_{x'}(X)$) and vice versa.

    Finally, for a join node $x$ with children $x_1$ and $x_2$: each partial solution $X$ at $x$ with $|X\setminus B_x|=\budget$ is the union of two partial solutions $X_1,X_2$ at $x_1$ and $x_2$ respectively, such that $X_1\cap B_x = X_2\cap B_x$ and $(X_1\setminus B_x)\cap (X_2\setminus B_x)=\emptyset$. Let $\budget_1 = |X_1\setminus B_{x_1}|$ and $\budget_2 = |X_2\setminus B_{x_2}|$. Then it holds that $\budget_1+\budget_2=\budget$. Moreover, it holds from \cref{obs:tw-interwined-join-cmptb} that $\pat_x(X) = \pat_{x_1}(X_1)\join \pat_{x_2}(X_2)$. The claim follows by applying the induction hypothesis at $x_1$ and $x_2$.
\end{proof}

Now we aim to show that the support of $P_x^{\budget}$ represents the support of $Q_x^{\budget}$. In order to show this we first prove that all pattern operations preserve representation, i.e.\ if $R$ represents $S$, then $\op(R)$ represents $\op(S)$.
We start with a formal definition:

\begin{definition}\label{def:tw-pres-rep}
    Let $\op$ be some pattern operator of arity $r$. We define the \emph{exclusive} version $\exec$ of $\op$ as the mapping of the same parity that maps sets of patterns to sets of patterns, where for families of patterns $S_1,\dots S_r$ we define
    \[\exec(S_1,\dots,S_r) = \bigdelta\limits_{(p_1,\dots, p_r)\in S_1\times\dots\times S_r} \{\op(p_1,\dots,p_r)\},\]
    i.e., a pattern $p$ belongs to $\exec(S_1,\dots, S_r)$ if and only if there exists an odd number of tuples $(p_1,\dots, p_r)$ that map to $p$ with $p_i\in S_i$ for all $i\in[r]$.

    For $i,j\in[k]$ with $i\neq j$ we define $\exop{\cup} i$ as the exclusive version of $\cup i$, $\exop{\join}$ as the exclusive join $\join$, $\exop{\setminus}i$ as the exclusive version of $\setminus i$ and $\exop{\patadd}_{i,j}$ as the exclusive version of $\patadd_{i,j}$.

    Let $\op$ be an operation of parity $r$ over sets of patterns. Then we say that $\op$ \emph{preserves representation} if it holds for all sets of patterns $S_1,\dots S_r$ and $R_1,\dots, R_r$ where $R_i$ represents $S_i$ for all $i\in[k]$ that $\exec(R_1,\dots R_r)$ represents $\exec(S_1,\dots, S_r)$.
\end{definition}

\begin{lemma}\label{lem:ops-preserve-representation}
    All operations $\exop{\cup} i$, $\exop{\setminus} i$ and $\exop{\patadd}_{i,j}$ for all $i,j\in[k]$, $i\neq j$ preserve representation.
\end{lemma}

\begin{proof}
    First, we make the following claim: Let $\op$ be some unary operation over patterns such that for each pattern $q$ there exists a set of patterns $X$, where $\op(p)$ is compatible with $q$ if and only if $p$ is compatible with an odd number of patterns of $X$. Then $\exec$ preserves representation.
    Let $S,R\subseteq \TP$ such that $R$ represents $S$, and let $q$ be an arbitrary pattern. Let $X$ be the family described above. Then it holds that
    \begin{alignat*}{2}
        &|\{p\in \exec(S)\colon p\cmptb q\}|
        &&\bquiv|\{p\in S\colon \op(p)\cmptb q\}|\\
        &\bquiv |\{p\in S, q' \in X \colon p\cmptb q'\}|
        &&\bquiv \sumstack{q'\in X}|\{p\in S\colon p\cmptb q'\}|\\
        &\bquiv \sumstack{q'\in X}|\{p\in R\colon p\cmptb q'\}|
        &&\bquiv |\{p\in R, q' \in X \colon p\cmptb q'\}|\\
        &\bquiv |\{p\in R\colon \op(p)\cmptb q\}|
        &&\bquiv |\{p\in \exec(R)\colon p\cmptb q\}|.
    \end{alignat*}

    Now we use this claim to prove the lemma. For $p\cup i$ and a pattern $q$, let $X:=\emptyset$ if $i \notin L_q$, and $X:=\{p\setminus i\}$ otherwise. In the former case, it holds that $p\cup i \not\cmptb q$ for all patterns $p$, and $p$ is not compatible with any pattern of $X$ (since $X$ is empty). For the latter case, if $i\in L_p$, then it holds that $p\cup i = \perp$.
    If this is the case, or if $L_p \neq L_q\setminus \{i\}$, then it holds that $p\cup i \not\cmptb q$ and $p\not\cmptb q\setminus i$.
    Otherwise, we assume that $i\notin L_p$ and $L_q = L_p \cup i$. 
    Since $G_{p\cup i}$ results from $G_p$ by adding the vertices $v_i,w_{i}$ with an edge between them as a new connected component, with $v_i$ labeled $i$, the glueing of $G_{p\cup i}$ and $G_q$ results from the glueing of $G_p$ and $G_{q\setminus i}$ by adding the edge $\{v_i, w\}$ and an edge between $v_i$ and the vertex $w_S$ where $S$ is the set containing $i$ in $q$, i.e., by hanging a path ($P_2$) at $w_S$. Hence, the resulting graph is acyclic if and only if the original graph is.

    For $p\setminus i$: if $i\in L_q$ then $p\setminus i \not\cmptb q$ for all patterns $p$. In this case, we define $X=\emptyset$.
    Otherwise, assume that $i\notin L_q$. We define $X=\{q, q\cup i\}$. Since $L_q\neq L_{q\cup i}$, it holds for each pattern $p$ that $p$ is compatible with at most one of $q$ and $q\cup i$. If $i\notin L_p$, then it holds that $p\setminus i = p$, and the claim is trivial in this case. Otherwise, it holds that $i\in L_p$ and hence, $p\not\cmptb q$. If $L_p \neq L_q \cup \{i\}$, then $p\setminus i \not\cmptb q$ and $p\not\cmptb q\cup i$. Otherwise,
    the glueing of $G_p$ and $G_{q\cup i}$ results from the glueing of $G_{p\setminus i}$ and $G_q$ by hanging a path ($P_2$) at the vertex $w_S$ corresponding to the set containing $i$ in $p$. Hence, the resulting graph is acyclic if and only if the original graph is.

    Finally, for $\patadd_{i,j}(p)$, let $p':=\patadd_{i,j}(p)$ and $q'=\patadd_{i,j}(q)$. We define $X:=\{q'\}$. If $L_p\neq L_q$ then it holds that $\patadd_{i,j}(p)\not\cmptb q$ and $p\not\cmptb q'$. Otherwise, if $\{i,j\}\not \subseteq L_p$ then it holds that $p'=p$ and $q' = q$. Hence, the claim is trivial in this case. Finally, assume that $i,j\in L_p = L_q$. Let $r := \big\{\{i,j\}\big\}\cup \{\{k\}\colon k\in L_p\setminus\{i,j\}\}$ be the pattern consisting out of the set $\{i,j\}$ and all other labels of $L_p$ as singletons. Then it holds that $G_{p'}$ is the glueing of $G_{r}$ and $G_{p}$, and that $G_{q'}$ is the glueing of $G_r$ and $G_q$. 
    Since the glueing process is both associative and commutative, it holds that the glueing of $G_{p'}$ and $G_q$ is the same as the glueing of $G_p$ and $G_{q'}$. Hence, $p'\cmptb q$ if and only if $p\cmptb q'$.
\end{proof}

\begin{lemma}\label{lem:tw-join-pres-rep}
    Let $S_1,S_2,R_1,R_2\subseteq \TP$ be set of patterns such that $R_1$ represents $S_1$ and $R_2$ represents $S_2$. Let $S:=S_1\exop{\join} S_2$ and $R:=R_1\exop{\join} R_2$. Then $R$ represents $S$.
\end{lemma}

\begin{proof}
    It holds for each pattern $r\in\TP$ 
    \begin{alignat*}{2}
        &|\{p\in R\colon p\cmptb r\}|
        &&\bquiv|\{(q_1,q_2)\in R_1\times R_2\colon q_1\join q_2 \cmptb r\}|\\
        &\bquiv \sum_{q_2\in R_2} |\{q_1\in R_1\colon q_1\join q_2 \cmptb r\}|
        &&\bquiv \sum_{q_2\in R_2} |\{q_1\in R_1\colon q_1\cmptb q_2\join r\}|\\
        &\bquiv \sum_{q_2\in R_2} |\{p_1\in S_1\colon p_1\cmptb q_2\join r\}|
        &&\bquiv \sum_{p_1\in S_1} |\{q_2\in R_2\colon p_1\cmptb q_2\join r\}|\\
        &\bquiv \sum_{p_1\in S_1} |\{q_2\in R_2\colon q_2\cmptb p_1\join r\}|
        &&\bquiv \sum_{p_1\in S_1} |\{p_2\in S_2\colon p_2\cmptb p_1\join r\}|\\
        &\bquiv \sum_{p_1\in S_1} |\{p_2\in S_2\colon p_1\join p_2\cmptb r\}|
        &&\bquiv|\{(p_1,p_2)\in S_1\times S_2\colon p_1\join p_2 \cmptb r\}|\\
        & &&\bquiv |\{p\in S\colon p\cmptb r\}|,
    \end{alignat*}
    where the third, the sixth and the seventh congruences follow from \cref{obs:tw-interwined-join-cmptb}.
\end{proof}

\begin{lemma}\label{lem:tw-P-Q-equiv}
    It holds for all nodes $x$ and all values $\budget$ that the support of $Q_x^{\budget}$ represents the support of $P_x^{\budget}$.
\end{lemma}

\begin{proof}
    We prove the claim by induction over $\nodes$. The claim is trivial at nodes introducing the vertex $v_0$, as both tables have the same support.
    Let $I_x^{\budget}$ be the support of $P_x^{\budget}$, and $J_x^{\budget}$ be the support of $Q_x^{\budget}$.
    For an introduce vertex node with child $x'$, let $i$ be the label of the introduced vertex. It holds that $I_x^{\budget} = (I_{x'}^{\budget}\exop{\cup} i) \Delta I_{x'}^{\budget-1}$ and $J_x^{\budget} = (J_{x'}^{\budget}\exop{\cup} i) \Delta J_{x'}^{\budget-1}$. The claim follows from the fact that both $\Delta$ and $\cup i$ preserve representation.
    Similarly, for a forget vertex node with label $i$, it holds that $I_x^{\budget} = I_{x'}^{\budget}\exop{\setminus} i$ and $J_x^{\budget} = J_{x'}^{\budget}\exop{\setminus} i$. The claim follows from the fact that $\setminus i$ preserves representation.

    Now let $x$ be an introduce edge $\{u,v\}$ node, with $\phi(u)=i$ and $\phi(v)=j$.
    It holds that $J_x^{\budget} = \exop{\patadd}_{i,j}(J_{x'}^{\budget})$. Let $\tilde{Y} = \exop{\patadd}_{i,j}(I_{x'}^{\budget})$. It holds from \cref{lem:ops-preserve-representation} that $\tilde{Y}$ represents $J_x^{\budget}$. Moreover, it holds that $I_x^{\budget} = \redpat(\tilde{Y})$. It follows from \cref{lem:tw-red-equiv-and-time} that $I_x^{\budget}$ represents $\tilde{Y}$ and hence, by transitivity of representation, it represents $J_x^{\budget}$.
    Finally, let $x$ be a union node with children $x_1, x_2$. Let $\budget_1,\budget_2\in[n]_0$ with $\budget_1 + \budget_2=\budget$, and let $I^{b_1,b_2} = I_{x_1}^{\budget_1} \exop{\join} I_{x_2}^{\budget_2}$ and $J^{b_1,b_2} = J_{x_1}^{\budget_1} \exop{\join} J_{x_2}^{\budget_2}$. Then it holds by \cref{lem:tw-join-pres-rep} that $I^{b_1,b_2}$ represents $J^{b_1,b_2}$. Moreover, it holds by the definitions of $P$ and $Q$ that $I_x^{\budget} = \bigdelta\limits_{\budget_1+\budget_2=\budget} I^{b_1,b_2}$ and $J_x^{\budget} = \bigdelta\limits_{\budget_1+\budget_2=\budget} J^{b_1,b_2}$. It follows from \cref{lem:tw-delta-preserves-equiv} that $I_x^{\budget}$ represents $J_x^{\budget}$.
\end{proof}

\begin{corollary}\label{cor:tw-algo-correct}
    It holds for each value $\budget\in[n]_0$ that $P_r^{n-\budget}[\nopat] = 1$ if and only if there exists an odd number of feedback vertex sets of size $\budget$ in $G$.
\end{corollary}

\begin{proof}
    The number of feedback vertex sets of size $\target$ in $G$ is equal to the number of feedback vertex sets of size $\target$ in $G_r$ that exclude the vertex $v_0$, which is equal to the number of partial solutions $S$ of size $n + 1 - \target$ at $r$ that contain the vertex $v_0$. Since $B_r = \{v_0\}$, it holds that $\pat_r(S) = \nopat$ for all such partial solutions $S$, and that $|S\setminus B_r| = n-\target$. It follows by \cref{lem:tw-Q-tables-count-sol} that the number of partial solutions of size $n+1 - \target$ at $r$ that contain the vertex $v_0$ is congruent (modulo $2$) to $Q_r^{n - \target}[\nopat]$.
    
    Since the pattern $\nopat$ is the only pattern compatible with $\nopat$, it holds by \cref{lem:tw-P-Q-equiv} that $P_r^{n- \target}[\nopat] \bquiv 1$ if and only if $Q_r^{\target}[\nopat] \bquiv 1$.
    It follows that $P_r^{n-\target}[\nopat] = 1$ if and only if there exists an odd number of feedback vertex sets of size $\target$ in $G$.
\end{proof}

\subsection{Fast join operation}

Now we define a set of states $\mathcal{Z}$, and a canonical bijection $\phi_{\mathcal{Z}}$ between $\CTP$ and $\mathcal{Z}$. We define a join operation $\twstjoin$ over $\mathcal{Z}$, such that $\phi$ is an isomorphism that preserves the join operations between $\CTP$ and $\mathcal{Z}$. We show how to compute the convolution over $\twstjoin$ in time $\ostar(3^k)$. As a result, we show that a single join node can be processed in time $\ostar(3^k)$.

\begin{definition}
    Let $Z_0 = \{\stnone, \stdisc, \stconn\}$ be a set of \emph{simple} states. We define the family of states $\mathcal{Z} := Z_0^{k}$, and the bijection $\phi_{\mathcal{Z}}$ between $\CTP$ and $\mathcal{Z}$, where we map each pattern $p\in\CTP$ to the state $z$ defined by assigning to each label $i\in[k]$ the state $\stnone$ if $i\notin L_p$, $\stdisc$ if $i\in L_p \setminus Z_p$, or $\stconn$ otherwise.
\end{definition}
    
\begin{observation}
    It holds by the bijection $\phi_{\mathcal{Z}}$ that $|\CTP| = 3^k$.
\end{observation}

\begin{remark}
    Given a pattern $p\in\CTP$, the bijection $\phi_{\mathcal{Z}}(p)$ implicitly assigns a state to each label $i\in[k]$.
     The family $\mathcal{Z}$ is actually the restriction of $\states$ (\cref{def:csp-states-bijection}) to the simple states in $Z_0$. Intuitively, compared to clique-width, each label appears in a treewidth boundary at most once. This allows us to spare the so called ``composed'' states $\stddisc$, $\stplus$ and $\stdconn$. Accordingly, $\phi_{\mathcal{Z}}$ can be seen as the restriction of the bijection $\phi_{CS}$ from \cref{def:csp-states-bijection}.
\end{remark}

\begin{definition}
    We define the binary operation $\twsmpjoin:Z_0\times Z_0\rightarrow Z_0\cup\{\perp\}$ as given by the following table:
    \begin{center}
    \begin{tabular}{|c|c|c|c|}
    \hline
    \twsmpjoin &\stnone&\stdisc&\stconn\\
    \hline
    \stnone&\stnone&$\perp$&$\perp$\\
    \hline
    \stdisc&$\perp$&\stdisc&\stconn\\
    \hline
    \stconn&$\perp$&\stconn&$\perp$\\
    \hline
\end{tabular}
\end{center}
We define the operation $\twstjoin$ as the $k$th power of $\twsmpjoin$, where we replace any tuple in the output by $\perp$, if any of its indices is equal to $\perp$. Let $\twstconv$ be the convolution of $\twstjoin$ over $\bin$.

Let $(Z_0, \leq)$ be the ordering given by the single chain $\stdisc \leq \stconn$. We define the ordering $\preceq$ over $\CTP$ as the ordering the results from the $k$th power of $\leq$ by applying the bijection $\phi_{\mathcal{Z}}$, i.e.\ it holds that $p\preceq q$ if $p(i) \leq q(i)$ for all $i\in[k]$.
\end{definition}

\begin{observation}\label{obs:tw-join-isomorphism}
    The bijection $\phi_{\mathcal{Z}}$ is an isomorphism between the restriction of the operation $\join$ to $\CTP$ and the operation $\twstjoin$ over $\mathcal{Z}$.
\end{observation}

Now we show how to compute the convolution $\twstconv$ in time $\ostar(3^k)$ by utilizing a fast lattice convolution technique over power lattices from \cite{DBLP:conf/esa/HegerfeldK23}. We state the following lemma from Hegerfeld and Kratsch~\cite{DBLP:conf/esa/HegerfeldK23} that builds on the work of Björklund et al.~\cite{DBLP:journals/talg/BjorklundHKKNP16} for join lattices. We restate the lemma in terms that fit our paper. While the lemma in the original work is stated for lattices only, the proof also works for all orderings $(S, \preceq)$, where for any element $p \in S$ the set $\{q \in S \colon q \preceq p\}$ induces a join-semilattice. We refer to \cite{DBLP:conf/birthday/Rooij20} for a similar treatment.

\begin{lemma}[{\cite[Corollary A.10.]{DBLP:journals/corr/abs-2302-03627/HegerfeldK23}}]\label{lem:join-prod-time}
    Let $(S,\leq)$ be an ordering, such that $\{q\in S\colon q\leq p\}$ for each $p\in S$ is a join-semilattice, and let $(\mathcal{S},\preceq)$ be the $k$th power of this lattice. Let $T_1,T_2 \in \mathbb{F}^S$ be two tables for some ring $\mathbb{F}$. 
    Assuming both the join operation over $\preceq$ and basic math operations over $\mathbb{F}$ can be computed in polynomial time, then the join product $T_1\joinprod T_2$ can be computed in time $\ostar(|S|^k)$.
\end{lemma}

\begin{lemma}\label{lem:tw-conv-time}
    The convolution $\twstconv$ over $\bin$ can be computed in time $\ostar(3^k)$.
\end{lemma}

\begin{proof}
    Let us fix two tables $A_1,A_2\in\bin^{\mathcal{Z}}$. Our goal is to compute $A_1\twstconv A_2$ in time $\ostar(3^k)$.
    For a state $s\in\mathcal{Z}$, let $c_p := |s^{-1}(\stconn)|$ be the number of labels having state $\stconn$ in $s$. For each state $s\in\mathcal{Z}$ and each value $i\in[k]_0$, we define the predicate $\chi_i(s) := (c_s = i)$, i.e.\ $\chi_i(s)$ holds if and only if $s$ has exactly $i$ indices with state $\stconn$.
    For each value $i\in[k]$, let $A_1^i,A_2^i\in\bin^{\mathcal{Z}}$ be the tables defined from $A_1$ and $A_2$ respectively, where $A^i_x[s]=[\chi_i(s)]\cdot A_x[s]$ for $x\in\{1,2\}$ and $s\in \mathcal{Z}$.

    Let $s_1,s_2\in\mathcal{Z}$, and let $s:=s_1\lor s_2$. We claim that $s_1\lor s_2 = s_1\twstjoin s_2$ if and only if $c_{s_1} + c_{s_2} = c_s$, and that $p_1\twstjoin p_2 = \perp$ otherwise. 
    The inequality $c_{s_1}+c_{s_2}\geq c_s$ holds, since for each index $i$ with $s_i = \stconn$ it must hold that at least on of $s_1$ and $s_2$ has state $\stconn$ at the index $i$. If the equality does not hold, then there must exist a position $i\in[k]$ with $s_1(i) = s_2(i) = \stconn$, and hence, $s_1\twsmpjoin s_2=\perp$. Now assume that the equality holds. If $s_1\twstjoin s_2\neq s$, then there exists $i\in[k]$ with $s_1(i)=s_2(i)=\stconn$, but then there must exist an index $i'$ with $s(i') =\stconn$ and both $s_1(i')$ and $s_2(i')$ different from $\stconn$, which contradicts the assumption that $s = s_1\lor s_2$, which proves the claim.

    It follows from this claim that
    \begin{align*}
        A_1\twstconv A_2[s] &= \sumstack{s_1,s_2\in\mathcal{Z}\\s_1\twstjoin s_2 = s} A_1[s_1] \cdot A_2[s_2]\\
        &= \sumstack{s_1,s_2\in\mathcal{Z}\\s_1\lor s_2 = s}[c_{s_1}+c_{s_2} = c_s] A_1[s_1] \cdot A_2[s_2]\\
        &= \sum\limits_{i=0}^k{c_s} \sumstack{s_1,s_2\in\mathcal{Z}\\s_1\lor s_2 = s} A_1^i[s_1] \cdot A_2^{c_s-i}[s_2]\\
        &= \sum\limits_{i=0}^k (A_1^i \joinprod A_2^{c_s-i})[s].
    \end{align*}

    Therefore, $A_1\twstconv A_2$ can be computed by computing all join products $A_1^i\joinprod A_2^j$ for $i,j\in[k]_0$. It follows from \cref{lem:join-prod-time} that $A_1\twstconv A_2$ can be computed in time $\ostar(3^k)$.
\end{proof}

\subsection{Running time}

\begin{lemma}\label{lem:tw-join-time}
    Let $x$ be a join node.
    Given all tables $P_{x'}^{\budget'}$ for all children $x'$ of $x$ and all values $\budget'$, the tables $P_x^{\budget}$ for all values of $\budget$ can be computed in time $\ostar(3^k)$.
\end{lemma}

\begin{proof}
    We iterate over all values $\budget \in [n]_0$, and for each we compute the table $P_x^{\budget}$ as follows: we iterate over all pairs $\budget_1,\budget_2\in[n]_0$ with $\budget_1+\budget_2=\budget$, and we compute the convolution of the tables $P_{x_1}^{\budget_1}$ and $P_{x_2}^{\budget_2}$ over $\join$, where $x_1,x_2$ are the children of $x$. We add the resulting table to $P_x^{\budget}$. The correctness follows by the definition of the table $P_x^{\budget}$ at a join node (\cref{alg:tw}).

    It holds by \cref{obs:tw-join-isomorphism} that $\phi_{\mathcal{Z}}$ is an isomorphism between the operation $\join$ restricted to $\CTP$ and the operation $\twstjoin$ over $\mathcal{Z}$, and hence, by \cref{lem:tw-conv-time} and \cref{obs:iso-same-time} that the convolution over $\join$ (restricted to $\CTP$) can be computed in time $\ostar(3^{k})$.
    Since for each value $\budget$ we iterate over at most $n+1$ pairs
    $(\budget_1,\budget_2)$ with $\budget_1+\budget_2=\budget$, and since we iterate over $n+1$ values of $\budget$, the entire computation runs in time
    $\ostar(3^k)$.
\end{proof}

\begin{corollary}\label{cor:tw-total-time}
    All tables $P_x^{\budget}$ for all $x\in\nodes$ and $\budget\in[n]_0$ can be computed in time $\ostar(3^{\tw})$.
\end{corollary}

\begin{proof}
    We achieve this through a bottom up dynamic programming over $\nodes$. Since $|\nodes|$ is bounded polynomially in $n$, and since we iterate over $n+1$ values of $\budget$, it suffices to show that the table $T_x^{\budget}$ for a fixed node $x$ and a fixed value $\budget$ can be computed in time $\ostar(3^{\tw})$ given all tables $T_{x'}^{\budget'}$ for each child $x'$ of $x$ and each value $\budget'$. Also, since $|\CTP|=6^k$, we assume that one can iterate over all entries $T_{x'}^{\budget'}[p]$ for all $p\in\CTP$ in time $\ostar(3^k)$. 
    
    Now we prove the claim by distinguishing the type of the a node $x\in\nodes$.
    For a node $x$ introducing the vertex $v_0$, this is trivial since we only set constant values.
    For a node introducing a vertex $v\neq v_0$, this can be done by iterating over all patterns $p\in\CSP$ and by adding the value of $P_{x'}^{\budget}[p]$ to $P_{x}^{\budget}[p]$, and the value of $P_{x'}^{\budget}[p]$ to $P_{x}^{\budget}[p\cup i]$ if $p\cup i\neq \perp$.
    Similarly, for a node forgetting a vertex $v$ labeled $i$, we iterate over all patterns $p\in\CTP$. We add to $P_x^{\budget}[p\setminus i]$ the value of $P_{x'}^{\budget}[p]$ if $i \notin L_p$, or the value of $P_{x'}^{\budget-1}[p]$ otherwise.

    For a node introducing an edge $\{u,v\}$, we iterate over all patterns $p\in\TP$, we compute the pattern $q = \patadd_{i,j}^{-1}(p)$, and the set $S = \redpat(q)$. It holds by \cref{lem:tw-red-equiv-and-time} that $S$ can be computed in polynomial time, and has size at most $9$ patterns, since there are at most two labels that appear in a non-singleton set of $q$ different from $Z_q$, i.e.\ $\ell \leq 2$. We add $P_{x'}^{\budget}[p]$ to $P_x^{\budget}[s]$ for each $s\in S$.
    Finally, for a join node $x$, the claim follows from \cref{lem:tw-join-time}, since $k=\tw+1$.
\end{proof}

\begin{proof}[Proof of~\cref{theo:tw-count}]
    By \cref{rem:tw:dec-to-very-nice}, we can assume that the tree decomposition $\mathcal{B}'$ is very nice, and we can compute the decomposition $\mathcal{B}$ that results from $\mathcal{B}'$ by adding $v_0$ to each bag in polynomial time.
    It follows from \cref{cor:tw-total-time} that all tables $P_x^{\budget}$ for all $x\in\nodes$ and $\budget\in[n]_0$ can be computed in time $\ostar(3^{k}) = \ostar(3^{\tw})$, since $k=\tw+1$, and from \cref{cor:tw-algo-correct} that $P_r^{\budget}[\nopat]$ is congruent (modulo $2$) to the parity of the number of feedback vertex sets of size $\budget$ in $G$.
\end{proof}

\section{Connected Feedback Vertex Set}\label{sec:cfvs}

\subsection{Upper bound}

In this section, we show that our acyclicity representation combined with the ``isolating a representative'' technique of Bojikian and Kratsch~\cite{DBLP:conf/icalp/BojikianK24} result in a tight algorithm for the \Cfvsp with running time $\ostar(18^{\cw})$. We provide a SETH based matching lower bound. We start with some notation from \cite{DBLP:conf/icalp/BojikianK24}. In order to distinguish our notation, we will call the families $\Pat$ and $\CSP$ the families of acyclicity patterns and the family of very nice acyclicity patterns respectively.

Let $(G, \target)$ be the given instance. We can assume that $\target > 0$, since for $\target = 0$, we can check in polynomial time whether $G$ is a forest.
Let us fix a vertex $u_0\in V$. We aim to find a connected feedback vertex set of size $\budget$ and weight $\weight$, that contains the vertex $u_0$ (but excludes $v_0$ in $G_r$). In the final algorithm, we will iterate over all 
choices of vertices $u_0\in V$, and run the algorithm for each independently. Hence, if a solution $S$ of size $\target>0$ exists, then any choice of $u_0\in S$ will suffice. Therefore, we define a \emph{partial solution} at a node $x$ as a set $S\subseteq V_x$ that contains $u_0$ if $u_0 \in V_x$, but not $v_0$, such that $G_x-S$ is a forest.

\begin{definition}
    A \emph{connectivity pattern} $p$ is a set of subsets of $[k]_0$, such that exactly one set $Z_p$ of $p$ contains the element $0$ in it (called the zero set of $p$).
    We denote by $L_p := \bigcup_{S\in p}S\setminus\{0\}$ the set of labels of $p$. We call $\inc(p) := L_p\setminus Z_p$ the set of \emph{incomplete labels} in $p$.

    Given a partial solution $S\subseteq V_x$ at a node $x\in\nodes_x$, we define the connectivity pattern $\pat^C_x(S)$ induced by $S$ at $x$ as follows: for each connected component $C$ of $G_x[S]$, we add to $\pat^C_x(S)$ the set $\{\lab(v)\colon v\in C\}$, where we add $0$ to the set corresponding to the component that contains $u_0$ if $u_0 \in V_x$, or as a singleton otherwise.
    
    We call a connectivity pattern $p$ \emph{complete}, if $\inc(p)=\emptyset$, and we call a complete connectivity pattern $p$ \emph{nice}, if each set of $p$ different from $Z_p$ is a singleton. We denote by $\mathscr{C}$ the family of all connectivity patterns, by $\mathscr{C}_C$ the family of complete connectivity patterns, and by $\mathscr{C}_{CS}$ the family of nice connectivity patterns.
\end{definition}

Intuitively, a pattern is complete, if each label that appears in $p$ appears as a singleton as well. Connectivity patterns represent connectivity in partial solutions, where the zero set corresponds to the component containing $u_0$. This can be seen as an additional vertex of a special label, hanging to $u_0$ if $u_0\in V_x$ or as an isolated vertex otherwise. We use complete patterns to create existential representation of connectivity patterns, and nice pattern to create a representation of complete patterns in a counting sense (as defined in the previous sections of this paper).

We denote a pattern $\{\{i^1_1,\dots i^1_{r_1}\}, \dots \{i^t_1,\dots i^t_{r_t}\}\}$ by $[i^1_1i^1_2\dots i^1_{r_1},\dots, i^t_1i^t_2\dots i^t_{r_t}]$ when this does not cause confusion. We also sometimes omit the element $0$ when it is a singleton. For example, $[ij]$ denotes the patter $\{\{0\},\{i,j\}\}$, and $[0]$ denotes the pattern $\{\{0\}\}$.

\begin{observation}
    Let $S$ be a partial solution at a node $x \in \nodes$ with $u_0 \in S$. Then $G_x[S]$ is connected if and only if $\pat_x^C(S)$ is a singleton.
\end{observation}

Now we recall some of the pattern operations defined in \cite{DBLP:conf/icalp/BojikianK24}, that allow to extend patterns corresponding to all partial solutions of a specific size correctly over $\syntaxtree$.

\begin{definition}
    For two patterns $p,q\in\mathscr{C}$ and for $i,j\in[k]$ with $i\neq j$, we define the operations $p_{i\rightarrow j}$, $\patadd_{i,j}(p)$ and $p\punion q$ over connectivity patterns, where $(p)_{i\rightarrow_j}$ is the pattern that results from $p$ by replacing $i$ with $j$ in each set that contains $i$ in $p$, and $\patadd_{i,j}(p)$ is the pattern resulting from $p$ be combining all sets that contain $i$ or $j$, if both labels appear in $p$, or is the pattern $p$ itself otherwise. Finally, $p\punion q$ is the pattern resulting from the union of $p$ and $q$ by replacing their zero sets with the union of these two sets.
\end{definition}

In order to prove the correctness of our algorithm, we will follow the two step scheme used along this paper, by defining larger tables as a middle step, that count all partial solutions. We then show that the dynamic programming tables of our algorithm represent these larger tables. However, in order to keep simplicity, we will start by defining these larger tables (over complete patterns) first, and prove their correctness.

In \cite{DBLP:conf/icalp/BojikianK24}, the authors define the concept of \emph{actions}, as operations over patterns, that produce existential representation of a partial solution using complete patterns only. They introduced the notion of an \emph{action sequence} as the sequence of actions taken over $\syntaxtree$ to produce each pattern of such a representation. We will use these notions as a black box, and cite the corresponding results from the original paper.

\begin{definition}
    An \emph{action} is a mapping that takes as input a pattern $p$ and an integer $i\in[4]$ and outputs a complete pattern, or $\perp$. An \emph{action sequence} $\pi$ at a node $x\in\nodes$ is a mapping that assigns to each node $y\in\nodes_x$ a value in $[4]$. Given a partial solution $S\subseteq V_x$, an action sequence $\pi$ generates a complete pattern $p$ from $S$ or $\perp$.
\end{definition}

\begin{lemma}[{\cite[Lemma 29]{DBLP:conf/icalp/BojikianK24}}]
    \label{lem:cfvs-sol-if-action-seq}
    Given a partial solution $S$ at a node $x\in\nodes$ that contains the vertex $u_0$, the set $S$ induces a connected subgraph of $G$, if and only if there exists an action sequence that generates the pattern $[0]$ from $S$.
\end{lemma}

Let us fix a value $D$, and a weight function $\weightf'\colon \nodes\times[4]\rightarrow[D]$, that both will be chosen later, where $D$ is bounded polynomially in $n$. The function $\weightf'$ assigns a weight to each action sequence $\pi$ at a node $x$, defined as $\weightf'(\pi) := \sum_{y\in\nodes_x} \weightf'(y, \pi(y))$.

We restate the recursive formulas used to count action sequences from~\cite{DBLP:conf/icalp/BojikianK24}. We use wording that suits our presentation, and choose the set of terminals $\{u_0\}$, i.e.\ we do not force any vertex (other than $u_0$) into the solution.

\begin{definition}
    We define the tables $D_x^{\budget,\weight, d} \in \bin^{\mathscr{C}_C}$ recursively over $\syntaxtree$ as follows:
\begin{itemize}
    \item Introduce vertex $v$ ($\mu_x = i(v)$): If $v\neq u_0$, we set $D_x^{1,\weightf(v),0}[[0,i]]$ and $D_x^{0,0,0}[[0]]$ to $1$. Otherwise, let $p = [0i]$, and for $\ell\in[2]$, let $p_{\ell} = \action(p, \ell)$. Then we set $D^{1,\weightf(v), \weightf'(x, \ell)}[p_{\ell}]$ to $1$.
    We set all the other values of $D_x^{\budget, \weight, d}[p]$ to $0$.
    
    \item Relabel node $\mu_x = \relabel{i}{j}(\mu_{x'})$: We define 
    $D_x^{\budget,\weight, d}[p] = \sumstack{p' \in \mathscr{C}_C\\ p'_{i\rightarrow j} = p} D_{x'}^{\budget,\weight,d}[p']$.

    \item Join node $\mu_x = \clqadd{i}{j}(\mu_{x'})$: We define
    $D_x^{\budget,\weight, d}[p] = \sumstack{p'\in\Pat, \ell \in [4]\\\action(p',\ell)=p} \sumstack{p'' \in \mathscr{C}_C\\ \patadd_{i,j}(p'') = p'} D_{x'}^{\budget,\weight, d - \weightf'(x, \ell)}[p'']$.

    \item Union node $\mu_x = \mu_{x_1} \clqunion \mu_{x_\ell}$: We define
    \[
    D_x^{\budget,\weight, d}[p] = \sumstack{\budget_1+\budget_2 = \budget\\ \weight_1+\weight_2 = \weight\\ d_1 + d_2 = d}\sumstack{p_1,p_2\in\mathscr{C}_C\\ p = p_1\punion p_2} D_{x_1}^{\budget_1,\weight_1, d_1}[p_1] \cdot D_{x_2}^{\budget_2,\weight_2, d_2}[p_2].
    \]
\end{itemize}
\end{definition}

We cite the following result:

\begin{lemma}[{\cite[Lemma 31]{DBLP:conf/icalp/BojikianK24}}]
    It holds that $D_x^{\budget, \weight, d}[p] =1$ if and only if there exists an odd number of pairs $(S,\pi)$, with $S\subseteq V_x$ such that $v_0\notin S$, $u_0\subseteq S$, $|S|=\budget$, $\weightf(S) = \weight$, and $\pi$ is an action sequence of weight $d$ at $x$ producing the pattern $p$ from $S$.
\end{lemma}

Our goal is to count such pairs $(S,\pi)$, where additionally it holds that $G_x \setminus S$ is a forest. In order to do so, we combine connectivity patterns with acyclicity patterns in our states.

\begin{definition}
    Let $\mathcal{I} := \Pat\times \mathscr{C}_C$ be the set of all combinations of an acyclicity pattern and a complete connectivity pattern. We define the tables $A_x^{\budget,\weight,d} \in \bin^{\mathcal{I}}$ recursively as follows:
    \begin{itemize}
    \item Introduce vertex $v$ ($\mu_x = i(v)$): If $v\neq u_0$, we define $A_x^{1,\weightf(v),0}[\langle\idv_0\rangle,[0,i]] = 1$ and $A_x^{0,0,0}[\langle \idv_0,\idv_i\rangle,[0]] = 1$. Otherwise, let $p = [0i]$, and for $\ell\in[2]$, let $p_{\ell} = \action(p, \ell)$. Then we set $D^{1,\weightf(v), \weightf'(x, \ell)}[\langle\idv_0\rangle,p_{\ell}]$ to $1$.
    We set all other values $A_x^{\budget, \weight, d}[p,q]$ to $0$.
    
    \item Relabel node $\mu_x = \relabel{i}{j}(\mu_{x'})$: We define 
    $A_x^{\budget,\weight, d}[p,q] = \sumstack{p' \in \Pat, q' \in \mathscr{C}_C\\ p'_{i\rightarrow j} = p, q'_{i\rightarrow j} = q} A_{x'}^{\budget,\weight,d}[p', q']$.

    \item Join node $\mu_x = \clqadd{i}{j}(\mu_{x'})$: We define
    \[A_x^{\budget,\weight, d}[p,q] = \sumstack{p'\in\Pat\\\patadd_{i,j}p' = p}\sum\limits_{\ell \in [4]}\sumstack{q'\in\mathscr{C}\\\action(q',\ell)=q} \sumstack{q'' \in \mathscr{C}_C\\ \patadd_{i,j}(q'') = q'} A_{x'}^{\budget,\weight, d - \weightf'(x, \ell)}[p',q''].\]

    \item Union node $\mu_x = \mu_{x_1} \clqunion \mu_{x_\ell}$: We define
    \[
    A_x^{\budget,\weight, d}[p,q] = \sumstack{\budget_1+\budget_2 = \budget\\ \weight_1+\weight_2 = \weight\\ d_1 + d_2 = d}\sumstack{p_1,p_2\in\Pat\\ p = p_1\patunion p_2}\sumstack{q_1,q_2\in\mathscr{C}_C\\ q = q_1\punion q_2} A_{x_1}^{\budget_1,\weight_1, d_1}[p_1, q_1] \cdot A_{x_2}^{\budget_2,\weight_2, d_2}[p_2, q_2].
    \]
    \end{itemize}
\end{definition}

\begin{lemma}\label{lem:cfvs-ub-a-count-ac}
    It holds for all $x\in\nodes$, for all values of $\budget$, $\weight$ and $d$, and for all $(p,q) \in I$ that $A_x^{\budget,\weight,d}[(p,q)] = 1$ if and only if there exists an odd number of pairs $(S,\pi)$ with $S\subseteq V_x$ such that $v_0\notin S$, $u_0\in S$ if $u_0\in V_x$, $|S|=\budget$, $\weightf(S) = \weight$, $\pat(G_x-S) = p$ and $\pi$ is an action sequence at $x$ producing the pattern $q$ from $S$.
\end{lemma}

\begin{proof}
    We follow the same structure as in the proof of \cref{lem:alg-counts-sol} using a bottom up induction over $\syntaxtree$, where the correctness for an introduce node is straightforward. The correctness for the other operations follow from the proof of \cite[Lemma~31]{DBLP:conf/icalp/BojikianK24} and \cref{lem:alg-counts-sol}, where for each pair $(S,\pi)$ with $\pat(G_x\setminus S) = p$ and $\pi$ produces $q$ from $S$, it holds that the given formulas correctly modify the patterns $p$ and $q$ according to the given clique-operation. For a join node we also count for the extensions of $\pi$ by assigning the different values in $[4]$ to $x$.
\end{proof}

In \cite{DBLP:conf/icalp/BojikianK24}, the authors define a notion of compatibility over connectivity patterns. Based on this notion, they deduce a notion of representation similar to the equivalence relation defined in \cref{def:partial-compat-equiv}. They call it (parity) representation. More formally, a family of connectivity patterns $R$ represents another family $S$ over the family $\mathscr{C}_C$, if for each pattern $p\in\mathscr{C}_C$ it holds that the $R$ contains an odd number of pattern compatible with with $p$ if and only if $S$ contains an odd number of patterns compatible with $p$. 

Essentially, they prove that the operations $\Delta$, $\patadd_{i,j}$, $p_{i\rightarrow j}$, $\patunion$ and $\action(p, i)$ preserve representation. Moreover, they define the operation $\parrep$, that takes as input a complete pattern $p$ and outputs a family of $CS$-patterns that represents $p$. They prove the following result:

\begin{lemma}[{\cite[Corollary 45]{DBLP:conf/icalp/BojikianK24}}]
    Given a complete connectivity pattern $p\in\mathscr{C}_C$, the operation $\parrep(p)$ produces a set of nice connectivity patterns that represents $p$. Moreover, if there exists a pattern $q_0\in \mathscr{C}_{CS}$, and values $i,j\in[k]$ and $\ell\in[4]$ such that for $q = \patadd_{i,j}(q_0)$ it holds that $p = \action(q, \ell)$, then $\parrep(p)$ can be computed in polynomial time, and produces a set of at most $4$ patterns.
\end{lemma}

We use these notions as a black box in our algorithm.
Using the operation $\parrep$, we aim to reduce the set of indices $\mathcal{I}$ to a smaller set---namely $\CSP \times \mathscr{C}_{CS}$. We define the tables $T_x^{\budget, \weight, d}$ indexed by the latter set, that represent the tables $A_x^{\budget, \weight, d}$. In order to achieve this, we need to combine both notions of representation: connectivity- and acyclicity-representation. We call two pairs $(p,q), (p',q') \in \mathcal{I}$ \emph{compatible}, if it holds that $p\cmptb p'$ and $q\cmptb q'$. We say that a set of pairs $R\subseteq \mathcal{I}$ \emph{represents} another set $S\subseteq \mathcal{I}$, if for each pair $(p',q')\in \mathcal{I}$ it holds that $R$ contains an odd number of pairs compatible with $(p',q')$ if and only if $S$ contains an odd number of pairs compatible with $(p',q')$.

\begin{algorithm}
    Let $\mathscr{I}_{CS} = \CSP \times \mathscr{C}_{CS}$ be a set of indices.
    We define the tables $T_x^{\budget,\weight} \in \bin^{\mathscr{I}_{CS}}$ recursively over $\syntaxtree$ as follows:
    \begin{itemize}
    \item Introduce vertex $v$ ($\mu_x = i(v)$): If $v\neq u_0$, we define $T_x^{1,\weightf(v),0}[\langle\idv_0\rangle,[0,i]] = 1$ and $T_x^{0,0,0}[\langle \idv_0,\idv_i\rangle,[0]] = 1$. Otherwise, let $p = [0i]$, and for $\ell\in[2]$, let $p_{\ell} = \action(p, \ell)$. Then we set $D^{1,\weightf(v), \weightf'(x, \ell)}[\langle\idv_0\rangle,p_{\ell}]$ to $1$.
    We set all other values $T_x^{\budget, \weight, d}$ to $0$.
    
    \item Relabel node $\mu_x = \relabel{i}{j}(\mu_{x'})$: We define 
    $T_x^{\budget,\weight, d}[p,q] = \sumstack{p' \in \Pat, q' \in \mathscr{C}_C\\ p'_{i\rightarrow j} = p, q'_{i\rightarrow j} = q} T_{x'}^{\budget,\weight,d}[p', q']$.

    \item Join node $\mu_x = \clqadd{i}{j}(\mu_{x'})$: We define
    \[T_x^{\budget,\weight, d}[p,q] = \sumstack{p'\in \Pat\\p\in\redpat(p')}\sumstack{p''\in\CSP\\\patadd_{i,j}p'' = p'}\sumstack{q^*\in \mathscr{C}_C\\q\in \parrep(q^*)}\sum\limits_{\ell \in [4]}\sumstack{q'\in\mathscr{C}\\\action(q',\ell)=q^*} \sumstack{q'' \in \mathscr{C}_{CS}\\ \patadd_{i,j}(q'') = q'} T_{x'}^{\budget,\weight, d - \weightf'(x, \ell)}[p'',q''].\]

    \item Union node $\mu_x = \mu_{x_1} \clqunion \mu_{x_\ell}$: We define
    \[
    T_x^{\budget,\weight, d}[p,q] = \sumstack{\budget_1+\budget_2 = \budget\\ \weight_1+\weight_2 = \weight\\ d_1 + d_2 = d}\sumstack{p_1,p_2\in\Pat\\ p = p_1\patunion p_2}\sumstack{q_1,q_2\in\mathscr{C}_C\\ q = q_1\punion q_2} T_{x_1}^{\budget_1,\weight_1, d_1}[p_1, q_1] \cdot T_{x_2}^{\budget_2,\weight_2, d_2}[p_2, q_2].
    \]
\end{itemize}
\end{algorithm}

\begin{lemma}\label{lem:cfvs-ub-t-rep-a}
    It holds for all $x\in\nodes$, and all values $\budget,\weight,d$ and for all $(p,q)\in \mathscr{I}$ that 
    \[
    \sumstack{(p,q)\in \mathscr{I}_{CS}\\p\cmptb p', q\cmptb q'} A_x^{\budget,\weight,d}[p,q] \bquiv\sumstack{(p,q)\in \mathscr{I}_{CS}\\p\cmptb p', q\cmptb q'} T_x^{\budget,\weight,d}[p,q].
    \]
\end{lemma}

\begin{proof}
    We claim that the support of $T_x^{\budget,\weight,d}$ represents the support of $A_x^{\budget,\weight,d}$ which directly implies the lemma. First observe that for two families $R,S\subseteq \mathcal{I}$ with $R$ representing $S$, and an operation over acyclicity patterns $f$ that preserves representation, the extension of $f$ to $I$ by preserving the second index preserves representation. Similarly, for an operation over connectivity patterns $g$ that preserves representation, the extension of $g$ to $I$ by preserving the first index preserves representation.
    
    Now we prove the claim by induction, where for an introduce node, the two tables have the same support by definition. For a relabel node, the support of the tables at $x$ result from applying the operations $p_{i\rightarrow j}$ and $q_{i\rightarrow j}$ on the support of the children nodes,
    and for union nodes, by applying $\patunion$ and $\punion$ on the supports of the children nodes,
    where we replace the outer sum with $\Delta$. The claim then holds since all these operations preserve representation.
    
    For a join node, we first define the tables 
    \[\tilde{T}_x^{\budget,\weight, d} = \sumstack{p'\in\Pat\\\patadd_{i,j}p' = p}\sum\limits_{j \in [4]}\sumstack{q'\in\mathscr{C}\\\action(p',j)=p} \sumstack{p'' \in \mathscr{C}_C\\ \patadd_{i,j}(p'') = p'} A_{x'}^{\budget,\weight, d - \weightf'(x, j)}[q',p''].\]
    Then it holds that the support of $\tilde{T}_x^{\budget, \weight, d}$ represents the support of $A_x^{\budget, \weight, d}$ for all values of $\budget$, $\weight$, $d$, since the operations $\patadd_{i,j}$ and $\action(q,\ell)$ preserve representation. It also holds that the support of $T_x^{\budget, \weight, d}$ represents the support of $\tilde{T}_x^{\budget, \weight, d}$, since the operations $\redpat$ and $\parrep$ preserve representation. It follows by transitivity of representation that $T_x^{\budget, \weight, d}$ represents $A_x^{\budget, \weight, d}$.
\end{proof}

It is observed in \cite{DBLP:conf/icalp/BojikianK24} that the pattern $[0]$ is the only complete connectivity pattern compatible with the pattern $[0]$ itself. Given this, the following corollary follows directly from \cref{lem:cfvs-ub-a-count-ac} and \cref{lem:cfvs-ub-t-rep-a}:

\begin{corollary}\label{cor:cfvs-ub-t0-counts-ac0}
    There exists an odd number of pairs $(S,\pi)$, with $S\subseteq V_r$ such that $u_0 \in S$, $v_0\notin S$, $S$ has size $\budget$ and weight $\weight$, $G_r-S$ is a forest, and $\pi$ generates the pattern $[0]$ from $S$, if and only if it holds that $\sumstack{p \in \CSP}T_x^{\budget, \weight, d}[p, [0]]\bquiv 1$.
\end{corollary}

Now we aim to bound the running time of the algorithm.
As a bottleneck, we need to process union nodes in time $\ostar(18^k)$.
In order to do so, we define the operation $\sqcup$ over $\mathscr{I}_{CS}$ defined by applying $\patunion$ on the first index, and $\punion$ on the second, i.e.\ $(p_1, q_1)\sqcup(p_1,q_2)=(p,q)$, where $p=p_1\patunion p_2$ and $q = q_1\punion q_2$. We show that the convolution $\cfvsconv$ of $\sqcup$ over $\bin$ can be computed in time $\ostar(18^k)$, which will allow us to process union nodes in the desired time.

\begin{lemma}\label{lem:cfvs-conv}
    The convolution of $\sqcup$ over $\bin$ can be computed in time $\ostar(18^k)$.
\end{lemma}

\begin{proof}
    Let us define the ordering $\preceq_C$ over $\mathscr{C}_{CS}$, where $p\preceq_C q$, if for each label $i \in [k]$ it holds that $i\in L_p$ if and only if $i\in L_q$, and $i\in Z_p$ implies that $i\in Z_q$. For a single label $i\in[k]$, we can define the state of $i$ in $p$ as $\stnone$, if $i\notin L_P$, $\stS$ if $i\in L_p\setminus Z_p$, and $\stR$ if $i\in Z_p$. 
    Then it holds that $\preceq_C$ is the $k$th power of the ordering over single labels, defined by the chain $\stS \leq \stR$.
    Let $\zeta_C$ and $\mu_C$ be the Zeta and Mobius transforms of $\preceq_C$ respectively.

    Let $\zeta_1, \zeta_2, \zeta_3$ and $\mu_1, \mu_2, \mu_3$ be the transforms defined in \cref{def:three-lattices}. We define the extensions of $\zeta_C$, $\mu_C$, and of all transforms $\zeta_{\ell}$ and $\mu_{\ell}$ for all $\ell\in[3]$ to $\mathscr{I}_{CS}$, by extending the ordering with the empty ordering on the other side, i.e. we extend $\preceq_C$ to $\mathscr{I}_{CS}$ by defining $(p_1,q_1) \preceq_C (p_2,q_2)$ if and only if $p_1 = p_2$ and $q_1 \preceq_C q_2$, and we extend $\preceq_i$ by defining $(p_1,q_1) \preceq_i (p_2,q_2)$ if and only if $q_1 = q_2$ and $p_1 \preceq_i p_2$. Then each of these transforms can be computed in time $\ostar(18^k)$ by \cref{lem:mobius-zeta-time}.

    Given two tables $T_1, T_2 \in \bin^{\mathscr{I}_{CS}}$, we aim to compute the table $T_1\cfvsconv T_2$. We first compute $\hat{T_1} = \zeta_C(T_1)$ and $\hat{T_2} = \zeta_C(T_2)$. We then compute the convolution $T = \hat{T_1} \stateconv \hat{T_2}$ as described in the proof of \cref{lem:conv-time} using the extensions of $\zeta_i$ and $\mu_i$. We claim that the resulting table equals $\zeta_C(T)$, and hence the table $T$ can be computed by applying $\mu_C$ to the resulting table.
    
    Now we prove the claim. For $\ell \in [2]$, it holds by \cref{lem:conv-zeta-form} that
    \[(\hat{T}_{\ell}^{(3)})_{i,j}(x) = \sum\limits_{x'\preceq_C x}\sum\limits_{y\preceq x'} \Big[\big|C_{y_{C_x}}\big| =i\Big]\Big[\big|D_{y_{D_x}}\big| = j\Big] T_i(y).\]
    It follows (similar to \cref{obs:conv-prod-form}) that
    \[T^{(3)}_{i,j}(x) = \sum\limits_{x'\preceq_C x}\sum\limits_{y,z\preceq x'} 
    \Big[\big|C_{y_{C_x}}\big|+\big|C_{z_{C_x}}\big| = i\Big]
    \Big[\big|D_{y_{D_x}}\big|+\big|D_{z_{D_x}}\big| = j\Big] 
     f(y)\cdot g(z).\]
     The claim follows directly from this, since the table $T$ results from the tables $T^{(3)}_{i,j}$ by applying the transforms $\mu_1, \mu_2, \mu_3$, and by canceling the indices $i$ and $j$, which all depend only on the second index of a pair in $\mathcal{I}_{CS}$. Hence, the outer sum $\sum\limits_{x'\preceq_C x}$ stays as it is after each step.
\end{proof}

\begin{lemma}\label{lem:cfvs-alg-time}
    The tables $T_x^{\budget,\weight, d}$ for all nodes $x\in\nodes$ and all values $\budget,\weight, d$ can be computed in time $\ostar(18^k)$.
\end{lemma}

\begin{proof}
    Since it holds that $\nodes$ is bounded polynomially in $n$, and we only iterate over a polynomial number of value $\budget, \weight, d$, it suffices to bound the running time of computing $T_x^{\budget, \weight, d}$ for a each fixed node $x$ and fixed values $\budget, \weight, d$. For an introduce node, this is trivial, since we set all values to constants. For a relabel node, we iterate ove all indices $(p,q)\in\mathscr{I}_{CS}$ and we add $T_{x'}^{\budget, \weight, d}[p,q]$ to $T_{x}^{\budget,\weight,d}[p_{i\rightarrow j},q_{i\rightarrow j}]$.

    For a join node, we iterate ove all indices $(p_0,q_0)\in\mathscr{I}_{CS}$, an we compute $p := \patadd_{i,j}(p_0)$, and $q=\patadd_{i,j}(q_0)$. Let $P = \redpat(p)$. For each $\ell\in[4]$, let $q^{\ell} = \action(q, \ell)$, and $Q^{\ell} = \parrep(q^{\ell})$. Then we add $T_x^{\budget,\weight, d- \weightf'(x, \ell)}[p,q]$ to $T_x^{\budget, \weight, d}[p',q']$ for each pair $(p',q')\in P\times Q^{\ell}$.

    Finally, for a union node, we iterate over all values $\budget_1,\budget_2$, $\weight_1,\weight_2$, $d_1,d_2$, such that $\budget_1+\budget_2=\budget$, $\weight_1+\weight_2=\weight$, and $d_1+d_2=d$. For each we compute the convolution $T_{x_1}^{\budget_1,\weight_1,d_1}\cfvsconv T_{x_2}^{\budget_2,\weight_2,d_2}$, and add the resulting table to $T_x^{\budget,\weight,d}$. This can be done by \cref{lem:cfvs-conv} in tme $\ostar(18^k)$ and corresponds exactly to the definition of $T_x^{\budget, \weight, d}$.
\end{proof}

So far we have shown that the number of pairs $(S, \pi)$ generating the pattern $[0]$ can be counted efficiently. We use the isolation lemma to isolate a single such pair---and hence, decide the existence of a connected feedback vertex set of a fixed size---with high probability.

\begin{lemma}\label{lem:cfvs-ub-iso}
    Fix $\W = (2+\sqrt{2})n$ and $D = (2+\sqrt{2})4|\nodes|$, and choose $\weightf$ and $\weightf'$ by choosing the image of each element independently und uniformly at random. Let $(S, \pi)$ be a pair, such that $S$ is a minimum weight connected feedback vertex set of size $\budget$ and $\pi$ a minimum weight actions sequence generating $[0]$ from $S$. Then $(S, \pi)$ is unique with probability at least $1/2$.
\end{lemma}

\begin{proof}
    Let $\mathcal{F}$ be the family of all connected feedback vertex sets of $G$. Then it follows from \cref{lem:iso} that $\weightf$ isolates $\mathcal{F}$ with probability at least $1-1/(2+\sqrt 2)$.

    Let $U = \nodes \times [4]$. Then there is a canonical injective mapping that assigns to each action sequence a subset of $U$, namely the set $\{(x,\pi(x))\colon x\in \nodes\}$. For a set $S\subseteq V$, let $\mathcal{F}_S$ be the family of subsets of $U$ corresponding to action sequences that generate the pattern $[0]$ from $S$. Then it follows from \cref{lem:iso} that $\weightf'$ isolates $\mathcal{F}_U$ with probability at least $1-1/(2 + \sqrt 2)$ and independent of the choice of $S$.
    It follows that $(S, \pi)$ is unique with probability at least $\left(1-\frac{1}{2 + \sqrt 2}\right)^2 = \frac{1}{2}$.
\end{proof}

\begin{proof}[Proof of \cref{theo:cfvs-ub}]
    We fix $\W = (2+\sqrt 2)n$ and $D = (2+\sqrt 2)4|\nodes|$ and choose $\weightf$ and $\weightf'$ randomly as described in \cref{lem:cfvs-ub-iso}.
    The algorithm iterates over all choices $u_0\in V$, and for each,
    it builds the tables $T_x^{\budget, \weight, d}$. The algorithm accepts if there exist values $\weight$, $d$, such that $\sumstack{p\in\CSP}T_x^{\target, \weight, d}[(p,[0])] \bquiv1$, an rejects otherwise.

    It holds by \cref{lem:cfvs-alg-time} that the tables $T$ can be computed in time $\ostar(18^k)$.
    Since it holds by \cref{lem:cfvs-sol-if-action-seq}, that a partial solution $S$ containing the vertex $u_0$ induces a connected subgraph of $G_r$ if and only if there exists an action sequence generating $[0]$ from $S$, it holds by \cref{cor:cfvs-ub-t0-counts-ac0} that, if no solution exists, then for all pairs $(S, \pi)$, $\pi$ does not generate $[0]$ from $S$, and hence, $\sumstack{p\in\CSP}T_x^{\target, \weight, d}[(p,[0])] \bquiv 0$ for all values $\weight$ and $d$.
    
    On the other hand, if a solution $S$ containing some vertex $u_0$ exists, then there exists an action sequence $\pi$ generating $[0]$ from $S$.
    It holds then by \cref{lem:cfvs-ub-iso} that with probability at least $1/2$, there exist values $\weight$ and $d$ such that there exists a unique such pair $(S, \pi)$ where $S$ has weight $\weight$ and $\pi$ has weight $d$. It holds by \cref{cor:cfvs-ub-t0-counts-ac0} that this is the case, if and only if the the sum $\sumstack{p\in\CSP}T_x^{\target, \weight, d}[(p,[0])]$ is congruent to $1$, and hence, with probability at least one half, the algorithm outputs $1$ for this choice of $u_0$. Since we iterate over all choice of $u_0\in V$, and since we assume that the solution is not empty, the algorithm finds the solution with probability at least $1/2$.
\end{proof}

\subsection{Lower bound}
Similar to \cref{sec:lb}, we also prove the lower bound through a reduction from the $B$-CSP-$q$ problem for each value of $q$. However, here we fix $B=18$ instead of $6$.
We follow the same graph construction $G$ from \cref{sec:lb} with some modifications. First, we define the number of path gadgets on each path sequence $c$ to be equal $(17n + 1)m$. We also add a vertex $r_C$, and call it the connectivity root, and add two private neighbors adjacent to each other and to $r_C$ only.
We add an edge between $r_C$ and the entry vertices of the first path gadget of each path sequence, and between $r_C$ and the exit vertices of the last path gadget of each path sequence.
We also add an edge between each vertex of a constraint gadget, and the vertex $r_C$. We also modify a path gadget as follows:

\subparagraph{Path gadget.} In a path gadget $X$, we replace the $6$ clique vertices with $18$ vertices $x_1,\dots x_{18}$, all adjacent to $r_C$, and we replace the two entry and the two exit cycles with three \emph{entry structures} $C_1, C_2, C_3$, and three \emph{exit structures} $C_4, C_5, C_6$.
Each of these structures $C_i$ consists of four vertices $v_i, v'_i, u_i, u'_i$, where $u_i$ and $u'_i$ are adjacent to both $v_i$ and $v'_i$ building a bipartite clique $K_{2,2}$. We also make $u_i$ adjacent to both $r$ and $r_C$, and $v'_i$ adjacent to $r_C$. We add a deletion edge between $v_i$ and $v'_i$, and one between $u_i$ and $u'_i$. See \cref{fig:cfvs-cycle} for graphical depiction.

\begin{figure}[ht]
    \centering
    \includegraphics[width=.2\textwidth]{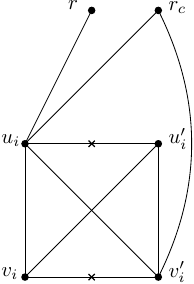}
    \caption{\label{fig:cfvs-cycle}An entry/exit structure of a path gadget. The vertex $v_i$ is the \emph{boundary} vertex.}
\end{figure}

We replace the $2$-to-$2$ transition mapping $t$ with a $3$-to-$3$ transition mapping defined as follows: We define the family of simple states $Z = \{\stdisc, \stconn, \stS, \stR\}$, where $\stdisc$ stands for $v_i$ being excluded from the solution, and not connected to $r$ through $X$, $\stconn$ stands for $v_i$ being excluded from the solution and connected to $r$ through $X$, $\stS$ stands for $v_i$ being in the solution but not connected to $r_C$ through $X$, and $\stR$ stands for $v_i$ being in the solution and connected to $r_C$ through $X$. We define $t_C\colon [18]\rightarrow Z^3\times Z^3$ with $t_C(i) = (t_i, t_{19-i})$ where the values $t_i$ are give in the following list:

\begin{tabularx}{.8\textwidth}{XXX}
    $t_{1} = (\stR,\stR,\stR)$,&
    $t_{2} = (\stR,\stS,\stS)$,&
    $t_{3} = (\stS,\stS,\stS)$,\\
    $t_{4} = (\stD,\stR,\stR)$,&
    $t_{5} = (\stD,\stR,\stS)$,&
    $t_{6} = (\stD,\stS,\stS)$,\\
    $t_{7} = (\stC,\stR,\stR)$,&
    $t_{8} = (\stC,\stR,\stS)$,&
    $t_{9} = (\stC,\stS,\stS)$,\\
    $t_{10} = (\stD,\stD,\stR)$,&
    $t_{11} = (\stD,\stD,\stS)$,&
    $t_{12} = (\stD,\stD,\stD)$,\\
    $t_{13} = (\stC,\stD,\stR)$,&
    $t_{14} = (\stC,\stD,\stD)$,&
    $t_{15} = (\stC,\stD,\stD)$,\\
    $t_{16} = (\stC,\stC,\stR)$,&
    $t_{17} = (\stC,\stC,\stS)$,&
    $t_{18} = (\stC,\stC,\stC)$.\\
\end{tabularx}

Therefore, for each $i\in[18]$, the image $t_C(i)$ can be seen as a $6$-tuple of states. We add the deletion edges between $x_i$ and $C_j$ for $i\in[18]$ and $j\in[6]$ as follows: Let $X = t_C(i)_j$, we add a deletion edge between $x_i$ and $v_j$ if $X \in \{\stR,\stS\}$, or between $x_i$ and $v'_i$ otherwise. We also add a deletion edge between $x_i$ and $u_j$ if $X \in\{\stR, \stD\}$, or between $x_i$ and $u'_j$ otherwise. Intuitively, the first deletion edge decides whether the boundary vertex $v_i$ is in the solution, where the second decides its connectivity under the state defined by $t_C(i)$.

We redefine the target size of a connected feedback vertex set $\budget$ by setting $\lbound(\inctri{K}) = 15$ for the clique-vertices $K$ together with their deletion edges, and by setting $\lbound(\inctri{C_j}) = 2$. Clearly, a solution must pack at least two vertices in each structure $C_j$, given the two deletion edges between $v_i$ and $v'_i$ and between $u_i$ and $u'_i$. We add $r_C$ together with the two vertices added to create a triangle at $r_C$ to $\lbfamily$, and define $\lbound(\inctri{r_C}) = 1$.
Hence, it holds analogously, that a solution of size $\budget$ must again pack exactly $\lbound(X)$ vertices in each component $X\in\lbfamily$. It also holds that each solution must contain an endpoint of each deletion edge, and the vertex $r_C$.

\begin{lemma}\label{lem:cfvs-lb-cw-bound}
    There exists a constant $k_0$, such that for $k=n+k_0$, a linear $k$-expression $\mu$ of $G$ can be constructed in polynomial time.
\end{lemma}

We build $\mu$ following the same steps of \cref{lem:lb-cw-bound}, where we add one label for $r_C$, and start by creating $r_C$ together with the triangle added at $r_C$. The bound follows from the fact, that we only increase the size of a path gadget by a constant.

\begin{lemma} \label{lem:cfvs-lb-sol-if-sat}
    If $I$ is satisfiable, then $G$ admits a connected feedback vertex set of size $\budget$.
\end{lemma}

\begin{proof}
    Let $\alpha$ be a satisfying assignment. We describe a solution $\sol$ of size $\budget$ following the proof of \cref{lem:lb-sol-if-sat}, where we add the vertex $r_C$, and for a path gadget $\pgadget$ on the $i$th path sequence, we add all clique vertices of $\pgadget$ to the solution, except the vertex $x_{\ell}$ for $\ell = \alpha(i)$. From each structure $C_i$, we add the two vertices of $C_i$ to $\sol$ sharing deletion edges with $x_{\ell}$.
    The bound on the size of $\sol$ follows, since we add $\lbound(Z)$ vertices to $\sol$ from each structure $Z\in \lbfamily$, and no other vertices.

    Now we show that $\sol$ is a valid solution. First, prove that $\sol$ is a feedback vertex set. This follows from the same arguments in \cref{lem:lb-sol-if-sat}, where for two consecutive path gadgets $\pgadget$ and $\pgadget'$, it holds by the choice of the mapping $t_C$, that if at least two exit vertices of $\pgadget$ are excluded from the solution, then at most one entry vertex of $\pgadget'$ is excluded, and if there exists an exit structure, with both $v_i$ and $u_i$ excluded from the solution, then it holds for each entry structure $C_j$ of $\pgadget'$, that if $v_j$ is excluded from the solution, then $u_j$ is included in the solution. Finally, it also holds that if two exit structures of $\pgadget$ have both the vertices $v_i$ and $u_i$ excluded from the solution (corresponding to two states $\stC$ in the image of $t_C$), then all entry vertices $v_j$ of $\pgadget'$ are included in the solution.

    Now we prove the connectivity of the solution. Since we assume that $r_C \in \sol$, it suffices to show that every vertex of $\sol$ is connected to $r_C$ in $G[\sol]$. This holds trivially for the vertex of each constraint gadget that belongs to $\sol$, since all these vertices are adjacent to $r_C$. For a path gadget $\pgadget$ belonging to the $\ell$th path sequence, this also holds for all clique-vertices of $\pgadget$ and all vertices $u_i$ and $v'_i$ that belong to $\sol$, since all these vertices are also adjacent to $r_C$. So we only need to argue about vertices $u'_i$ and $v_i$. Since $u'_i$ is adjacent to both $v_i$ and $v'_i$, and at least one of them belongs to $\sol$, it suffices to argue the connectivity of $v_i$ to $r_C$. If $u_i\in \sol$, then this is trivial due to the edge between $v_i$ and $u_i$. Therefore, we can assume that $v_i$ belong to $\sol$ but not $u_i$. In this case, it must hold that the clique vertex not in $\sol$ has deletion edges to both $v_i$ and $u'_i$, and by construction of $\pgadget$ that $\alpha(\ell) = x$ for some state $x\in[18]$ with $t_C(x)_i = \stS$.
    
    If $v_i$ is an entry vertices of the first path gadget, or an exit vertex of the last path gadget in its path sequence, then it is adjacent to $r_C$ by definition. Otherwise, we can assume that $v_i$ is an entry vertex of $\pgadget$, i.e.\ $i\in[3]$, since the other case is symmetric. Let $\pgadget'$ be the path gadget preceding $\pgadget$.
    Since it holds for each state $t_x$ containing the state $\stS$, and for $y= 18 - x + 1$ that $t_y$ contains $\stR$ or $t_y$ contains $\stS$ and $t_x$ contains $\stR$, it follows by the choice of $\sol$ that there exists $j\in\{4,5,6\}$ where $v_j(\pgadget')$ belongs to $\sol$, and either $u_j(\pgadget')$ belongs to $\sol$, or there exist $i'$ where $v_{i'}$ and $u_{i'}$ of $\pgadget$ are both in $\sol$. By the construction of $G$, it holds that $v_i$ and $v_{i'}$ are both adjacent to $u_j(\pgadget')$, and hence, $v_i$ is connected to $r_C$ in $G[\sol]$.
\end{proof}

Now we aim to prove the other direction of the lemma. From now on, let $\sol$ be a fixed solution of size $\budget$.
We define the state of a path gadget $\pgadget$ defined by $\sol$ (denoted $\statef(\pgadget)$) as the unique value $i\in[18]$ such that $x_i \notin \sol$. We prove a corresponding variant of \cref{lem:lb-state-imply-v}, and use it to show that the states defined by $\sol$ along consecutive path gadgets can only increase (analogous to \cref{lem:lb-states-transition-direction}).

\begin{lemma}\label{lem:cfvs-state-imply-v}
    Let $\operatorname{REACH}_G(x,y)$ be the predicate, that decides whether the graph $G$ contains a path from $x$ to $y$.
    For a path gadget $\pgadget$, let $G_{\pgadget}^{\sol} = G\big[(\pgadget\cup\{r_C\}\big)\cap S]$ and $G_{\pgadget}^{\overline{\sol}} = \hat{G}\big[(\pgadget\cup\{r\})\setminus S\big]$. We define the following predicates for each $i\in[6]$:
    \begin{align*}
    \phi^{\stS}_i:=(v_i\in\sol \land \lnot \reachsol(v_i, r_C)),\qquad&
    \phi^{\stR}_i:=(v_i\in\sol\land \reachsol(v_i, r_C)),\\
    \phi^{\stD}_i:=(v_i\notin\sol\land \lnot \reachrest(v_i, r_C)),\qquad&
    \phi^{\stC}_i:=(v_i\notin\sol\land \reachrest(v_i, r_C)).
    \end{align*}
    Then it holds for $i\in[6]$, $s:= \statef(\pgadget)$ and $X\in Z$ that $\phi^{X}_i$ holds if and only if $t_C(s)_i = X$, i.e.\ if $X$ is the $i$th states in the transition $s$.
\end{lemma}

\begin{proof}
    It holds clearly, that the four predicates are pairwise disjoint. Hence, it suffices to show that $t_C(s)_i = X$ implies $\phi^X_i$ for all path gadgets $X$ and all values $i\in[6]$, but this follows directly from the deletion edges between the clique vertices and the two vertices $v_i$ and $u_i$, since all adjacencies of $v_i \in\sol$ other than $u_i$ and $u'_i$ are deletion vertices, and since $u'_i$ is only adjacent to $v_i$ and $v'_i$ by non-deletion edges.
\end{proof}

\begin{lemma}\label{lem:cfvs-lb-states-direction}
    Let $\pgadget$ and $\pgadget'$ be two consecutive path gadgets on the same path sequence. Then it holds that $\statef(\pgadget')\leq \statef(\pgadget)$.
\end{lemma}

\begin{proof}
    Let $x = \statef(\pgadget)$ and $y = \statef(\pgadget')$, and for the sake of contradiction, assume that $y > x$.
    Let us partition the states $t_i$ into six groups ot size $3$ each. If $x$ and $y$ belong to two different groups, then a contradiction follows by the same arguments from the proof of \cref{lem:lb-states-transition-direction}, where the state $\stdconn$ corresponds to at least two occurrences of $\stC$, $\stddisc$ to at least two occurrences of $\stD$, and $\stplus$ to an occurrence of $\stC$ and at least one occurrence of $\stD$.

    Now assume that both $x$ and $y$ belong to the same group. Let $t_C(x) = (X_1,\dots X_6)$ and $t_C(y) = (Y_1,\dots Y_6)$. We define the sets $\mathcal{X} = \{X_4, X_5, X_6\}$ and $\mathcal{Y} = \{Y_1, Y_2, Y_3\}$. If $\stS$ belongs to one of the sets $\mathcal{X}, \mathcal{Y}$, then the other set must contain $\stS$ or $\stR$ and it must also hold in this case that $\stR \in \mathcal{X}\cup \mathcal{Y}$, since otherwise, it holds by \cref{lem:cfvs-state-imply-v} that there exist a vertex $v_i\in \sol$ where $v_i$ is not connected to $r_C$ in $G[S]$ which contradicts the connectivity of $S$.

    Hence, it holds that either $y = 3k$ for some value $k\in[6]$, or that $y = 3k-1$ and $x= 3k-2$. Let us now distinguish two cases, if $x$ and $y$ belong to one of the first three groups ($k\in[3]$), then it holds in case of $y=3k$ that $\stS\in\mathcal{Y}$ and $\stR\notin\mathcal{Y}$, and hence, it must hold that $x=3k$ as well, since in this case, $\stR$ appears in $t_{19-x}$ only when $x = 3k$, contradicting the assumption that $y>x$. If $y=3k-1$ and $x=3k-2$, then $\stS \in \mathcal{Y}$, but $\mathcal{X}$ does not contain $\stS$ or $\stR$ contradicting our argument in the previous paragraph.
    Now assume that that $x$ and $y$ belong to one of the latter three groups (i.e., $k\in\{4,5,6\})$, if $y = 3k$, then it holds that both $\stS$ and $\stR$ are not contained in $\mathcal{Y}$, and hence it must hold that $x=3k$ since for all other vales of $x$ it holds that $\stS \in t_{19-x}$. Finally, for $y=3k-1$ and $x = 3k-2$, it holds that $\stS\in\mathcal{X}$, but $\stR \notin \mathcal{X}\cup \mathcal{Y}$ which also contradicts the argument in the previous paragraph.
\end{proof}

\begin{lemma}\label{lem:cfvs-lb-sat-if-sol}
    If $G$ contains a connected feedback vertex set of size $\budget$, then $I$ is satisfiable.
\end{lemma}

The proof follows \cref{lem:lb-sat-if-sol}, where it holds by \cref{lem:cfvs-lb-states-direction} that on each path sequence at most $17$ state changes occur. Since the graph $G$ contains $17n+1$ sections, there must exist a section $j\in [17n+1]$, where all path gadgets of each segment on this section have the same state. Then a satisfying assignment $\pi$ can be built by assigning to each variable $v_i$ the state assigned by $\sol$ to all path gadgets on the $j$th segment of the $i$th path sequence. The resulting assignment is then satisfying, since for each constraint of $I$, there exists a constraint gadget corresponding to this constraint, attached to some column of this section, and that the deletion edges between this constraint gadget and the path gadgets in its column ensure that the assigned states satisfy this constraint.

\begin{proof}[Proof of \cref{theo:cfvs-lb}]
    Assume that there exists an algorithm that solves the \Cfvsp problem in time $\ostar((18-\varepsilon)^k)$ for some value $\varepsilon > 0$, given a linear $k$-expression of the input graph $G$ with input. Given an instance $I$ of the $q$-CSP-$B$ problem over $n$ variables for $B=18$ and some value $q\in\mathbb{N}$, we start by building the graph $G$ and the linear $k$-expression of $G$ as described in \cref{lem:cfvs-lb-cw-bound}. We run the algorithm on the instance $(G, \budget)$ as defined above.
    
    It holds by \cref{lem:cfvs-lb-sol-if-sat} and \cref{lem:cfvs-lb-sat-if-sol} that $G$ admits a connected feedback vertex set of size $\budget$ if and only if $I$ is satisfiable. Since it holds by \cref{lem:cfvs-lb-cw-bound} that $\mu$ is a linear $k$-expression of $G$ that can be constructed in polynomial time for $k = n +k_0$ for some constant $k_0$, it holds that the algorithm runs in time $\ostar((18-\varepsilon)^n)$ which contradicts SETH by \cref{theo:lampis-csp-lb}.
\end{proof}

\section{Conclusion}\label{sec:conclusion}
In this work, we presented a one-sided error Monte-Carlo algorithm for the \Fvsp problem parameterized by clique-width with running time $\ostar(6^{\cw})$, when parameterized by the clique-width of the graph, based on a new notion of acyclicity representation in labeled graphs, and we provided a matching SETH-based lower bound. Based on this representation, we also presented an algorithm, that counts (modulo $2$) the number of solutions in time $\ostar(3^{\tw})$ in graphs of treewidth $\tw$.
This matches the SETH-tight running time of the decision version of this problem when parameterized by treewidth.

As an additional consequence of this representation, we presented an $\ostar(18^{\cw})$ algorithm for the \textsc{Connected Feedback Vertex Set} problem parameterized by clique-width. We proved the tightness of our running time by providing a matching SETH-based lower bound. In addition to acyclicity representation, we achieved the claimed running times for both problems by developing an involved fast convolution technique.

As mentioned in the introduction, this leaves the tight complexity of only one problem (relative to clique-width) unresolved among those whose tight complexity relative to treewidth was determined in the original work of Cygan et al.~\cite{DBLP:journals/talg/CyganNPPRW22}, namely the \textsc{Exact $k$-Leaf Spanning Tree} problem. While the problem is $W[1]$-hard, we believe that our technique could be useful in achieving a single-exponential $XP$ running time for this problem, as the same challenge of counting edges and connected components arises in this problem, when trying to adapt the \cnc technique to clique-width.

A further intriguing question is whether our convolution technique can be extended to arbitrary depth and any number of filters, yielding a general framework that encompasses all known acyclic convolution techniques for structural parameters.

\bibliography{ref}

@article{DBLP:journals/tcs/AdlerK15,
  author       = {Isolde Adler and
                  Mamadou Moustapha Kant{\'{e}}},
  title        = {Linear rank-width and linear clique-width of trees},
  journal      = {Theor. Comput. Sci.},
  volume       = {589},
  pages        = {87--98},
  year         = {2015},
  url          = {https://doi.org/10.1016/j.tcs.2015.04.021},
  doi          = {10.1016/j.tcs.2015.04.021},
  bibsource    = {dblp computer science bibliography, https://dblp.org}
}

@article{DBLP:journals/tcs/BergougnouxK19,
  author       = {Benjamin Bergougnoux and
                  Mamadou Moustapha Kant{\'{e}}},
  title        = {Fast exact algorithms for some connectivity problems parameterized
                  by clique-width},
  journal      = {Theor. Comput. Sci.},
  volume       = {782},
  pages        = {30--53},
  year         = {2019},
  url          = {https://doi.org/10.1016/j.tcs.2019.02.030},
  doi          = {10.1016/j.tcs.2019.02.030},
  bibsource    = {dblp computer science bibliography, https://dblp.org}
}

@inproceedings{DBLP:conf/stacs/BergougnouxKN23,
  author       = {Benjamin Bergougnoux and
                  Tuukka Korhonen and
                  Jesper Nederlof},
  editor       = {Petra Berenbrink and
                  Patricia Bouyer and
                  Anuj Dawar and
                  Mamadou Moustapha Kant{\'{e}}},
  title        = {Tight Lower Bounds for Problems Parameterized by Rank-Width},
  booktitle    = {40th International Symposium on Theoretical Aspects of Computer Science,
                  {STACS} 2023, March 7-9, 2023, Hamburg, Germany},
  series       = {LIPIcs},
  volume       = {254},
  pages        = {11:1--11:17},
  publisher    = {Schloss Dagstuhl - Leibniz-Zentrum f{\"{u}}r Informatik},
  year         = {2023},
  url          = {https://doi.org/10.4230/LIPIcs.STACS.2023.11},
  doi          = {10.4230/LIPIcs.STACS.2023.11},
  bibsource    = {dblp computer science bibliography, https://dblp.org}
}

@inproceedings{DBLP:conf/stoc/BjorklundHKK07,
  author       = {Andreas Bj{\"{o}}rklund and
                  Thore Husfeldt and
                  Petteri Kaski and
                  Mikko Koivisto},
  editor       = {David S. Johnson and
                  Uriel Feige},
  title        = {Fourier meets m{\"{o}}bius: fast subset convolution},
  booktitle    = {Proceedings of the 39th Annual {ACM} Symposium on Theory of Computing,
                  San Diego, California, USA, June 11-13, 2007},
  pages        = {67--74},
  publisher    = {{ACM}},
  year         = {2007},
  url          = {https://doi.org/10.1145/1250790.1250801},
  doi          = {10.1145/1250790.1250801},
  bibsource    = {dblp computer science bibliography, https://dblp.org}
}

@article{DBLP:journals/talg/BjorklundHKKNP16,
  author       = {Andreas Bj{\"{o}}rklund and
                  Thore Husfeldt and
                  Petteri Kaski and
                  Mikko Koivisto and
                  Jesper Nederlof and
                  Pekka Parviainen},
  title        = {Fast Zeta Transforms for Lattices with Few Irreducibles},
  journal      = {{ACM} Trans. Algorithms},
  volume       = {12},
  number       = {1},
  pages        = {4:1--4:19},
  year         = {2016},
  url          = {https://doi.org/10.1145/2629429},
  doi          = {10.1145/2629429},
  bibsource    = {dblp computer science bibliography, https://dblp.org}
}

@article{DBLP:journals/iandc/BodlaenderCKN15,
  author       = {Hans L. Bodlaender and
                  Marek Cygan and
                  Stefan Kratsch and
                  Jesper Nederlof},
  title        = {Deterministic single exponential time algorithms for connectivity
                  problems parameterized by treewidth},
  journal      = {Inf. Comput.},
  volume       = {243},
  pages        = {86--111},
  year         = {2015},
  url          = {https://doi.org/10.1016/j.ic.2014.12.008},
  doi          = {10.1016/j.ic.2014.12.008},
  bibsource    = {dblp computer science bibliography, https://dblp.org}
}

@inproceedings{DBLP:conf/stacs/BojikianCHK23,
  author       = {Narek Bojikian and
                  Vera Chekan and
                  Falko Hegerfeld and
                  Stefan Kratsch},
  editor       = {Petra Berenbrink and
                  Patricia Bouyer and
                  Anuj Dawar and
                  Mamadou Moustapha Kant{\'{e}}},
  title        = {Tight Bounds for Connectivity Problems Parameterized by Cutwidth},
  booktitle    = {40th International Symposium on Theoretical Aspects of Computer Science,
                  {STACS} 2023, March 7-9, 2023, Hamburg, Germany},
  series       = {LIPIcs},
  volume       = {254},
  pages        = {14:1--14:16},
  publisher    = {Schloss Dagstuhl - Leibniz-Zentrum f{\"{u}}r Informatik},
  year         = {2023},
  url          = {https://doi.org/10.4230/LIPIcs.STACS.2023.14},
  doi          = {10.4230/LIPIcs.STACS.2023.14},
  bibsource    = {dblp computer science bibliography, https://dblp.org}
}

@article{DBLP:journals/corr/bojikianfghs25,
  author       = {Narek Bojikian and
                  Alexander Firbas and
                  Robert Ganian and
                  Hung P. Hoang and
                  Krisztina Szil{\'{a}}gyi},
  title        = {Fine-Grained Complexity of Computing Degree-Constrained Spanning Trees},
  journal      = {CoRR},
  volume       = {abs/2503.15226},
  year         = {2025},
  url          = {https://doi.org/10.48550/arXiv.2503.15226},
  doi          = {10.48550/ARXIV.2503.15226},
  eprinttype    = {arXiv},
  eprint       = {2503.15226},
  bibsource    = {dblp computer science bibliography, https://dblp.org}
}

@inproceedings{DBLP:conf/icalp/BojikianK24,
  author       = {Narek Bojikian and
                  Stefan Kratsch},
  editor       = {Karl Bringmann and
                  Martin Grohe and
                  Gabriele Puppis and
                  Ola Svensson},
  title        = {A Tight Monte-Carlo Algorithm for Steiner Tree Parameterized by Clique-Width},
  booktitle    = {51st International Colloquium on Automata, Languages, and Programming,
                  {ICALP} 2024, July 8-12, 2024, Tallinn, Estonia},
  series       = {LIPIcs},
  volume       = {297},
  pages        = {29:1--29:18},
  publisher    = {Schloss Dagstuhl - Leibniz-Zentrum f{\"{u}}r Informatik},
  year         = {2024},
  url          = {https://doi.org/10.4230/LIPIcs.ICALP.2024.29},
  doi          = {10.4230/LIPICS.ICALP.2024.29},
  bibsource    = {dblp computer science bibliography, https://dblp.org}
}

@article{DBLP:journals/corr/bojikiank24,
  author       = {Narek Bojikian and
                  Stefan Kratsch},
  title        = {Tight Algorithm for Connected Odd Cycle Transversal Parameterized
                  by Clique-width},
  journal      = {CoRR},
  volume       = {abs/2402.08046},
  year         = {2024},
  url          = {https://doi.org/10.48550/arXiv.2402.08046},
  doi          = {10.48550/ARXIV.2402.08046},
  eprinttype    = {arXiv},
  eprint       = {2402.08046},
  bibsource    = {dblp computer science bibliography, https://dblp.org}
}

@article{DBLP:journals/corr/BrandCLP25,
  author       = {Cornelius Brand and
                  Radu Curticapean and
                  Baitian Li and
                  Kevin Pratt},
  title        = {Faster Convolutions: Yates and Strassen Revisited},
  journal      = {CoRR},
  volume       = {abs/2505.22410},
  year         = {2025},
  url          = {https://doi.org/10.48550/arXiv.2505.22410},
  doi          = {10.48550/ARXIV.2505.22410},
  eprinttype    = {arXiv},
  eprint       = {2505.22410},
  bibsource    = {dblp computer science bibliography, https://dblp.org}
}

@article{CourcelleO00,
  author       = {Bruno Courcelle and
                  Stephan Olariu},
  title        = {Upper bounds to the clique width of graphs},
  journal      = {Discret. Appl. Math.},
  volume       = {101},
  number       = {1-3},
  pages        = {77--114},
  year         = {2000},
  url          = {https://doi.org/10.1016/S0166-218X(99)00184-5},
  doi          = {10.1016/S0166-218X(99)00184-5},
  bibsource    = {dblp computer science bibliography, https://dblp.org}
}

@inproceedings{DBLP:conf/soda/CurticapeanLN18,
  author       = {Radu Curticapean and
                  Nathan Lindzey and
                  Jesper Nederlof},
  editor       = {Artur Czumaj},
  title        = {A Tight Lower Bound for Counting Hamiltonian Cycles via Matrix Rank},
  booktitle    = {Proceedings of the Twenty-Ninth Annual {ACM-SIAM} Symposium on Discrete
                  Algorithms, {SODA} 2018, New Orleans, LA, USA, January 7-10, 2018},
  pages        = {1080--1099},
  publisher    = {{SIAM}},
  year         = {2018},
  url          = {https://doi.org/10.1137/1.9781611975031.70},
  doi          = {10.1137/1.9781611975031.70},
  bibsource    = {dblp computer science bibliography, https://dblp.org}
}

@book{DBLP:books/sp/CyganFKLMPPS15,
  author       = {Marek Cygan and
                  Fedor V. Fomin and
                  Lukasz Kowalik and
                  Daniel Lokshtanov and
                  D{\'{a}}niel Marx and
                  Marcin Pilipczuk and
                  Michal Pilipczuk and
                  Saket Saurabh},
  title        = {Parameterized Algorithms},
  publisher    = {Springer},
  year         = {2015},
  url          = {https://doi.org/10.1007/978-3-319-21275-3},
  doi          = {10.1007/978-3-319-21275-3},
  isbn         = {978-3-319-21274-6},
  bibsource    = {dblp computer science bibliography, https://dblp.org}
}

@article{DBLP:journals/jacm/CyganKN18,
  author       = {Marek Cygan and
                  Stefan Kratsch and
                  Jesper Nederlof},
  title        = {Fast Hamiltonicity Checking Via Bases of Perfect Matchings},
  journal      = {J. {ACM}},
  volume       = {65},
  number       = {3},
  pages        = {12:1--12:46},
  year         = {2018},
  url          = {https://doi.org/10.1145/3148227},
  doi          = {10.1145/3148227},
  bibsource    = {dblp computer science bibliography, https://dblp.org}
}

@article{DBLP:journals/talg/CyganNPPRW22,
  author       = {Marek Cygan and
                  Jesper Nederlof and
                  Marcin Pilipczuk and
                  Michal Pilipczuk and
                  Johan M. M. van Rooij and
                  Jakub Onufry Wojtaszczyk},
  title        = {Solving Connectivity Problems Parameterized by Treewidth in Single
                  Exponential Time},
  journal      = {{ACM} Trans. Algorithms},
  volume       = {18},
  number       = {2},
  pages        = {17:1--17:31},
  year         = {2022},
  url          = {https://doi.org/10.1145/3506707},
  doi          = {10.1145/3506707},
  bibsource    = {dblp computer science bibliography, https://dblp.org}
}

@article{DBLP:journals/siamcomp/FominGLS10,
  author       = {Fedor V. Fomin and
                  Petr A. Golovach and
                  Daniel Lokshtanov and
                  Saket Saurabh},
  title        = {Intractability of Clique-Width Parameterizations},
  journal      = {{SIAM} J. Comput.},
  volume       = {39},
  number       = {5},
  pages        = {1941--1956},
  year         = {2010},
  url          = {https://doi.org/10.1137/080742270},
  doi          = {10.1137/080742270},
  bibsource    = {dblp computer science bibliography, https://dblp.org}
}

@inproceedings{DBLP:conf/latin/Furer14,
  author       = {Martin F{\"{u}}rer},
  editor       = {Alberto Pardo and
                  Alfredo Viola},
  title        = {A Natural Generalization of Bounded Tree-Width and Bounded Clique-Width},
  booktitle    = {{LATIN} 2014: Theoretical Informatics - 11th Latin American Symposium,
                  Montevideo, Uruguay, March 31 - April 4, 2014. Proceedings},
  series       = {Lecture Notes in Computer Science},
  volume       = {8392},
  pages        = {72--83},
  publisher    = {Springer},
  year         = {2014},
  url          = {https://doi.org/10.1007/978-3-642-54423-1\_7},
  doi          = {10.1007/978-3-642-54423-1\_7},
  bibsource    = {dblp computer science bibliography, https://dblp.org}
}

@article{DBLP:journals/jgaa/GeffenJKM20,
  author       = {Bas A. M. van Geffen and
                  Bart M. P. Jansen and
                  Arnoud A. W. M. de Kroon and
                  Rolf Morel},
  title        = {Lower Bounds for Dynamic Programming on Planar Graphs of Bounded Cutwidth},
  journal      = {J. Graph Algorithms Appl.},
  volume       = {24},
  number       = {3},
  pages        = {461--482},
  year         = {2020},
  url          = {https://doi.org/10.7155/jgaa.00542},
  doi          = {10.7155/jgaa.00542},
  bibsource    = {dblp computer science bibliography, https://dblp.org}
}

@inproceedings{DBLP:conf/stacs/GroenlandMNS22,
  author       = {Carla Groenland and
                  Isja Mannens and
                  Jesper Nederlof and
                  Krisztina Szil{\'{a}}gyi},
  editor       = {Petra Berenbrink and
                  Benjamin Monmege},
  title        = {Tight Bounds for Counting Colorings and Connected Edge Sets Parameterized
                  by Cutwidth},
  booktitle    = {39th International Symposium on Theoretical Aspects of Computer Science,
                  {STACS} 2022, March 15-18, 2022, Marseille, France (Virtual Conference)},
  series       = {LIPIcs},
  volume       = {219},
  pages        = {36:1--36:20},
  publisher    = {Schloss Dagstuhl - Leibniz-Zentrum f{\"{u}}r Informatik},
  year         = {2022},
  url          = {https://doi.org/10.4230/LIPIcs.STACS.2022.36},
  doi          = {10.4230/LIPIcs.STACS.2022.36},
  bibsource    = {dblp computer science bibliography, https://dblp.org}
}

@article{DBLP:journals/tcs/GurskiW05,
  author       = {Frank Gurski and
                  Egon Wanke},
  title        = {On the relationship between NLC-width and linear NLC-width},
  journal      = {Theor. Comput. Sci.},
  volume       = {347},
  number       = {1-2},
  pages        = {76--89},
  year         = {2005},
  url          = {https://doi.org/10.1016/j.tcs.2005.05.018},
  doi          = {10.1016/j.tcs.2005.05.018},
  bibsource    = {dblp computer science bibliography, https://dblp.org}
}

@inproceedings{DBLP:conf/esa/HegerfeldK23,
  author       = {Falko Hegerfeld and
                  Stefan Kratsch},
  editor       = {Inge Li G{\o}rtz and
                  Martin Farach{-}Colton and
                  Simon J. Puglisi and
                  Grzegorz Herman},
  title        = {Tight Algorithms for Connectivity Problems Parameterized by Clique-Width},
  booktitle    = {31st Annual European Symposium on Algorithms, {ESA} 2023, September
                  4-6, 2023, Amsterdam, The Netherlands},
  series       = {LIPIcs},
  volume       = {274},
  pages        = {59:1--59:19},
  publisher    = {Schloss Dagstuhl - Leibniz-Zentrum f{\"{u}}r Informatik},
  year         = {2023},
  url          = {https://doi.org/10.4230/LIPIcs.ESA.2023.59},
  doi          = {10.4230/LIPIcs.ESA.2023.59},
  bibsource    = {dblp computer science bibliography, https://dblp.org}
}

@inproceedings{DBLP:conf/wg/HegerfeldK23,
  author       = {Falko Hegerfeld and
                  Stefan Kratsch},
  editor       = {Dani{\"{e}}l Paulusma and
                  Bernard Ries},
  title        = {Tight Algorithms for Connectivity Problems Parameterized by Modular-Treewidth},
  booktitle    = {Graph-Theoretic Concepts in Computer Science - 49th International
                  Workshop, {WG} 2023, Fribourg, Switzerland, June 28-30, 2023, Revised
                  Selected Papers},
  series       = {Lecture Notes in Computer Science},
  volume       = {14093},
  pages        = {388--402},
  publisher    = {Springer},
  year         = {2023},
  url          = {https://doi.org/10.1007/978-3-031-43380-1\_28},
  doi          = {10.1007/978-3-031-43380-1\_28},
  bibsource    = {dblp computer science bibliography, https://dblp.org}
}

@article{DBLP:journals/corr/abs-2302-03627/HegerfeldK23,
  author       = {Falko Hegerfeld and
                  Stefan Kratsch},
  title        = {Tight algorithms for connectivity problems parameterized by clique-width},
  journal      = {CoRR},
  volume       = {abs/2302.03627},
  year         = {2023},
  url          = {https://doi.org/10.48550/arXiv.2302.03627},
  doi          = {10.48550/arXiv.2302.03627},
  eprinttype    = {arXiv},
  eprint       = {2302.03627},
  bibsource    = {dblp computer science bibliography, https://dblp.org}
}

@article{DBLP:journals/dam/HeggernesMP12,
  author       = {Pinar Heggernes and
                  Daniel Meister and
                  Charis Papadopoulos},
  title        = {Characterising the linear clique-width of a class of graphs by forbidden
                  induced subgraphs},
  journal      = {Discret. Appl. Math.},
  volume       = {160},
  number       = {6},
  pages        = {888--901},
  year         = {2012},
  url          = {https://doi.org/10.1016/j.dam.2011.03.018},
  doi          = {10.1016/j.dam.2011.03.018},
  bibsource    = {dblp computer science bibliography, https://dblp.org}
}

@article{DBLP:journals/jcss/ImpagliazzoP01,
  author       = {Russell Impagliazzo and
                  Ramamohan Paturi},
  title        = {On the Complexity of k-SAT},
  journal      = {J. Comput. Syst. Sci.},
  volume       = {62},
  number       = {2},
  pages        = {367--375},
  year         = {2001},
  url          = {https://doi.org/10.1006/jcss.2000.1727},
  doi          = {10.1006/jcss.2000.1727},
  bibsource    = {dblp computer science bibliography, https://dblp.org}
}

@article{DBLP:journals/jcss/ImpagliazzoPZ01,
  author       = {Russell Impagliazzo and
                  Ramamohan Paturi and
                  Francis Zane},
  title        = {Which Problems Have Strongly Exponential Complexity?},
  journal      = {J. Comput. Syst. Sci.},
  volume       = {63},
  number       = {4},
  pages        = {512--530},
  year         = {2001},
  url          = {https://doi.org/10.1006/jcss.2001.1774},
  doi          = {10.1006/jcss.2001.1774},
  bibsource    = {dblp computer science bibliography, https://dblp.org}
}

@article{DBLP:journals/tcs/JansenN19,
  author       = {Bart M. P. Jansen and
                  Jesper Nederlof},
  title        = {Computing the chromatic number using graph decompositions via matrix
                  rank},
  journal      = {Theor. Comput. Sci.},
  volume       = {795},
  pages        = {520--539},
  year         = {2019},
  url          = {https://doi.org/10.1016/j.tcs.2019.08.006},
  doi          = {10.1016/j.tcs.2019.08.006},
  bibsource    = {dblp computer science bibliography, https://dblp.org}
}

@article{DBLP:journals/siamdm/Lampis20,
  author       = {Michael Lampis},
  title        = {Finer Tight Bounds for Coloring on Clique-Width},
  journal      = {{SIAM} J. Discret. Math.},
  volume       = {34},
  number       = {3},
  pages        = {1538--1558},
  year         = {2020},
  url          = {https://doi.org/10.1137/19M1280326},
  doi          = {10.1137/19M1280326},
  bibsource    = {dblp computer science bibliography, https://dblp.org}
}

@article{DBLP:journals/talg/LokshtanovMS18,
  author       = {Daniel Lokshtanov and
                  D{\'{a}}niel Marx and
                  Saket Saurabh},
  title        = {Known Algorithms on Graphs of Bounded Treewidth Are Probably Optimal},
  journal      = {{ACM} Trans. Algorithms},
  volume       = {14},
  number       = {2},
  pages        = {13:1--13:30},
  year         = {2018},
  url          = {https://doi.org/10.1145/3170442},
  doi          = {10.1145/3170442},
  bibsource    = {dblp computer science bibliography, https://dblp.org}
}

@article{DBLP:journals/combinatorica/MulmuleyVV87,
  author       = {Ketan Mulmuley and
                  Umesh V. Vazirani and
                  Vijay V. Vazirani},
  title        = {Matching is as easy as matrix inversion},
  journal      = {Comb.},
  volume       = {7},
  number       = {1},
  pages        = {105--113},
  year         = {1987},
  url          = {https://doi.org/10.1007/BF02579206},
  doi          = {10.1007/BF02579206},
  bibsource    = {dblp computer science bibliography, https://dblp.org}
}

@inproceedings{DBLP:conf/birthday/Rooij20,
  author       = {Johan M. M. van Rooij},
  editor       = {Fedor V. Fomin and
                  Stefan Kratsch and
                  Erik Jan van Leeuwen},
  title        = {Fast Algorithms for Join Operations on Tree Decompositions},
  booktitle    = {Treewidth, Kernels, and Algorithms - Essays Dedicated to Hans L. Bodlaender
                  on the Occasion of His 60th Birthday},
  series       = {Lecture Notes in Computer Science},
  volume       = {12160},
  pages        = {262--297},
  publisher    = {Springer},
  year         = {2020},
  url          = {https://doi.org/10.1007/978-3-030-42071-0\_18},
  doi          = {10.1007/978-3-030-42071-0\_18},
  bibsource    = {dblp computer science bibliography, https://dblp.org}
}

\end{document}